\documentclass[12pt,english,reqno]{article}

\usepackage{graphicx}
\usepackage{amsmath,amsfonts,amssymb,amstext,amsthm}
\usepackage{appendix}
\usepackage{babel}
\usepackage{bm,bbm}
\usepackage{float}
\usepackage{url}
\usepackage{enumitem}

\usepackage[all]{nowidow}

\usepackage[hyperfootnotes=false]{hyperref}
\hypersetup{
    hidelinks,
    linkbordercolor = {1 1 1},
    citebordercolor = {1 1 1},
    urlbordercolor = {1 1 1} 
}

\newtheorem{theorem}{Theorem}
\newtheorem{lemma}{Lemma}

\newtheorem{corollary}{Corollary}

\newtheorem{assumption}{Assumption}

\usepackage{lmodern}
\usepackage{amsmath}
\usepackage{amssymb}
\usepackage{amsfonts}
\usepackage{mathtools}
\usepackage{bm}
\usepackage{bbm}
\newcommand\independent{\protect\mathpalette{\protect\independenT}{\perp}}
\def\independenT#1#2{\mathrel{\rlap{$#1#2$}\mkern2mu{#1#2}}}
\usepackage{relsize}

\usepackage[margin=1in]{geometry}
\usepackage[nodisplayskipstretch]{setspace}
\usepackage{footmisc}

\usepackage{booktabs,caption,subcaption}

\usepackage{threeparttable}
\usepackage{rotating}

\usepackage{algorithm}
\usepackage[noend]{algpseudocode}

\usepackage{natbib,apalike}

\DeclareMathOperator*{\argmin}{arg\,min}

\newcommand{\E}{\mathrm{E}}

\newcommand{\Cov}{\mathrm{Cov}}
\newcommand{\Var}{\mathrm{Var}}
\newcommand{\En}{\mathbbm{E}_n}
\newcommand{\Gn}{\mathbbm{G}_n}
\newcommand{\inv}{^{\text{-}1}}
\def\independenT#1#2{\mathrel{\rlap{$#1#2$}\mkern2mu{#1#2}}}
\usepackage{relsize}

\newif\ifpublic
\publictrue
\newcommand{\inPublicVersion}[1]{\ifpublic #1\fi}

\newif\ifanonymous
\anonymousfalse
\newcommand{\inAnonymousVersion}[1]{\ifanonymous #1\fi}

\begin{document}
\title{Optimal Categorical Instrumental Variables\inPublicVersion{\thanks{I thank St\'ephane Bonhomme, Bruce Hansen, Christian Hansen, Phillip Heiler, Samuel Higbee, Thibaut Lamadon, Lihua Lei, Jonas Lieber, Elena Manresa, Sendhil Mullainathan, Whitney Newey, Kirill Ponomarev, Guillaume Pouliot, Vitor Possebom, Max Tabord-Meehan, and Alexander Torgovitsky for valuable comments and suggestions, along with participants at the University of Chicago Econometrics advising group, the International Association for Applied Econometrics Conference 2023, and the North American Winter Meetings of the Econometric Society 2024. All remaining errors are my own.}}}

\inPublicVersion{\author{Thomas Wiemann\thanks{University of Chicago, \protect\url{wiemann@uchicago.edu}.}}}

\date{\today}
\maketitle

\vspace{-2.5em}

\begin{abstract}
\noindent This paper discusses estimation with a categorical instrumental variable in settings with potentially few observations per category. The proposed categorical instrumental variable estimator (CIV) leverages a regularization assumption that implies existence of a latent categorical variable with fixed finite support achieving the same first stage fit as the observed instrument. In asymptotic regimes that allow the number of observations per category to grow at arbitrary small polynomial rate with the sample size, I show that when the cardinality of the support of the optimal instrument is known, CIV is root-$n$ asymptotically normal, achieves the same asymptotic variance as the oracle IV estimator that presumes knowledge of the optimal instrument, and is semiparametrically efficient under homoskedasticity. Under-specifying the number of support points reduces efficiency but maintains asymptotic normality. In an application that leverages judge fixed effects as instruments, CIV compares favorably to commonly used jackknife-based instrumental variable estimators.   
\end{abstract}

\pagenumbering{gobble}
\newpage
\pagenumbering{arabic}

\section{Introduction}\label{sec:Introduction}

Optimal instrumental variable estimators aim to improve statistical precision by maximizing the strength of the first stage. In the absence of functional form assumptions, optimal instruments need to be nonparametrically estimated which can introduce bias from over-fitting. In response to these challenges, a growing literature considers complexity-reducing assumptions on the first stage reduced form that -- when leveraged appropriately -- allow for second stage estimators with the same asymptotic variance as the infeasible oracle estimator that presumes knowledge of the optimal instrument. Increasingly popular in practice is the post-lasso IV estimator of \citet{belloni2012sparse} that assumes approximate sparsity of the first stage reduced form \citep[see, e.g.,][]{gilchrist2016something, dhar2022reshaping}.\footnote{Informally, approximate sparsity presumes that a slowly increasing unknown subset of instruments suffices to approximate the optimal instrument relative to the reduced form estimation error.} However, while approximate sparsity may be a well-suited assumption for economic settings with continuous instruments, it is often ill-suited when instruments are categorical. Simulations with categorical instruments in \citet{angrist2020machine} and in this paper highlight that even in settings with many more observations than instruments, lasso-based IV estimators can have substantially worse finite sample behavior than even two-stage least squares (TSLS).\footnote{Recently, \citet{kolesar2023fragility} formally characterize disadvantages of applying sparsity assumptions to settings with categorical variables. While their discussion focuses on settings with categorical controls, analogous arguments apply to sparsity assumptions for categorical instruments.} 

This paper proposes a new optimal instrumental variable estimator for settings with a large number of categorical instruments. To approximate the practical settings in which the number of observations per category is small, I characterize the proposed categorical instrumental variable estimator (CIV) in asymptotic regimes that allow the expected number of observations per category to grow at arbitrarily small polynomial rate with the sample size. To obtain root-$n$ normality in these challenging ``moderately many instruments'' settings, I consider a regularization assumption designed specifically for categorical instruments: Fixed finite support of the optimal instrument. When the cardinality of the support of the optimal instrument is known, I show that CIV achieves the same asymptotic variance as the infeasible oracle two-stage least squares estimator that presumes knowledge of the optimal instrument, and is semiparametrically efficient under homoskedasticity. Further, under-specifying the number of support points maintains asymptotic normality but results in efficiency loss.\footnote{An implementation of CIV is provided in the \textsf{R} package \inPublicVersion{\href{https://www.thomaswiemann.com/civ}{\texttt{civ}}}\inAnonymousVersion{\texttt{civ}} available on CRAN.}

The key idea of the categorical instrumental variable estimator is to leverage a latent categorical variable with fewer categories that achieves the same population-level fit in the first stage. Under the assumption that the support of the latent categorical variable is fixed with finite cardinality, it is possible to estimate a mapping from the observed categories to the latent categories. This estimated mapping can then be used to simplify the optimal instrumental variable estimator to a finite dimensional regression problem. Asymptotic properties of the CIV estimator then follow if the first-stage mapping can be estimated at a sufficiently fast rate. I provide sufficient conditions for estimation of the mapping at exponential rate using a $K$-Conditional-Means ($K$CMeans) estimator. The proposed $K$CMeans estimator is exact and computes very quickly with time polynomial in the number of observed categories, thus avoiding heuristic solution approaches otherwise associated with $K$Means-type problems.\footnote{$K$CMeans also has applications outside of instrumental variable settings. For example, $K$CMeans can easily be combined with the double/debiased machine learning framework of \citet{chernozhukov2018double}. An implementation of $K$CMeans is provided in the \textsf{R} package \inPublicVersion{\href{https://www.thomaswiemann.com/kcmeans}{\texttt{kcmeans}}}\inAnonymousVersion{\texttt{kcmeans}} \href{https://www.thomaswiemann.com/kcmeans} available on CRAN.}

The focus on categorical instrumental variables is motivated by the many examples in empirical economics, including leading examples in the many and weak instruments literature, featuring categorical instruments. In their analysis of returns to education, for example, \citet{angrist1991does} consider interactions between quarter of birth indicators and year and place of birth indicators, resulting in an instrument with 180 categories. Extensions of their design consider the fully saturated first stage of all interactions between the three sets of indicators resulting in 1530 categories, each representing a unique combination of quarter, year, and place of birth \citep[see, e.g.,][]{mikusheva2020inference, angrist2020machine}. More recently, a large empirical literature uses so-called examiner fixed effects as instruments \citep[see, e.g.,][]{kling2006incarceration,maestas2013does,aizer2015juvenile,dobbie2018effects,bhuller2020incarceration,agan2023misdemeanor}. For example, \citet{dobbie2018effects} use judge identities as instruments for pre-trial detention to analyze its effect on the probability of conviction.\footnote{Although judge identity is the most widely adopted examiner fixed effect instrument, other quasi-randomly assigned decision makers are also frequently considered. \citet{agan2023misdemeanor}, for example, leverage random assignment of assistant district attorneys to non-violent misdemeanor cases. See \citet{chyn2024examiner} for a comprehensive overview of empirical analyses using examiner fixed effect instruments.} 

The empirical literature on examiner fixed effects also motivates the analysis of the proposed CIV estimator under ``moderately many instruments'' asymptotics. In these economic analyses using examiner fixed effects as instruments, it is common practice to estimate first-stage fitted values via a leave-one-out procedure -- i.e., to compute the fitted value of a particular decision without directly using the decision itself. These jackknifing procedures are often motivated as a solution to a many instruments problem arising due to only few observations per decision maker \citep{chyn2024examiner}. However, unlike the econometric literature on many instruments that often focuses on asymptotic regimes in which the number of observations per instrument \textit{is fixed} \citep[e.g., ][]{bekker1994alternative}, practitioners appear to primarily consider settings that are approximated by regimes in which the number of observations per instruments \textit{grows slowly}. This follows from the combined use of jackknifed fitted values and conventional (heteroskedasticity-robust or clustered) standard errors, the latter resulting in invalid inferential statements under the traditional ``many instruments'' asymptotic regime but leading to correct inference in regimes with a growing number of observations per instrument (even if these are growing very slowly) \citep{chao2012asymptotic}.

In the moderately many instrument setting considered in this paper, traditional jackknife IV estimators (JIVE) \citep{phillips1977bias,angrist1999jackknife} are first-order equivalent to CIV. Indeed, the results of \citet{chao2012asymptotic} imply semiparametric efficiency of JIVE under weaker assumptions than I leverage for CIV as the former does not rely on a finite support restriction. Yet, as illustrated in the simulations and the application, JIVE does not uniformly outperform CIV: JIVE (and its IJIVE \citep{ackerberg2009improved} and UJIVE \citep{kolesar2013estimation} variants) can suffer from high instability and unnecessarily large standard errors -- a well-documented practical caveat of jackknife estimators that has been a topic of substantial discussion in the econometrics literature \citep[e.g.,][]{davidson2006case,ackerberg2006comment}. In contrast, CIV is nearly as stable as TSLS while exhibiting substantially less bias and achieving better coverage.

The key difference between jackknife-based estimators with categorical instruments and CIV is how the bias from a large number of instruments is controlled. For a heuristic illustration, it is useful to consider the expansion  
\begin{align}\label{eq:illustration_expansion}
 \frac{1}{\sqrt{n}}\sum_{i=1}^n\hat{m}_n(Z_i)U_i = \frac{1}{\sqrt{n}}\sum_{i=1}^n m_0(Z_i)U_i + \frac{1}{\sqrt{n}}\sum_{i=1}^n(\hat{m}_n(Z_i) - m_0(Z_i))U_i,
\end{align}
which plays a key role in the convergence of JIVE and CIV estimators. Here, $Z_i$ is the (categorical) instrument of the $i$th observation, $U_i$ is a structural residual, and $m_0(Z_i)$ and $\hat{m}_n(Z_i)$ are the true and estimated optimal instrument, respectively. The first term on the right hand side in \eqref{eq:illustration_expansion} is well-behaved under standard regularity conditions and converges in distribution to a Normal random variable. The second term can potentially diverge due to correlation between the residuals $U_i$ and the estimation error $\hat{m}_n(Z_i) - m_0(Z_i)$. For example, \citet{donald2001choosing} highlight that when $\hat{m}_n$ is a first-stage linear regression estimator (as in TSLS), divergence can occur if the number of instruments grows faster than $\sqrt{n}$. JIVE resolves this problem of potential divergence by enforcing independence between the residuals and estimation errors for any particular observation via leave-one-out estimation. Under this independence, it then suffices for $\hat{m}_n$ to be consistent for $m_0$.\footnote{Note that leave-one-out estimation of the first stage does not immediately imply independence of the estimation error and the residuals if observations are correlated, a point previously highlighted in \citet{frandsen2023cluster}. \citet{frandsen2023cluster} propose a jackknife IV estimator for examiner fixed effect designs in which observations are clustered at a level finer than that of the examiner (e.g., at the shift level). However, if observations are arbitrarily correlated at the level of the examiner -- as suggested by current empirical literature that often clusters at the examiner level -- no applicable jackknife-based IV estimator exists. The regularization approach presented in this paper may be an appealing solution to the problem of correlated examiner decisions since the key concentration inequalities leveraged in the proof extend to weakly dependent data. Formally investigating the robustness of CIV to weakly dependent judge decisions may thus be a potentially fruitful topic of future research.} An alternative solution is to ensure that the $\hat{m}_n$ is not just consistent but to also control the estimation errors via regularization. Sufficiently tight control of estimation errors then allows the second term  on the right hand side in \eqref{eq:illustration_expansion} to vanish without relying on leave-one-out estimation (or alternative sample splitting procedures). For example, \citet{belloni2012sparse} show that under approximate sparsity, a lasso-based estimator of the optimal instrument is sufficiently well controlled to allow for root-$n$ normal inference. Similarly, in this paper, I show that CIV provides sufficiently tight control on the estimation errors for the second term in \eqref{eq:illustration_expansion} to vanish without the need to explicitly enforce independence via jackknifing.\footnote{Sample splitting and regularized estimation of the optimal instrument are not exclusive and can often be used jointly to weaken conditions for the control of estimation errors. For example, \citet{belloni2012sparse} show that sample splitting allows for a denser representation of the optimal instrument. \citet{hk2014jive} combine JIVE and ridge regularization to admit estimation with more instruments than observations in non-sparse regimes.} 

The advantage of the proposed CIV estimator over alternative optimal IV estimators is that the regularization assumption employed for tight control of the estimation errors places restrictions on the data generating process that have straightforward economic interpretations when applied to categorical instruments such as examiner fixed effects. In particular, the regularization assumption I consider presumes existence of an unobserved combination of categories (or: examiners) that fully captures relevant variation from the instruments. In the application of \citet{dobbie2018effects}, for example, this regularization assumption presumes that the leniency of judges is a discrete (unobserved) type. For the restriction of the optimal instrument to two support points, this would correspond to categorizing judges as either being ``strict'' or ``lenient.'' While semiparametric efficiency of CIV only follows if this regularization assumption holds exactly, I show that CIV remains root-$n$ normal even if judges are categorized into fewer ``pseudo'' types than their true types (e.g., if there are also ``moderate'' judges). CIV thus provides a solution to a ``moderately many categorical instruments'' problem under conditions that economists can readily understand and discuss.

\textit{Related Literature.} The paper primarily draws from and contributes to three strands of literature. First, the literature on many instruments that develops estimators robust to asymptotic regimes in which the number of instruments is proportional to the sample size \citep[e.g.,][]{phillips1977bias, bekker1994alternative, angrist1995split, angrist1999jackknife,chao2005consistent, hansen2008estimation, chao2012asymptotic, hausman2012instrumental, kolesar2013estimation}. Most closely related is \citet{bekker2005instrumental} who provide limiting distributions of two-stage least squares (TSLS), limited information maximum likelihood (LIML), and heteroskedasticity-adjusted estimators under group asymptotics that consider replications of categorical instruments with a constant number of observations per category. In the less stringent asymptotic regimes I consider in this paper where the number of categories grows at a slower rate than the sample size, their results imply first-order equivalence of the LIML and the oracle IV estimator in the presence of heteroskedasticity when observations are equally distributed across categories and effects are constant.\footnote{Note also that Lemma 6.A of \citet{donald2001choosing} implies that LIML using categorical instruments achieves first-order oracle equivalence when the number of categories grows below the sample rate and causal effects are constant.} Despite the favorable statistical properties of LIML in settings with categorical instrumental variables and homogeneous effects, its application to causal effects estimation in economics is limited by its strong reliance on constant effects in the linear IV model. \citet{kolesar2013estimation} shows that under the nonparametric causal model of \citet{imbens1994identification}, the LIML estimand cannot generally be interpreted as a positively (weighted) average of causal effects. In the terminology of \citet{blandhol2022tsls}, LIML thus does not generally admit a weakly causal interpretation. In contrast, the proposed CIV estimator falls in the class of two-step estimators of \citet{kolesar2013estimation} and therefore admits a weakly causal interpretation in the presences of unobserved heterogeneity. 

Second, I draw from the literature on optimal instrumental variable estimators. Optimal instruments are conditional expectations that -- in the absence of functional form assumptions -- can be nonparametrically estimated \citep[e.g.,][]{amemiya1974multivariate,chamberlain1987asymptotic,newey1990efficient}. \citet{newey1990efficient} considers approximation of optimal instruments using polynomial sieve regression and characterizes the growth rate of series terms relative to the sample size that allows for root-$n$ consistency. In the setting of categorical variables, the restrictions imply that the number of categories should grow slower than root-$n$ to avoid the many instruments bias. CIV contributes to the literature on optimal IV estimation that leverages regularization assumptions on the first stage to allow for a larger number of considered instruments. In homoskedastic linear IV models, \citet{donald2001choosing} propose instrument selection criteria,  \citet{chamberlain2004random} consider regularization via a random coefficient assumption, and \citet{okui2011instrumental} suggests first stage estimation via Ridge regression ($\ell_2$ regularization). In linear IV models with heteroskedasticity, \citet{carrasco2012regularization} consider $\ell_2$ regularization (including Tikhonov regularization) and provide conditions for asymptotic efficiency of the resulting IV estimator in settings when the number of instruments is allowed to grow at faster rate than the sample size. \citet{belloni2012sparse} apply the lasso and post-lasso ($\ell_1$ regularization) to estimate optimal instruments in the setting with very many instruments. The authors provide sufficient conditions for the asymptotic efficiency of the resulting IV estimator, most notably, an approximate sparsity assumption, which presumes that a slowly increasing unknown subset of instruments suffices to approximate the optimal instrument relative to the reduced form estimation error. A common theme in the regularization approaches of these previous approaches is shrinkage of the first stage coefficients to zero. In the setting of categorical instruments, this corresponds to existence of one large latent base category (i.e., the constant) and only a few small deviating latent categories. Settings in which differing latent categories are approximately proportional are not admitted in these shrinkage-to-zero approaches as observed categories cannot be arbitrarily merged. CIV complements these existing optimal instrument estimators by leveraging an alternative regularization assumption that admits approximately proportional latent categories via arbitrary combination of observed categories.

Finally, I draw from the literature on estimation with finite support restrictions in longitudinal data settings including \citet{hahn2010panel}, \citet{bonhomme2015grouped}, \citet{bester2016grouped} and \citet{su2016identifying}. \citet{hahn2010panel} show that finite support assumptions substantially decrease the incidental parameter problem associated with increasingly many fixed effects. \citet{bester2016grouped} consider grouped fixed effects when the grouping is known. \citet{bonhomme2015grouped} and \citet{su2016identifying}, among others, consider settings with unknown groups and parameters. I adapt the $K$Means fixed effects estimator of \citet{bonhomme2015grouped} for estimation of the optimal categorical instrumental variable. In doing so, I make two contributions to the theoretical analysis of $K$Means. First, I construct and characterize a $K$-\textit{Conditional}-Means estimator suitable for cross-sectional regression. Leveraging arguments from \citet{fisher1958grouping} and \citet{wang2011ckmeans}, this estimator has the advantage of achieving global in-sample optimality in time polynomial in the number of observed categories, by-passing the NP-hard problem of $K$Means in multiple dimensions. I thus do not need to abstract away from optimization error as is usually necessary in applications of $K$Means.\footnote{Computational solutions to $K$Means applications in longitudinal data settings are an active literature. See, in particular, the ongoing work of \citet{chetverikov2022spectral} and \citet{mugnier2022simple}.} Second, I show that the $K$CMeans estimand can serve as an approximation to a conditional expectation function when the number of allowed-for support points is under-specified. In the instrumental variable setting considered in this paper, this property is leveraged to allow for root-$n$ consistent estimation given only a lower-bound on the support points of the optimal instrument. 

\textit{Outline.} The remainder of the paper is organized as follows: Section \ref{sec:estimator} introduces the CIV estimator in the simple setting without additional covariates and discusses its motivating assumptions. Section \ref{sec:theory} states the CIV estimator with covariates and presents the main theoretical result of the paper. Section \ref{sec:simulation} provides a simulation exercise to contrast the finite sample performance of CIV with competing estimators. Section \ref{sec:empirical_example_judges} revisits the application of \citet{dobbie2018effects} and \citet{chyn2024examiner} where bail judge identities are used as instruments for pre-trial release of defendants. Section \ref{sec:conclusion} concludes. 

\textit{Notation.} It is useful to clarify some notation. In the following sections, I characterize the law $P_n$ of the random vector $(Y, D, X^\top, Z^{(n)}, Z^{(0)}, U)$ associated with a single observation. Here, $Y$ denotes the outcome, $D$ is the scalar-valued endogenous variable of interest, $X$ is a $J$ dimensional vector of control variables, $Z^{(n)}$ is the observed instrumental variable, $Z^{(0)}$ is a latent instrumental variable, and $U$ are all other determinants of $Y$ other than $(D, X^\top)$. It is often convenient to use $W\equiv (D, X^\top)^\top$.  Throughout, I maintain focus on a scalar-valued $D$ for ease of exposition but highlight that any fixed number of endogenous variables can easily be accommodated under the presented framework. The joint law $P_n$ is allowed to change with the sample size $n\in \mathbbm{N}$ to approximate settings with relatively few observations per observed category. To permit the study of semiparametric efficiency, however, I keep the marginal law of $(Y, D, X^\top, Z^{(0)}, U)$ fixed. Explicit references to $P_n$ are largely omitted for brevity. Further, for a random vector $S$, let $\mathcal{S}$ denote its support and $\vert \mathcal{S}\vert$ the cardinality of the support. For an i.i.d.\ sample $\{S_i\}_{i=1}^n$ from $S$, define the operators $\En S \equiv \frac{1}{n}\sum_{i=1}^n S_i$ and $\Gn S \equiv \frac{1}{\sqrt{n}}\sum_{i=1}^n (S_i - \E S)$. For a measurable set $\mathcal{A}$, the indicator function $\mathbbm{1}_{\mathcal{A}}(S)$ is equal to one if $S\in \mathcal{A}$ and zero otherwise. For a function $f:\mathcal{A}\to \mathbbm{R}$, let $f(\mathcal{A})$ denote its image. The $\ell_2$-norm is denoted by $\|\cdot\|$ where for a matrix $M$ the norm is $\|M\|=\operatorname{tr}(M^\top M)^{1/2}$.

\section{The Categorical Instrumental Variable Estimator} \label{sec:estimator}

This section introduces the categorical instrumental variable estimator and discusses its key motivating assumptions. I focus on the setting without control variables in this section, leaving the general specification for Section \ref{sec:theory} that provides the formal asymptotic analysis.

The instrumental variable model is \begin{align*}
    Y = D\tau_0 + U, \quad \E[U\vert Z] \overset{a.s.}{=} 0,
\end{align*}
where $\tau_0$ is the parameter of interest. In the linear model with homogeneous effects considered here, the coefficient corresponds to the change in the outcome caused by a marginal change in endogenous variable of interest.

The mean-independence of $U$ and $Z$ implies the moment condition $ \E\big(Y - D\tau_0\big)\big(m(Z) - \E m(Z)\big) = 0$ for any measurable function $m:\mathcal{Z}\to \mathbbm{R}$. If in addition $\Cov(D, m(Z)) \neq 0$, a solution to the moment condition is given by 
\begin{align}\label{eq:alpha0}
    \tau_0 = \frac{\Cov(Y, m(Z))}{\Cov(D, m(Z))}.
\end{align}
Replacing the covariances with their sample analogues, the instrumental variable model thus suggests potentially infinitely many estimators for $\tau_0$ as indexed by the function $m$, each consistent and root-$n$ asymptotically normal under regularity assumptions. To decide between these alternative estimators, econometricians have turned to study their efficiency. 

Suppose the econometrician observes an i.i.d.\ sample $\{(Y_i, D_i, Z_i)\}_{i=1}^n$ from $(Y, D, Z)$ and $U$ is homoskedastic -- i.e., $\E [U^2\vert Z]\overset{a.s.}{=}\sigma^2$. If $f$ is chosen to be the conditional expectation $m_0(z)\equiv \E[D\vert Z=z]$, then the asymptotic variance of the corresponding oracle estimator \begin{align}\label{eq:infeasible_tau}
     \hat{\tau}^* = \frac{\En \big(Y - \En Y\big)\big(m_0(Z) - \En m_0(Z)\big)}{\En\big(D - \En D\big)\big(m_0(Z) - \En m_0(Z)\big)},
\end{align}
achieves the semiparametric efficiency bound $\sigma^2/\Var(m_0(Z))$ (see, e.g., \citealp{chamberlain1987asymptotic}). The transformed instruments $m_0(Z)$ are thus often termed ``optimal instruments.'' 

Formulating estimators based on the moment solution \eqref{eq:alpha0} has additional benefit of falling in the class of ``two-step'' estimators as defined by \citet{kolesar2013estimation}. The author shows that even if the underlying structural model is not additively separable in the structural error $U$ as presumed here, two-step estimators admit interpretation as a convex combination of causal effects under the LATE assumptions of \citet{imbens1994identification}.\footnote{In addition to stronger exogeneity assumptions, the LATE assumptions include a monotonicity assumption that prohibits simultaneous movements in-and-out of treatment for any increment of the optimal instrument.} This starkly contrast the LIML estimator and its variants, which do not generally permit a weakly causal interpretation in the LATE framework \citep{kolesar2013estimation}. Estimation based on \eqref{eq:alpha0} and the optimal instrument $m_0$ thus has both important economic and statistical benefits.  

In economic applications, the conditional expectation $m_0$ is rarely known. The oracle estimator $\hat{\tau}^*$ is thus typically infeasible in practice. A growing literature focuses on estimating the optimal instruments such that the asymptotic distribution of the resulting estimator for $\tau_0$ achieves the same asymptotic variance as the infeasible estimator. For example, \citet{newey1990efficient} considers nearest-neighbor and series regression to approximate $m_0$. In settings with growing numbers of instruments, \citet{belloni2012sparse} and \citet{carrasco2012regularization} consider regularized regression estimators. 

This paper is concerned with estimation of optimal instruments in settings where the observed instrument is categorical. To provide a better asymptotic approximation to settings with relatively few observations per category, I allow the number of categories to grow with the sample size. Letting $Z^{(n)}$ denote the observed instrument to highlight this dependence on the sample size index $n \in \mathbbm{N}$, Assumption \ref{assumption:rate_condition} formally specifies the rate at which the number of categories is allowed to grow. In particular, the number of categories can increase such that the expected number of observations per category (i.e., $n\times \Pr(Z^{(n)}=z)$) grows at arbitrarily slow polynomial rate with the sample size.

\begin{assumption}\label{assumption:rate_condition}
$\forall z \in \bigcup \mathcal{Z}^{(n)}, \exists \lambda_z \in (0, 1],  a_z > 0,$ such that $(P_n)_{n\in\mathbbm{N}}$ satisfies \begin{align*}
    \Pr(Z^{(n)} = z) n^{1 - \lambda_z} \to a_z.
\end{align*}
\end{assumption}

If all $\lambda_z\in(0.5, 1]$ the number of categories grows sufficiently slowly such that the optimal instrument can be estimated at a sufficiently fast rate by simple least squares of $D$ on the (increasing) set of indicators $(\mathbbm{1}_z(Z^{(n)}))_{z\in\mathcal{Z}^{(n)}}$. The more interesting settings are thus if for some categories $\lambda_z\in(0, 0.5]$. Since $\lambda_z$ can be arbitrarily close to $0$, this regime can be viewed as approximating settings in which the number of observations per category is small or moderate. These settings seem of particular practical importance in economic applications. Indeed, as discussed in the introduction, the joint use of jackknife-based instrumental variable estimators and conventional (heteroskedasticity-robust or clustered) standard errors in the examiner fixed effects literature (see, e.g., \citealp{chyn2024examiner}) suggests these ``moderately many instruments'' regimes are the leading asymptotic approximation in current applied research that leverages judge identities as instruments.

To accommodate efficient estimation in these more challenging settings with growing number of categories, I make a complexity-reducing assumption on the structure of the optimal instrument. Assumption \ref{assumption:Z0_setup} asserts that the optimal instrument, denoted $Z^{(0)}$, has finite support of cardinality $K_0\in \mathbbm{N}$. Unlike the observed instrument $Z^{(n)}$, the cardinality of the support of optimal instrument is thus fixed as the sample size increases.\footnote{The application of a finite support restriction as in Assumption \ref{assumption:Z0_setup} along with a rate condition as in Assumption \ref{assumption:rate_condition} to cross-sectional regression on categorical variables appears novel, however, finite support assumptions have grown increasingly popular in longitudinal data settings where the categories are individual identifiers (see, in particular, \citealp{bonhomme2015grouped}). In these longitudinal data settings, Assumption \ref{assumption:Z0_setup} corresponds to a group-fixed effects assumption and Assumption \ref{assumption:rate_condition} regulates the rates at which the cross-section and the time dimension grow.}

\begin{assumption}\label{assumption:Z0_setup}
$ $

\vspace{-0.5\baselineskip}

\begin{enumerate}[start=1,label={(\alph*)}]\setlength\itemsep{0.25em}
   \item $\exists K_0 \in \mathbbm{N}$ such that $\vert \mathcal{Z}^{(0)}\vert = K_0$.
\item $\forall n \in \mathbbm{N},$ $P_n$ is such that $\exists m_0^{(n)}: \mathcal{Z}^{(n)}\to  \mathcal{Z}^{(0)},\: m_0^{(n)}(Z^{(n)}) \overset{a.s.}{=} Z^{(0)}.$
\end{enumerate}
\end{assumption}

Assumption \ref{assumption:Z0_setup} implies that all relevant information on the endogenous variable included in $Z^{(n)}$ is also captured by the latent instrument $Z^{(0)}$. That is, there exists a deterministic map from values of $Z^{(n)}$ to values of the optimal instrument $Z^{(0)}$ -- or -- a partition $(\mathcal{Z}^{(n,0)}_k)_{k=1}^{K_0}$ of $\mathcal{Z}^{(n)}$ such that the optimal instrument is constant within each $\mathcal{Z}^{(n,0)}_k$: $m_0^{(n)}(z) = m_0^{(n)}(z'), \forall z, z' \in \mathcal{Z}^{(n,0)}_k$. When this map is known, efficient estimation of $\tau_0$ simplifies to TSLS of $Y$ on $D$ using the (non-increasing) set of indicators $(\mathbbm{1}_{\mathcal{Z}^{(n,0)}_k}(Z^{(n)}))_{k \in \{1, \ldots, K_0\}}$ as instruments. In practice, the map is unknown and needs to be estimated. 

To estimate the optimal instrument, I propose a $K$-Conditional-Means ($K$CMeans) estimator given by\begin{align}\label{eq:definition_m_hat_simple}
    \hat{m}_K^{(n)}\equiv \argmin_{\substack{m: \mathcal{Z}^{(n)}\to \mathcal{M} \\ \vert m(\mathcal{Z}^{(n)})\vert \leq K}} \: \En(D - m(Z^{(n)}))^2, 
\end{align}
where $\mathcal{M}\subset \mathbbm{R}$ is compact and $K\in\mathbbm{R}$ is the number of support points allowed-for by the researcher. In contrast to (unconditional) $K$Means which clusters observations into $K$ groups, $K$CMeans creates a partition of the support of the categorical variable $Z^{(n)}$ allowing its application to reduced form regression problems as required here. In practice, $K$CMeans is solved via a dynamic programming algorithm adapted from the algorithm for $K$Means discussed in \citet{wang2011ckmeans}. An important feature of this approach is that $K$CMeans can be solved to a global minimum in time polynomial in $\vert \mathcal{Z}^{(n)}\vert$, avoiding the heuristic solution approaches to $K$Means problems in multiple dimensions. I thus do not need to abstract away from optimization error as is common in other applications of $K$Means estimators (see, e.g., \citealp{bonhomme2015grouped}).

The categorical instrumental variable (CIV) estimator is then simply given by \begin{align}
    \hat{\tau} = \frac{\En\left(Y - \En Y\right)\left(\hat{m}_{K}^{(n)}(Z) - \En \hat{m}_{K}^{(n)}(Z)\right)}{\En \left(D - \En D\right)\left(\hat{m}_{K}^{(n)}(Z) - \En \hat{m}_{K}^{(n)}(Z)\right)}.
\end{align}

When $K_0$ is known and $\hat{m}_{K_0}^{(n)}$ is the $K$CMeans estimator \eqref{eq:definition_m_hat_simple} with $K = K_0$, CIV is a feasible analogue to the infeasible oracle estimator $\hat{\tau}^*$ in \eqref{eq:infeasible_tau}. As discussed in detail in Section \ref{sec:theory} and given the therein stated assumptions, $\hat{\tau}$ has the same asymptotic distribution as the infeasible estimator and is semiparametrically efficient under homoskedasticity. 

The result of semiparametric efficiency of CIV depends on the correct choice of $K=K_0$. In applications, economic insight can occasionally provide concrete suggestions for a value of $K_0$. For example, Appendix \ref{app:ak_app} provides a rationale for why $K_0=2$ in the application of \citet{angrist1991does}. In other applications, knowledge of $K_0$ is less certain. In judge fixed effects applications, for example, a researcher may consider the classification of judges as ``lenient'' and ``strict'' as only an approximation to the true latent types of judges. In these settings where a concrete suggestions for $K_0$ is not available, the CIV estimator that uses $\hat{m}_{K}^{(n)}$ as an approximate optimal instrument maintains root-$n$ normality for any $K\in\{2,\ldots, K_0\}$ provided that the approximate optimal instrument remains relevant. In particular, for $K \in \{2, \ldots, K_0\}$, $\hat{m}_K^{(n)}$ estimates the approximate optimal instrument \begin{align}\label{eq:definition_mK}
    m_K^{(n)} \equiv \argmin_{\substack{m: \mathcal{Z}^{(n)}\to \mathcal{M} \\ \vert m(\mathcal{Z}^{(n)})\vert \leq K}} \: \E(Z^{(0)} - m(Z^{(n)}) )^2.
\end{align}
As stated in Theorem \ref{theorem:CIV_w_covariates_mk}, when $m_K^{(n)}(Z^{(n)})$ satisfies the usual instrumental variable rank condition, the corresponding CIV estimator remains asymptotically normal even when $K<K_0$. Albeit at a loss of statistical efficiency, this implies that knowledge of a lower-bound on $K_0$ is often sufficient for inference (e.g., choosing $K=2$).\footnote{Importantly, under-specification of $K$ has no implication on the interpretation of the CIV estimand as a convex combination of causal effects. As illustrated in the proof of Lemma \ref{lemma:mK_characteristics}, the $K$CMeans solution is contiguous in the sense of \citet{fisher1958grouping} -- i.e., only adjacent support points of $Z^{(0)}$ are combined when constructing the approximate optimal instrument. As a consequence, if a monotonicity assumption applies to $Z^{(0)}=m^{(n)}_0(Z^{(n)})$ to warrant interpretation of the TSLS as a convex combination of causal effects as in \citet{angrist1995two}, then a monotonicity assumption will also apply to approximate optimal instrument $m^{(n)}_K(Z^{(n)})$. Analogous arguments apply in a setting with covariates, albeit under the qualification of saturated controls which applies broadly to two-step estimators. See, for example, \citet{blandhol2022tsls} for a discussion on the importance of saturated controls for the interpretation of TSLS estimands.}

\section{Asymptotic Theory}\label{sec:theory}

This section provides a formal discussion of the CIV estimator. After defining the $K$CMeans and CIV estimators in the presence of additional exogenous variables of fixed dimension, I state the main result of the paper in Theorem \ref{theorem:CIV_w_covariates_mk}. Corollary \ref{corollary:semiparametric_efficiency} provides the statement of semiparametric efficiency for known $K_0$. 

Assumption \ref{assumption:iv_setup} defines the instrumental variable model with $Z^{(0)}$ being the unobserved optimal instrument previously characterized by Assumption \ref{assumption:Z0_setup}. The parameter of interest is the vector $\theta_0 \equiv (\tau_0, \beta_0^\top)^\top$. 

\begin{assumption}\label{assumption:iv_setup}
$ $

\vspace{-0.5\baselineskip}

    \begin{enumerate}[start=1,label={(\alph*)}]\setlength\itemsep{0.25em}
    \item $Y = D\tau_0 + X^\top\beta_0 + U,\: \E[U\vert X, Z^{(0)}]  \overset{a.s.}{=} 0$.
    \item $\E[D\vert X, Z^{(0)}]  \overset{a.s.}{=} Z^{(0)} + X^\top\pi_0$.
\end{enumerate}
\end{assumption}

For $K\in \mathbbm{N}$, define the CIV estimator as \begin{align*}
    \hat{\theta}^K \equiv \left(\En\hat{F}_KW^\top\right)\inv \En \hat{F}_KY
\end{align*}
where $\hat{F}_K \equiv (\hat{g}_K(Z,X), X^\top)^\top$ with $\hat{g}_K(Z,X) \equiv \hat{m}_K^{(n)}(Z) + X^\top \hat{\pi}$ and 
\begin{align}\label{eq:definition_m_hat}
    \hat{m}_K^{(n)}\equiv \argmin_{\substack{m: \mathcal{Z}^{(n)}\to \mathcal{M} \\ \vert m(\mathcal{Z}^{(n)})\vert \leq K}} \: \En(D - X^\top \hat{\pi} - m(Z^{(n)}))^2, 
\end{align}
with $\hat{\pi}$ being a root-$n$ consistent first-step estimator for $\pi_0$. Note that because $X$ is relatively low-dimensional, root-$n$ consistent estimators for $\pi_0$ can be obtained as within-category regression estimators of $D$ on $X$ with fixed effects defined by the categories of $Z^{(n)}$.

Assumption \ref{assumption:Pn_setup} ensures that the tail probabilities of the conditional expectation function residual $V \equiv D - Z^{(0)} - X^\top\pi_0$ and the exogenous variables $X$ decay at exponential rate. Analogously to \citet{bonhomme2015grouped}, Assumption \ref{assumption:Pn_setup} along with a dependence restriction (see Assumption \ref{assumption:asymptotics_setup} (d)) allows for application of exponential inequalities that are key to bound the probability of misclassifying values of $Z^{(n)}$ in estimation of the (approximate) optimal instruments. 

\begin{assumption}\label{assumption:Pn_setup}
$\exists b_1, b_2 > 0$ and $\tilde{b}_{1j},\tilde{b}_{2j}>0, \forall j \in [J]$, such that for all $n\in \mathbbm{N}$, $P_n$ satisfies
    \begin{enumerate}[start=1,label={(\alph*)}]\setlength\itemsep{0.25em}
\item $\Pr(\vert V\vert >v\vert Z^{(n)})\overset{a.s.}{\leq} \exp\left\{1 - \left(\frac{v}{b_1}\right)^{b_2}\right\}, \forall v >0$.
\item $\Pr(\vert X_j\vert >x\vert Z^{(n)})\overset{a.s.}{\leq} \exp\left\{1 - \left(\frac{x}{\tilde{b}_{1j}}\right)^{\tilde{b}_{2j}}\right\}, \forall x >0, j \in [J]$.
\end{enumerate}
\end{assumption}

Finally, Assumption \ref{assumption:asymptotics_setup} (a) places moment restrictions on the second stage error used, in particular, for consistent estimation of standard errors. Assumption \ref{assumption:asymptotics_setup} (b) requires compactness of the first stage coefficients. Assumption \ref{assumption:asymptotics_setup} (c) requires $\hat{\pi}$ to be a root-$n$ consistent estimator for $\pi_0$. Assumption \ref{assumption:asymptotics_setup} (d) asserts that the econometrician observes independent samples from $(Y, D, Z^{(n)}, X^\top)$.

\inPublicVersion{\clearpage}

\begin{assumption}\label{assumption:asymptotics_setup}
$ $

\vspace{-0.5\baselineskip}

\begin{enumerate}[start=1,label={(\alph*)}]\setlength\itemsep{0.25em}
\item $\exists L < \infty, \epsilon_2>2$ such that $\frac{1}{L}\leq \E U^{4}\leq L$. 
\item $\mathcal{Z}^{(0)}\subset \mathcal{M}$ and $\pi_0\in \Pi$ where $\mathcal{M}\subset \mathbbm{R}$ and  $\Pi\subset \mathbbm{R}^J$ are compact.
\item $\|\hat{\pi} - \pi_0\| = O_p(n^{-1/2}).$
\item The data is an i.i.d.\ sample \{$(Y_i, D_i, Z^{(n)}_i, X_i^\top)\}_{i=1}^n$ from $(Y, D, Z^{(n)}, X^\top)$. 
\end{enumerate}
\end{assumption}

Theorem \ref{theorem:CIV_w_covariates_mk} states the main result of the paper. In particular, it shows that when the infeasible oracle estimator \begin{align*}
    \tilde{\theta}^K &\equiv \left(\En F_KW^\top\right)\inv \En F_KY,
\end{align*}
where $F_K \equiv (g_K(Z^{(n)},X), X^\top)^\top$ with $g_K(Z^{(n)},X) \equiv m^{(n)}_K(Z^{(n)}) + X^\top \pi_0$, is root-$n$ normal, then assumptions \ref{assumption:rate_condition}-\ref{assumption:asymptotics_setup}  are sufficient for the CIV estimator $\hat{\theta}^K$ to be root-$n$ normal with the same asymptotic covariance matrix as long as $K\in \{2, \ldots, K_0\}$.

\begin{theorem}\label{theorem:CIV_w_covariates_mk}
    Let assumptions \ref{assumption:rate_condition}-\ref{assumption:asymptotics_setup} hold. If in addition $\E F_KW^\top$ for $K\leq K_0$ is non-singular, then, as $n\to \infty$,\begin{align*}
        \sqrt{n}\Sigma_K^{-\frac{1}{2}}(\hat{\theta}^K - \theta_0) \overset{d}{\to} N(0, \mathbbm{I}_J),
    \end{align*}
    where $\Sigma_K = \E[F_KW^\top]^{-1}\E[U^2F_KF_K^\top]\E[WF_K^\top]^{-1}$. This result continues to hold if $\Sigma_K$ is replaced with the consistent estimator \begin{align*}
        \hat{\Sigma}_K \equiv \En[\hat{F}_KW^\top]^{-1}\En[\hat{U}^2\hat{F}_K\hat{F}_K^\top]\En[W\hat{F}_K^\top]^{-1},
    \end{align*}
    where $\hat{U}\equiv Y - W^\top \hat{\theta}^K$.
\end{theorem}
\begin{proof}
    See Appendix \ref{app:proofs}.
\end{proof}

For the CIV estimator that uses $K=K_0$, Corollary \ref{corollary:semiparametric_efficiency} further provides a semiparametric efficiency result under homoskedasticity.\footnote{Note that in the heteroskedastic setting, weighting observations proportional to their variance can improve the asymptotic variance. Since this approach follows standard general method of moments arguments, I omit further discussion here.} 

 \begin{corollary}\label{corollary:semiparametric_efficiency}
     Let the assumptions of Theorem \ref{theorem:CIV_w_covariates_mk} hold. If in addition $\E[U^2\vert X, Z^{(0)}] \overset{a.s.}{=} \sigma^2$, then the asymptotic covariance $\Sigma_{K_0}$ of $\hat{\theta}^{K_0}$ achieves the semiparametric efficiency bound: \begin{align*}
        \Sigma_{K_0} = \sigma^2\E[h_0(Z^{(0)},X)h_0(Z^{(0)},X)^\top]^{-1},
     \end{align*}
     where $h_0(Z^{(0)},X) \equiv \E[W\vert X, Z^{(0)}].$
 \end{corollary}
\begin{proof}
    See Appendix \ref{app:semiparametric_proof}.
\end{proof}

\enlargethispage{\baselineskip}

\section{Monte Carlo Simulation}\label{sec:simulation}

This section discusses a Monte Carlo simulation exercise to illustrate finite sample behavior of the proposed CIV estimator and highlight key challenges of alternative optimal IV estimators for estimation with categorical instrumental variables. 

For $i = 1, \ldots, n$, the data generating process is given by \begin{align*}
    Y_i &= D_i \pi_{0}(X_i) + X_i\beta_0 + U_i,\\
    D_i&= m_0(Z_i) + X_i\gamma_0 + V_i,
\end{align*}
where $(U_i, V_i)\sim \mathcal{N}(0, \left[\begin{smallmatrix}
1 & 0.6 \\
0.6 & \sigma_V^2 \end{smallmatrix}\right])$, $D_i$ is a scalar-valued endogenous variable, $X_i\sim\textrm{Bernoulli}(\frac{1}{2})$ is a binary covariate and $\beta_0 = \gamma_0 = 0$, and $Z_i$ is the categorical instrument taking values in $\mathcal{Z} = \{1, \ldots, 40\}$ with equal probability. To introduce correlation between $Z_i$ and $X_i$, I further set $\Pr(Z_i \text{ is odd}\vert X_i = 0) = \Pr(Z_i \text{ is even}\vert X_i = 1) = 0$. The optimal instrument $m_0$ is constructed by first partitioning $\mathcal{Z}$ into $K_0$ equal subsets and then assigning evenly-spaced values in the interval $[0, C]$.\footnote{For example, for $K_0=2$, $m_0(z)=0$ for $z \in \{1,\ldots, 20\}$ and $m_0(z)=C$ for $z \in \{21,\ldots, 40\}$.} I choose the scalars $\sigma_V^2$ and $C$ such that the variance of the first stage variable is fixed to 1 and the concentration parameter in the smallest considered sample is $\mu^2 = 180$.\footnote{In particular, $\sigma_V^2=0.9$, and $C\approx 0.85$ for $K_0=2$ and $C\approx 1.153$ for $K_0=4$ so that with $n=800$, the concentration parameter $nM_0^\top (\Cov(\mathbbm{1}_z(Z))_{z \in \mathcal{Z}}) M_0/\sigma_V^2 = 180$ where $M_0$ is the 40-dimensional vector of first stage coefficients associated with every category. Choosing $\sigma_V^2$ and $C$ in this manner is akin to the simulation setup in \citet{belloni2012sparse}.} As in the simulation considered in \citet{kolesar2013estimation}, the data generating process allows for individual treatment effects $\pi_0(X_i)$ to differ with covariates. Here, $\pi_0(X_i) = \tau_0 + 0.5(1 - 2X_i)$ so that the expected treatment effect is simply $\E\pi_0(X) = \tau_0.$ As a consequence, the second stage is heteroskedastic and -- unlike two-step IV estimators like CIV -- the LIML estimator is inconsistent for the average treatment effect.

I compare properties of thirteen estimators in the simulation: An infeasible oracle estimator $\tilde{\theta}^{K_0}$ with known optimal instrument, CIV with $K=2$ and $K=4$, TSLS, JIVE, IJIVE, UJIVE, and LIML that use the observed instruments, and five machine-learning based IV estimators that use lasso with cross-validated or plug-in penalty parameters, ridge regression, gradient tree boosting, or random forests to estimate the optimal instrument.

Table \ref{tab:app_sim_resK2} provides the bias, median absolute error (MAE), rejection probabilities of a 5\% significance test (rp(0.05)),\footnote{Standard errors used for construction of rp(0.05) are heteroskedasitcity robust and do not include additional variance terms associated with traditional many instrument asymptotics. Inclusion of these additional terms has no qualitative consequence for the rejection probabilities of JIVE, IJIVE, or UJIVE.} and the inter-quantile range between 10th and 90th empirical quantiles (iqr(10, 90)) for a DGP in which $\tau_0 = 0$ and the optimal instrument has $K_0=2$ support points. All estimates are computed on sample sizes with 20, 25, 100, and 150 expected observations per observed instrument. As expected in the strong instruments setting considered here, the oracle estimator achieves small bias and nominal false rejection rates across all sample sizes. Its feasible analogues that attempt to estimate the optimal instrument in a first step, on the other hand, vary substantially across sample sizes. 

\begin{table}[!htbp]\small
        \begin{threeparttable}
        \centering
  \caption{Simulation Results ($K_0=2$)}\label{tab:app_sim_resK2}
\begin{tabular}{lccccccccc}
\toprule
\midrule
$K_0=2$ & \multicolumn{4}{c}{$\E N_z=20$ } &       & \multicolumn{4}{c}{$\E N_z= 25$ } \\
\cmidrule{2-5}\cmidrule{7-10}      & Bias  & MAE   & rp(0.05) & iqr(10,90) &       & Bias  & MAE   & rp(0.05) & iqr(10,90) \\ \midrule
Oracle & -0.006 & 0.058 & 0.050 & 0.238 &       & -0.008 & 0.058 & 0.037 & 0.215 \\
CIV (K=2) & 0.031 & 0.065 & 0.085 & 0.234 &       & 0.011 & 0.055 & 0.041 & 0.216 \\
CIV (K=4) & 0.101 & 0.104 & 0.269 & 0.237 &       & 0.078 & 0.088 & 0.232 & 0.223 \\
Lasso-IV (cv) & 0.184 & 0.190 & 0.587 & 0.331 &       & 0.156 & 0.157 & 0.530 & 0.278 \\
Lasso-IV (plug-in) & 0.261 & 0.306 & 0.526 & 0.941 &       & 0.251 & 0.306 & 0.524 & 0.866 \\
Ridge-IV (cv) & 0.120 & 0.123 & 0.353 & 0.249 &       & 0.097 & 0.104 & 0.333 & 0.225 \\
xboost-IV & 0.120 & 0.123 & 0.354 & 0.249 &       & 0.097 & 0.104 & 0.333 & 0.226 \\
ranger-IV & 0.129 & 0.132 & 0.394 & 0.258 &       & 0.107 & 0.113 & 0.368 & 0.235 \\
TSLS  & 0.120 & 0.123 & 0.353 & 0.249 &       & 0.097 & 0.104 & 0.333 & 0.226 \\
JIVE  & -0.022 & 0.087 & 0.118 & 0.342 &       & -0.020 & 0.079 & 0.085 & 0.295 \\
IJIVE & -0.013 & 0.085 & 0.118 & 0.335 &       & -0.012 & 0.080 & 0.083 & 0.291 \\
UJIVE & -0.013 & 0.085 & 0.119 & 0.336 &       & -0.013 & 0.080 & 0.083 & 0.291 \\
LIML  & -0.142 & 0.142 & 0.189 & 0.294 &       & -0.145 & 0.139 & 0.259 & 0.257 \\
      &       &       &       &       &       &       &       &       &  \\
$K_0=2$ & \multicolumn{4}{c}{$\E N_z=100$ } &       & \multicolumn{4}{c}{$\E N_z= 150$ } \\
\cmidrule{2-5}\cmidrule{7-10}      & Bias  & MAE   & rp(0.05) & iqr(10,90) &       & Bias  & MAE   & rp(0.05) & iqr(10,90) \\ \midrule
Oracle & -0.003 & 0.029 & 0.045 & 0.106 &       & -0.001 & 0.021 & 0.048 & 0.085 \\
CIV (K=2) & -0.003 & 0.029 & 0.045 & 0.106 &       & -0.001 & 0.021 & 0.048 & 0.085 \\
CIV (K=4) & 0.016 & 0.033 & 0.112 & 0.122 &       & 0.013 & 0.025 & 0.099 & 0.097 \\
Lasso-IV (cv) & 0.036 & 0.045 & 0.234 & 0.143 &       & 0.026 & 0.032 & 0.194 & 0.111 \\
Lasso-IV (plug-in) & 0.062 & 0.067 & 0.385 & 0.178 &       & 0.025 & 0.032 & 0.181 & 0.109 \\
Ridge-IV (cv) & 0.025 & 0.038 & 0.163 & 0.130 &       & 0.019 & 0.028 & 0.134 & 0.101 \\
xboost-IV & 0.025 & 0.038 & 0.162 & 0.129 &       & 0.019 & 0.028 & 0.134 & 0.101 \\
ranger-IV & 0.032 & 0.043 & 0.207 & 0.139 &       & 0.025 & 0.032 & 0.186 & 0.109 \\
TSLS  & 0.025 & 0.038 & 0.162 & 0.129 &       & 0.019 & 0.028 & 0.134 & 0.101 \\
JIVE  & -0.007 & 0.036 & 0.125 & 0.138 &       & -0.002 & 0.027 & 0.090 & 0.105 \\
IJIVE & -0.006 & 0.036 & 0.122 & 0.138 &       & -0.001 & 0.027 & 0.088 & 0.105 \\
UJIVE & -0.006 & 0.036 & 0.122 & 0.138 &       & -0.001 & 0.027 & 0.088 & 0.105 \\
LIML  & -0.141 & 0.141 & 0.826 & 0.129 &       & -0.139 & 0.138 & 0.940 & 0.102 \\ 
\midrule \bottomrule
\end{tabular}
  \begin{tablenotes}[para,flushleft]
  \scriptsize
  \item \textit{Notes.} Simulation results are based on 1000 replications using the DGP described in Section \ref{sec:simulation} with $K_0=2$. For each replication, $\E \pi_0(X) = 0$ but potentially $\En \pi_0(X) \neq 0$. Results are thus normalized by $\En \pi_0(X).$ $\E N_z$ denotes the expected number of observations per observed category, MAE denotes the median absolute error, rp(0.05) denotes the false rejection probability at a 5\% significance level, and iqr(10, 90) denotes the inter-quantile range between the 10\% and 90\% quantile. ``Oracle'' denotes the infeasible oracle estimator $\tilde{\theta}^{K_0}$ with known optimal instrument, ``CIV ($K=2$)'' and ``CIV ($K=4$)'' correspond to the proposed categorical IV estimators restricted to 2 and 4 support points in the first stage, ``Lasso-IV (cv)'' and ``Lasso-IV (plug-in)'' denotes IV estimators that use lasso to estimate the optimal instrument using penalty parameters chosen via 10-fold cross-validation or via the plug-in rule of \citet{belloni2012sparse} , ``Ridge-IV (cv)'' denotes an IV estimator that uses ridge regression to estimate the optimal instrument using a penalty parameter chosen via 10-fold cross-validation, ``\texttt{xgboost}-IV'' and ``\texttt{ranger}-IV'' denote IV estimators that use gradient tree boosting as implemented by the \texttt{xgboost} package and random forests as implemented by the \texttt{ranger} package to estimate the optimal instrument, ``TSLS'' denotes the two-stage least squares estimator using the observed instruments, ``JIVE'', ``IJIVE'', and ``UJIVE'' denote the jackknife instrumental variable estimators of \citet{angrist1999jackknife}, \citet{ackerberg2006comment}, and \citet{kolesar2013estimation}, respectively, and ``LIML'' denotes the limited information maximum likelihood estimator using the observed instruments.  
  \end{tablenotes}
    \end{threeparttable}
\end{table}

The CIV estimator restricted to two support points achieves near-oracle performance at a moderate number of observations per category. This is in strong contrast to the alternative optimal instrument estimators in this setting. In particular, even at the much larger sample size with 150 observations per category, TSLS has a false rejection rate of 0.15, far above the 5\% nominal level. Further, none of the considered machine-learning based optimal instrument estimators improve upon TSLS. Given that there are only 40 first stage instruments in a total sample size of up to 6000 ($=40 \times 150$), it may be surprising that application of the lasso does not result in better empirical performance. However, note that the shrinkage assumptions (implicitly) leveraged by any of the machine-learning based estimators are not suitable approximations of the categorical instrumental variable design considered here. Only the CIV estimator with over-specified number of support points has slightly lower bias and false rejection rates than TSLS, yet, remains inferior to the CIV estimator with correct number of groups. Note that the theory provided in this paper does not provide results for CIV estimators with $K>K_0$. Finally, all jackknife-based estimators exhibit small bias, but are substantially more dispersed than CIV. This dispersion prohibits JIVE, IJIVE, and UJIVE to accurately control size even at 150 expected observations per instrument. This dispersion seems in large part be due to the treatment effect heterogeneity present in the considered DGP: Replicating the results for a DGP \textit{without} treatment effect heterogeneity shows that the jackknife-based estimators control size at all sample sizes (see Appendix \ref{app:additional_sims}).

Table \ref{tab:app_sim_resK4} replicates the simulation in a DGP where $K_0=4$. In contrast to the previous results on CIV with \textit{over}-specified support points, the results show that estimated confidence intervals for CIV with \textit{under}-specified number of support points can achieve correct coverage. While the first-stage estimation problem is substantially more challenging with $K_0=4$ as all support points are closer together, CIV with both $K=2$ and $K=4$ achieves near-oracle performance for the larger sample sizes. This is again in strong contrast to any of the competing optimal instrument estimators, whose biases are an order of magnitude larger and whose false rejection probabilities are two to three times those of the two CIV estimators. Similarly, the jackknife-based estimators are more dispersed than CIV and fail to control size in the data generating process with treatment effect heterogeneity.

Finally, as expected given the heterogeneous second-stage effects, LIML is heavily biased for the average treatment effect throughout in both Table \ref{tab:app_sim_resK2} and  \ref{tab:app_sim_resK4}.

\begin{table}[!htbp]\small
        \begin{threeparttable}
        \centering
  \caption{Simulation Results ($K_0=4$)}\label{tab:app_sim_resK4}
\begin{tabular}{lccccccccc}
\toprule
\midrule
$K_0=4$ & \multicolumn{4}{c}{$\E N_z=20$ } &       & \multicolumn{4}{c}{$\E N_z= 25$ } \\
\cmidrule{2-5}\cmidrule{7-10}      & Bias  & MAE   & rp(0.05) & iqr(10,90) &       & Bias  & MAE   & rp(0.05) & iqr(10,90) \\ \midrule
Oracle & 0.003 & 0.064 & 0.040 & 0.229 &       & 0.005 & 0.060 & 0.034 & 0.227 \\
CIV (K=2) & 0.093 & 0.102 & 0.190 & 0.236 &       & 0.077 & 0.086 & 0.161 & 0.218 \\
CIV (K=4) & 0.113 & 0.115 & 0.326 & 0.244 &       & 0.092 & 0.096 & 0.270 & 0.225 \\
Lasso-IV (cv) & 0.168 & 0.159 & 0.452 & 0.352 &       & 0.131 & 0.133 & 0.416 & 0.308 \\
Lasso-IV (plug-in) & 0.190 & 0.249 & 0.475 & 0.763 &       & 0.183 & 0.219 & 0.492 & 0.657 \\
Ridge-IV (cv) & 0.122 & 0.125 & 0.365 & 0.258 &       & 0.103 & 0.105 & 0.333 & 0.234 \\
xboost-IV & 0.122 & 0.125 & 0.364 & 0.258 &       & 0.103 & 0.105 & 0.333 & 0.234 \\
ranger-IV & 0.125 & 0.128 & 0.378 & 0.263 &       & 0.106 & 0.107 & 0.347 & 0.241 \\
TSLS  & 0.122 & 0.125 & 0.364 & 0.258 &       & 0.103 & 0.105 & 0.333 & 0.234 \\
JIVE  & -0.018 & 0.095 & 0.096 & 0.343 &       & -0.012 & 0.080 & 0.106 & 0.299 \\
IJIVE & -0.009 & 0.092 & 0.100 & 0.337 &       & -0.005 & 0.078 & 0.106 & 0.295 \\
UJIVE & -0.010 & 0.093 & 0.098 & 0.337 &       & -0.005 & 0.078 & 0.106 & 0.295 \\
LIML  & -0.142 & 0.143 & 0.209 & 0.297 &       & -0.136 & 0.138 & 0.244 & 0.281 \\
      &       &       &       &       &       &       &       &       &  \\
$K_0=4$ & \multicolumn{4}{c}{$\E N_z=100$ } &       & \multicolumn{4}{c}{$\E N_z= 150$ } \\
\cmidrule{2-5}\cmidrule{7-10}      & Bias  & MAE   & rp(0.05) & iqr(10,90) &       & Bias  & MAE   & rp(0.05) & iqr(10,90) \\ \midrule
Oracle & 0.000 & 0.027 & 0.051 & 0.104 &       & 0.000 & 0.023 & 0.055 & 0.090 \\
CIV (K=2) & 0.008 & 0.031 & 0.064 & 0.119 &       & 0.002 & 0.027 & 0.069 & 0.103 \\
CIV (K=4) & 0.010 & 0.028 & 0.075 & 0.111 &       & 0.003 & 0.024 & 0.063 & 0.093 \\
Lasso-IV (cv) & 0.031 & 0.041 & 0.192 & 0.142 &       & 0.022 & 0.034 & 0.201 & 0.120 \\
Lasso-IV (plug-in) & 0.022 & 0.038 & 0.143 & 0.140 &       & 0.015 & 0.032 & 0.161 & 0.116 \\
Ridge-IV (cv) & 0.026 & 0.036 & 0.158 & 0.124 &       & 0.018 & 0.031 & 0.169 & 0.110 \\
xboost-IV & 0.026 & 0.036 & 0.157 & 0.124 &       & 0.018 & 0.031 & 0.169 & 0.110 \\
ranger-IV & 0.027 & 0.038 & 0.166 & 0.131 &       & 0.019 & 0.032 & 0.177 & 0.115 \\
TSLS  & 0.026 & 0.036 & 0.157 & 0.124 &       & 0.018 & 0.031 & 0.169 & 0.110 \\
JIVE  & -0.006 & 0.035 & 0.110 & 0.133 &       & -0.003 & 0.031 & 0.119 & 0.116 \\
IJIVE & -0.004 & 0.034 & 0.109 & 0.133 &       & -0.002 & 0.031 & 0.115 & 0.116 \\
UJIVE & -0.004 & 0.034 & 0.109 & 0.133 &       & -0.002 & 0.031 & 0.115 & 0.116 \\
LIML  & -0.141 & 0.142 & 0.837 & 0.124 &       & -0.140 & 0.139 & 0.944 & 0.102 \\
\midrule \bottomrule
\end{tabular}%
  \begin{tablenotes}[para,flushleft]
  \scriptsize
  \item \textit{Notes.} Simulation results are based on 1000 replications using the DGP described in Section \ref{sec:simulation} with $K_0=4$. See the notes of Table \ref{tab:app_sim_resK2} for a description of the estimators.
  \end{tablenotes}
    \end{threeparttable}
\end{table}

For additional insights on the empirical performance of the considered estimators, Figure \ref{fig:simMK_1} plots power curves for the hypothesis test $H_0: \tau_0 = 0$ at significance level $\alpha = 0.05$ in the design with $K_0=2$. Panels (a) and (b) keep the second stage coefficients constant, while panels (c) and (d) mirror the heterogeneous second stage effect design above. For brevity, the figure focuses on only a subset of estimators considered previously: CIV with $K=2$, the oracle estimator, TSLS, JIVE, LIML, and the lasso-based IV estimator with cross-validated penalty level. Appendix \ref{app:additional_sims} provides the figures corresponding to the remaining estimators.

\begin{figure}[!htbp] 
\caption{Power Curves with and without Treatment Effect Heterogeneity}\label{fig:simMK_1}
    \centering
    \begin{subfigure}[b]{0.375\textwidth}
    \centering
    \includegraphics[width=1\textwidth]{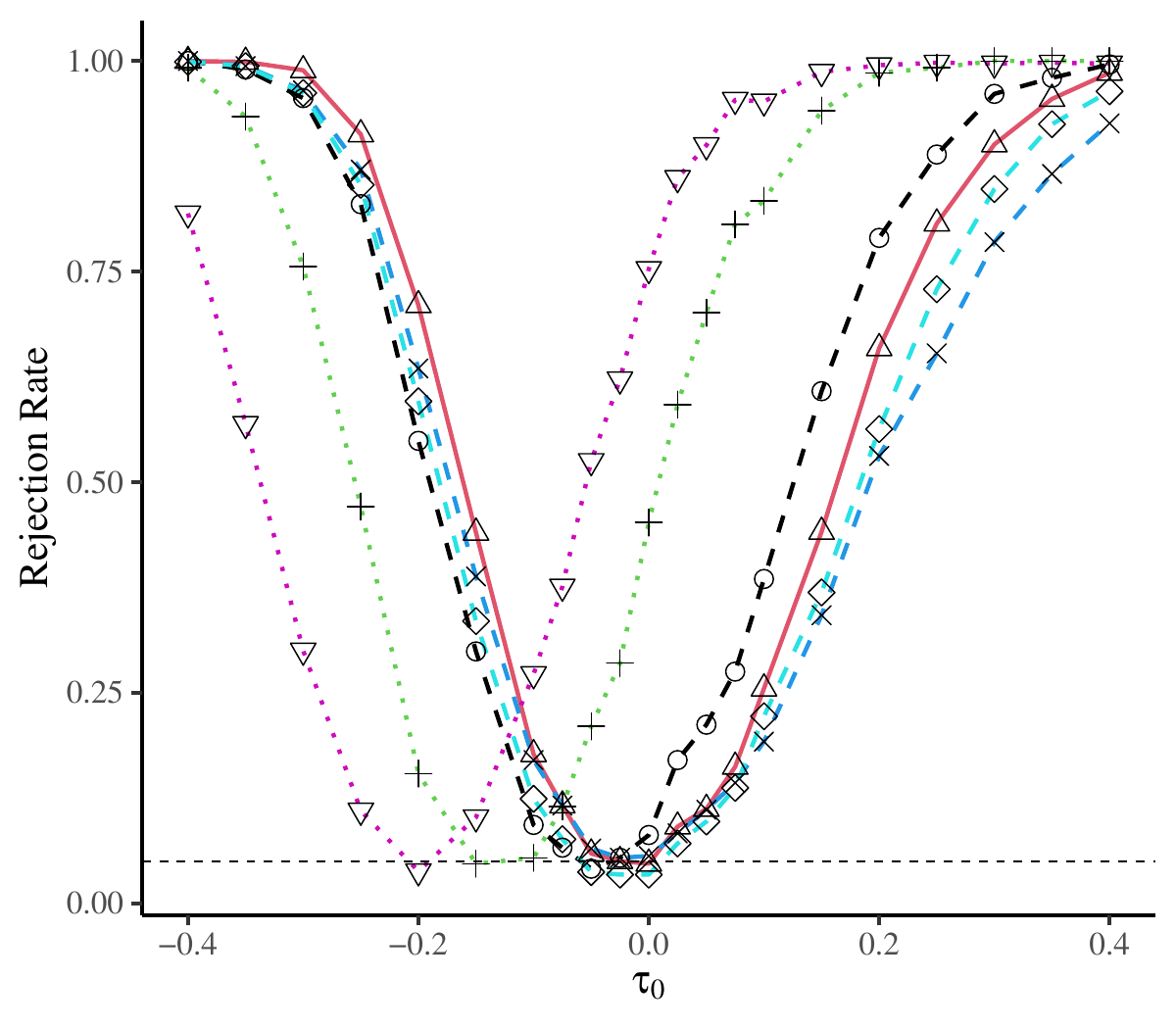}
    \subcaption{constant $\pi_0$, $\E N_z= 20$}
    \end{subfigure}
    \begin{subfigure}[b]{0.375\textwidth}
    \centering
    \includegraphics[width=1\textwidth]{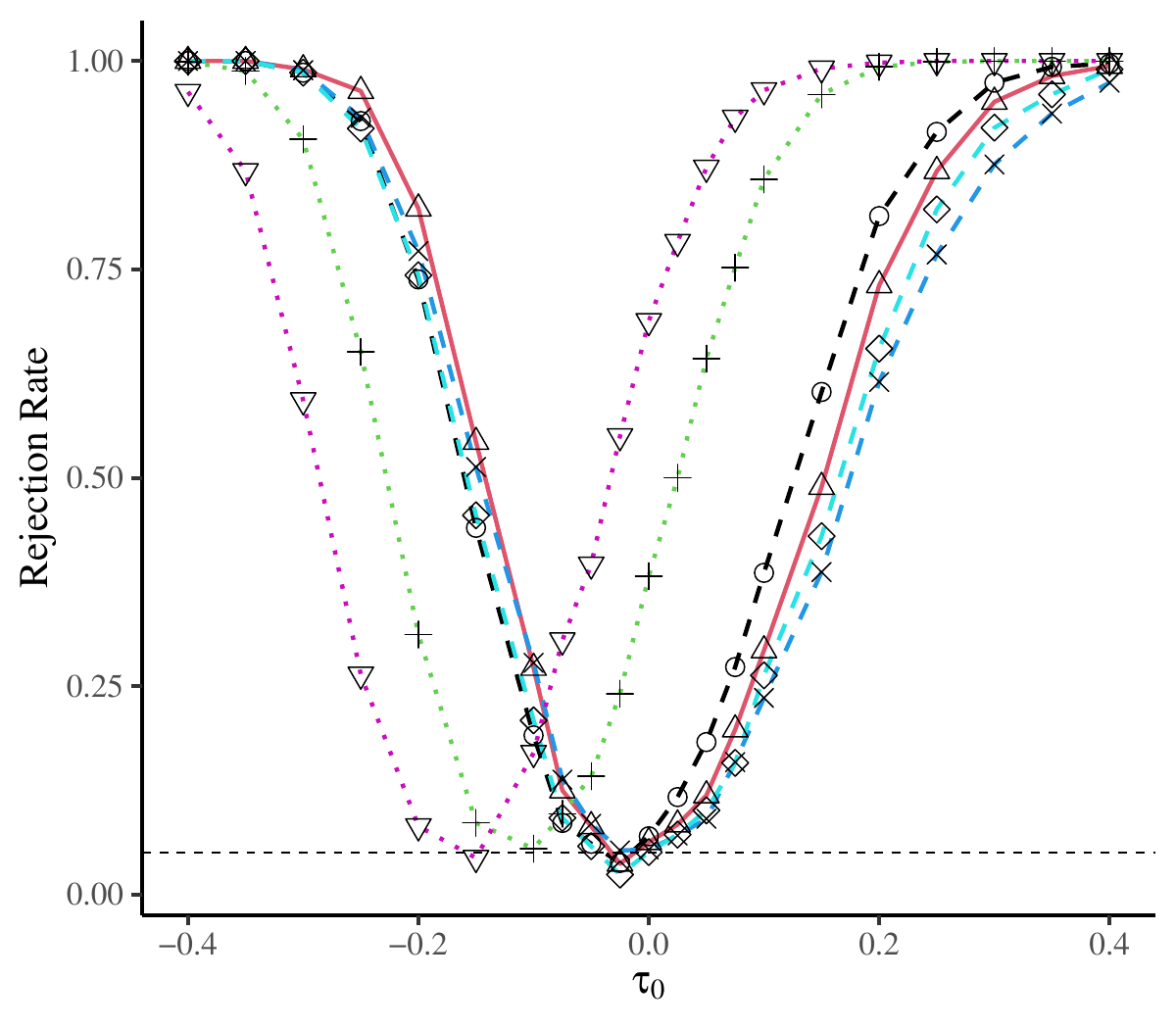}
    \subcaption{constant $\pi_0$, $\E N_z= 25$}
    \end{subfigure}
    \begin{subfigure}[b]{0.15\textwidth}
    \centering
    \includegraphics[width=1\textwidth]{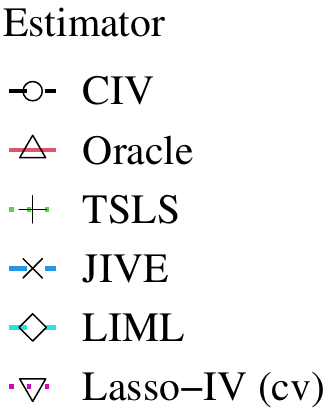}\vspace{7em}
    \end{subfigure}
    \\
    \begin{subfigure}[b]{0.375\textwidth}
    \centering
    \includegraphics[width=1\textwidth]{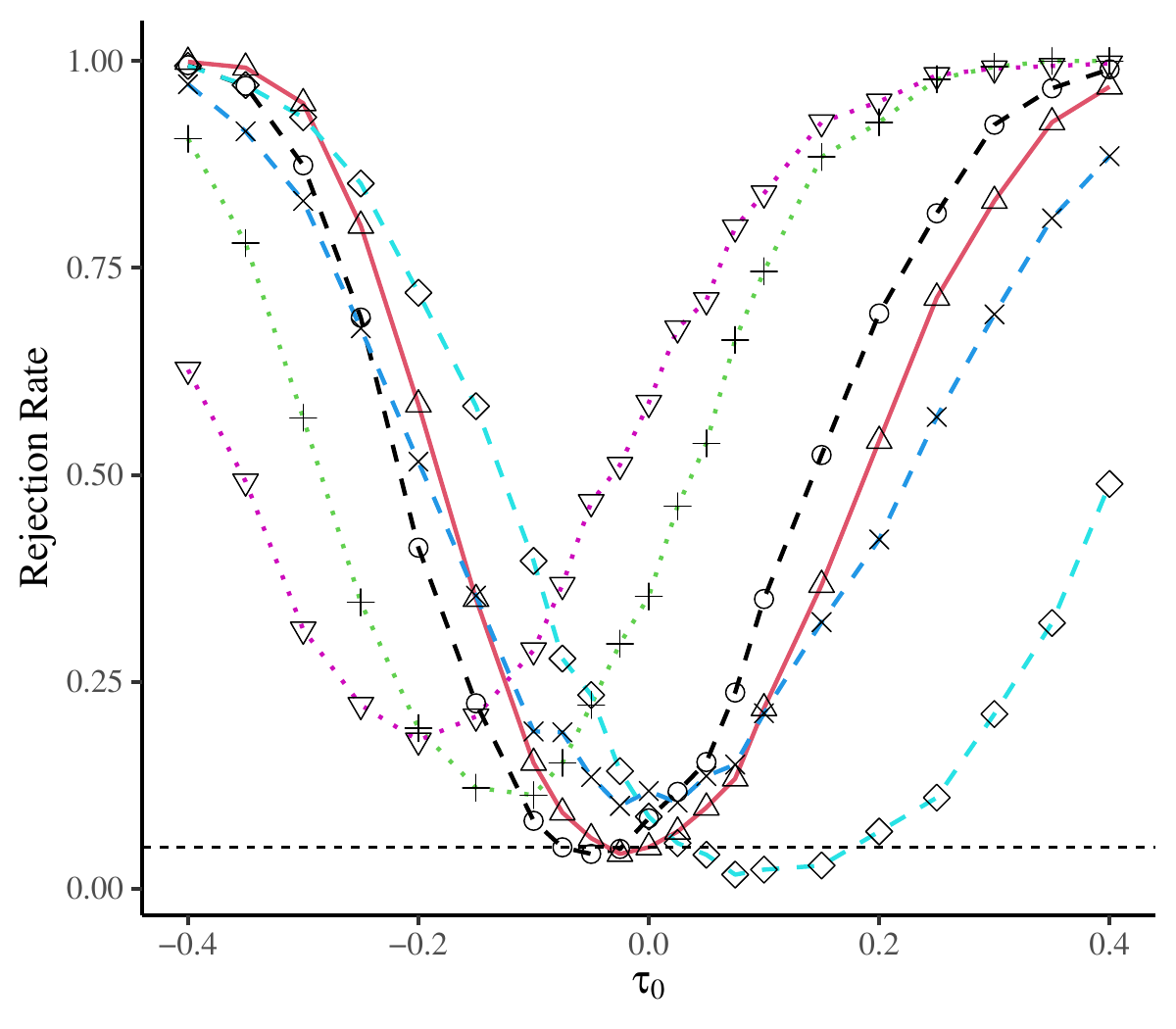}
    \subcaption{heterogeneous $\pi_0$, $\E N_z= 20$}
    \end{subfigure}
    \begin{subfigure}[b]{0.375\textwidth}
    \centering
    \includegraphics[width=1\textwidth]{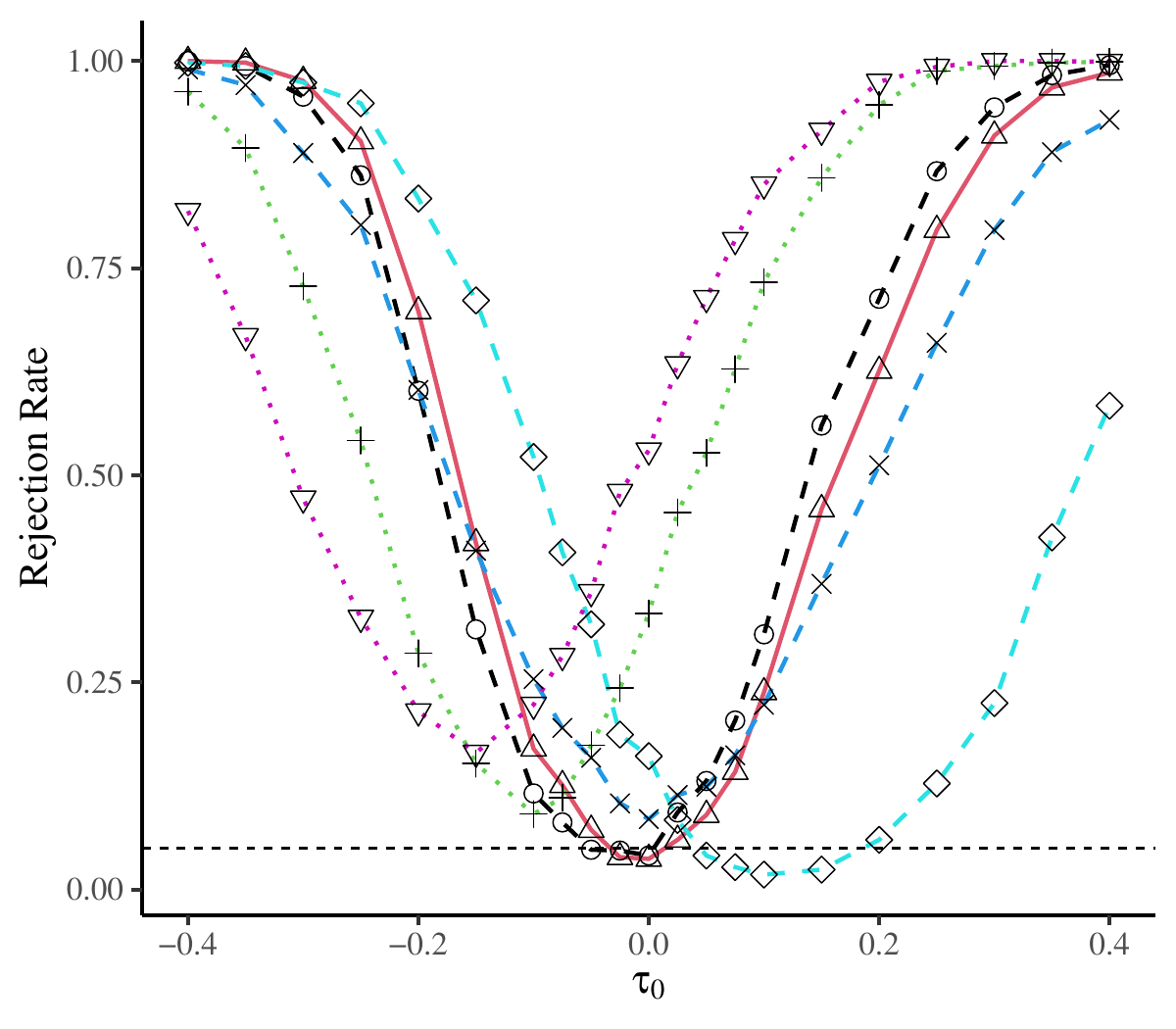}
    \subcaption{heterogeneous $\pi_0$, $\E N_z= 25$}
    \end{subfigure}\hspace{6em}
     \vskip0.25em
 \begin{minipage}{0.9\textwidth} 
 {\scriptsize \textit{Notes.} Simulation results are based on 1000 replications using the DGP described in Section \ref{sec:simulation} with $K_0=2$. Panels (a) and (b) are with constant effects so that $\pi_0(X_i) = \tau_0$. Panels (c) and (d) allow for covariate-dependent effects with $\pi_0(X_i) = 1 - 2X_i +  \tau_0$. The power curves plot the rejection rate of testing $H_0: \tau_0 = 0$ at significance level $\alpha=0.05$. ``CIV'' correspond to the proposed categorical IV estimator with $K=2$, ``Oracle'' denotes the infeasible oracle estimator $\tilde{\theta}^{K_0}$ with known optimal instrument, ``TSLS'' denotes the two-stage least squares estimator using the observed instruments, ``JIVE'' denotes the jackknife IV estimator of \citet{angrist1999jackknife}, ``LIML'' denotes the limited information maximum likelihood estimator using the observed instruments, and ``Lasso-IV (cv)'' denotes IV estimator that use lasso to estimate the optimal instrument using a cross-validated penalty parameter.} 
 \end{minipage}
\end{figure} 

When there is no treatment effect heterogeneity (panels (a) and (b)), the LIML and JIVE estimators achieves near oracle performance at even the smallest sample size considered. This mirrors the insights of \citet{donald2001choosing}, \citet{bekker2005instrumental}, and \citet{chao2012asymptotic}, as well as the excellent empirical performance of the LIML estimator in the simulations of \citet{angrist2020machine}. CIV with $K=2$ results in similar rejection rates for the moderate sample with 25 expected observations per category. In contrast to LIML and JIVE, CIV retains its near-oracle performance when second stage effects are heterogeneous (panels (c) and (d)). In all panels, TSLS and Lasso-IV (cv) are heavily biased.

\section{Effects of Pre-Trial Release on Conviction}\label{sec:empirical_example_judges}

This section applies the CIV estimator to a judge fixed effects instrumental variable analysis of the effect of pre-trial release on conviction. As in \citet{dobbie2018effects} and \citet{chyn2024examiner}, I consider a 2006-14 sample of misdemeanor and felony cases assigned to weekend bail hearings in the Miami-Dade County, Florida.\footnote{The data is publicly available from \citet{chyn2024examiner}.} The data covers 94,355 cases for 186 bail judges. The goal of the example is to illustrate the application of CIV and alternative IV estimators in settings with potentially noisy estimates of judge leniency.

Following the previous literature analyzing the Miami-Dade data, the relationship of interest is the causal effect of pre-trial release on the conviction of the defendant. Pre-trial detention occurs if a defendant who has been assigned to a bail hearing does not post bail thereafter.\footnote{As described in \citet{chyn2024examiner}, bail can also be posted after an initial bail value is set immediately following arrest. These observations are not in the sample.} Because the bail judge can change the bail amount, whether or not a judge is lenient may have a direct effect on the pre-trial release of a defendant. Using bail judge identity as an instrument is then motivated by the fact that bail judges are quasi-randomly assigned conditional on the court and date of the hearing. See, in particular, \citet{chyn2024examiner} for a comprehensive discussion of the IV assumptions in the Miami-Dade setting.

I focus on the analysis of a simple instrumental variable specification:\begin{align*}
    Y &= D\tau_0 + X^\top \beta_0 + \varepsilon,\\
    D &= m_0(Z) + X^\top \pi_0 + \nu, 
\end{align*}
with $\E[\varepsilon\vert Z, X] = \E[\nu\vert Z, X]=0$, where $Y$ is an indicator equal to one if the defendant is convicted, $D$ is an indicator equal to one if they have met bail, $X$ is a set of court-by-time fixed effects, $Z$ is the categorical instrument capturing bail judge identities, and $m_0(z)$ is the (unobserved) leniency of judge $z\in \mathcal{Z}$. The parameter of interest is the parameter $\tau_0$, which captures the causal effect of pre-trial release on conviction under correct model misspecification. Under stronger distributional independence and monotonicity assumptions, $\tau_0$ can also capture a convex combination of local average treatment effects in a data generating process with unobserved treatment effect heterogeneity \citep[see, in particular,][]{frandsen2023judging,blandhol2022tsls}.

A potential concern with estimating $\tau_0$ given an i.i.d.\ sample $\{(Y_i, D_i, X_i, Z_i)\}_{i=1}^n$ via TSLS is that the number of cases per judge can be moderate so that a many instrument bias arises. As shown in Figure \ref{fig:hist_ncases}, the distribution of cases per judge indeed varies substantially, with most judges working on around 500 cases but a few completing less than 200 cases. The fact that leniency of judges may be estimated noisily may then motivate the use of jackknife-based estimation in \citet{dobbie2018effects} and \citet{chyn2024examiner}.\footnote{In contrast, the dimension of the court-by-time fixed effects in the Miami-Dade County setting is not a first order concern. Appendix \ref{app:additional_plots_judges} shows that that almost all fixed effects are associated with more than 600 cases with no large left tail.} 

\begin{figure}[!h] 
\caption{Distribution of Cases per Judge}\label{fig:hist_ncases}
    \centering
    \includegraphics[width=0.5\textwidth]{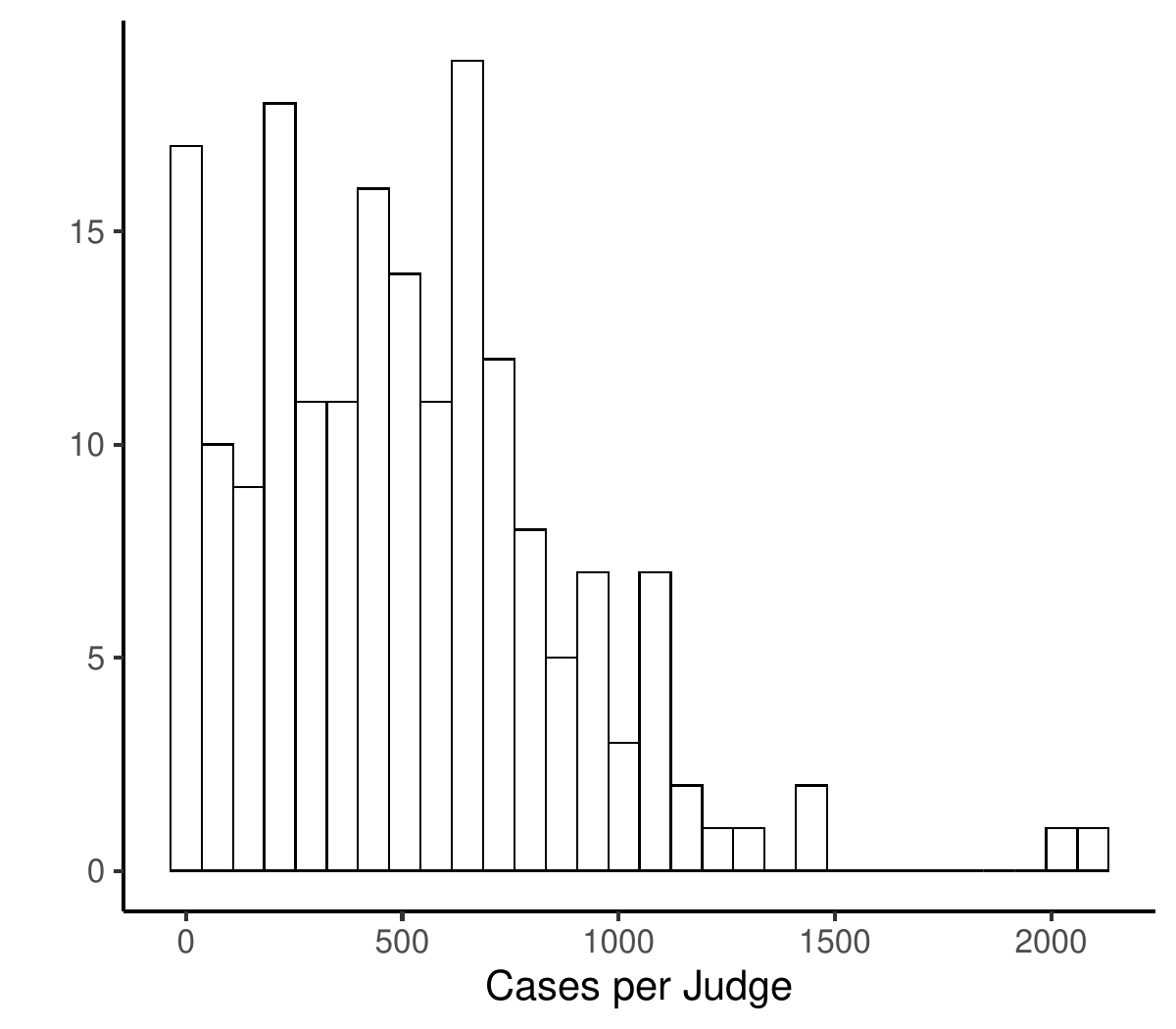}
     \vskip0.25em
 \begin{minipage}{0.9\textwidth} 
 {\scriptsize \textit{Notes.} The histogram plots the distribution of cases per judge in the Miami-Dade data. The vertical axis corresponds to the number of judges (not the number of observations).}
 \end{minipage}
\end{figure} 

Column (6) in Table \ref{tab:judge_res} presents estimates for $\tau_0$ using the sample of in \citet{chyn2024examiner} which is restricted to judges with at least 200 cases. The considered estimators are the proposed CIV estimator with $K=2$ and $K=5$ support points, TSLS, JIVE, IJIVE, UJIVE, and the post-Lasso IV estimator of \citet{belloni2012sparse} that selects among the judge fixed effects. I also provide OLS estimates for comparison. With exception of the post-Lasso IV estimator which is estimated with high variance,\footnote{The post-Lasso IV estimator using the plug-in penalty parameter of \citet{belloni2012sparse} selects only three judge fixed effects and removes any others from the specification. The poor performance of the lasso-based estimator mirrors its poor performance in categorical instruments settings in Section \ref{sec:simulation} and in \citet{angrist2020machine}.} the point estimates of TSLS differ only marginally from those of the jackknife-based estimators or CIV with any level of regularization. There are substantial differences in the standard errors, yet, given the large share of observations associated with judges who have more than 500 cases, it might not be surprising that any potential bias from many instruments is small.

To gain a better understanding of estimator differences in settings with potentially noisy leniency estimates, I thus also consider artificially restricted sub-samples of the Miami-Dade data. These sub-samples are meant to represent examiner fixed effects settings with fewer decisions per decision maker so that bias from many instruments may be more likely to arise. For this purpose, columns (1)-(4) restrict the sample to judges with \textit{at most} 400, 500, 600, and 700 cases, respectively. Throughout, I also restrict the sample to judges with at least 30 cases. Columns (1)-(4) may thus be well approximated by the ``moderately many instruments'' asymptotics outlined in this paper. For comparison, column (5) also provides results for judges with at least 30 cases but where the number of cases per judge has no upper bound. Throughout, I focus on comparison of estimators \textit{within} a particular column. Within a column, TSLS and the jackknife-based estimators target the same estimand as CIV with a correct choice of $K$. However, because the distribution of judges changes across columns, the targeted estimand can also differ across columns when treatment effects are heterogeneous. These changes in estimands can further complicate comparisons across columns.

\begin{table}[!htb]\small
\centering
        \begin{threeparttable}
  \caption{Estimates of the Effect of Pre-Trial Release on Conviction}\label{tab:judge_res}
\begin{tabular}{lcccccc}
\toprule
\midrule
      & \multicolumn{1}{c}{(1)} & \multicolumn{1}{c}{(2)} & \multicolumn{1}{c}{(3)} & \multicolumn{1}{c}{(4)} & \multicolumn{1}{c}{(5)} & \multicolumn{1}{c}{(6)} \\
\midrule
\multicolumn{1}{l}{OLS} & \multicolumn{1}{c}{-0.243} & \multicolumn{1}{c}{-0.248} & \multicolumn{1}{c}{-0.243} & \multicolumn{1}{c}{-0.241} & \multicolumn{1}{c}{-0.232} & \multicolumn{1}{c}{-0.232} \\
      & \multicolumn{1}{c}{(0.009)} & \multicolumn{1}{c}{(0.007)} & \multicolumn{1}{c}{(0.006)} & \multicolumn{1}{c}{(0.005)} & \multicolumn{1}{c}{(0.003)} & \multicolumn{1}{c}{(0.003)} \\
\multicolumn{1}{l}{TSLS} & \multicolumn{1}{c}{-0.358} & \multicolumn{1}{c}{-0.403} & \multicolumn{1}{c}{-0.303} & \multicolumn{1}{c}{-0.312} & \multicolumn{1}{c}{-0.267} & \multicolumn{1}{c}{-0.275} \\
      & \multicolumn{1}{c}{(0.092)} & \multicolumn{1}{c}{(0.083)} & \multicolumn{1}{c}{(0.075)} & \multicolumn{1}{c}{(0.064)} & \multicolumn{1}{c}{(0.049)} & \multicolumn{1}{c}{(0.052)} \\
\multicolumn{1}{l}{CIV (K=2)} & \multicolumn{1}{c}{-0.508} & \multicolumn{1}{c}{-0.549} & \multicolumn{1}{c}{-0.407} & \multicolumn{1}{c}{-0.392} & \multicolumn{1}{c}{-0.276} & \multicolumn{1}{c}{-0.271} \\
      & \multicolumn{1}{c}{(0.133)} & \multicolumn{1}{c}{(0.102)} & \multicolumn{1}{c}{(0.09)} & \multicolumn{1}{c}{(0.081)} & \multicolumn{1}{c}{(0.063)} & \multicolumn{1}{c}{(0.066)} \\
\multicolumn{1}{l}{CIV (K=5)} & \multicolumn{1}{c}{-0.407} & \multicolumn{1}{c}{-0.431} & \multicolumn{1}{c}{-0.324} & \multicolumn{1}{c}{-0.294} & \multicolumn{1}{c}{-0.279} & \multicolumn{1}{c}{-0.292} \\
      & \multicolumn{1}{c}{(0.098)} & \multicolumn{1}{c}{(0.085)} & \multicolumn{1}{c}{(0.077)} & \multicolumn{1}{c}{(0.066)} & \multicolumn{1}{c}{(0.051)} & \multicolumn{1}{c}{(0.054)} \\
\multicolumn{1}{l}{CIV (K=20)} & \multicolumn{1}{c}{-0.362} & \multicolumn{1}{c}{-0.409} & \multicolumn{1}{c}{-0.302} & \multicolumn{1}{c}{-0.312} & \multicolumn{1}{c}{-0.268} & \multicolumn{1}{c}{-0.274} \\
      & \multicolumn{1}{c}{(0.092)} & \multicolumn{1}{c}{(0.083)} & \multicolumn{1}{c}{(0.074)} & \multicolumn{1}{c}{(0.064)} & \multicolumn{1}{c}{(0.049)} & \multicolumn{1}{c}{(0.052)} \\
\multicolumn{1}{l}{JIVE} & \multicolumn{1}{c}{0.671} & \multicolumn{1}{c}{2.280} & \multicolumn{1}{c}{-1.614} & \multicolumn{1}{c}{-0.627} & \multicolumn{1}{c}{-0.310} & \multicolumn{1}{c}{-0.329} \\
      & \multicolumn{1}{c}{(0.93)} & \multicolumn{1}{c}{(3.285)} & \multicolumn{1}{c}{(2.782)} & \multicolumn{1}{c}{(0.373)} & \multicolumn{1}{c}{(0.115)} & \multicolumn{1}{c}{(0.129)} \\
\multicolumn{1}{l}{IJIVE} & \multicolumn{1}{c}{-0.484} & \multicolumn{1}{c}{-0.607} & \multicolumn{1}{c}{-0.382} & \multicolumn{1}{c}{-0.387} & \multicolumn{1}{c}{-0.289} & \multicolumn{1}{c}{-0.299} \\
      & \multicolumn{1}{c}{(0.184)} & \multicolumn{1}{c}{(0.194)} & \multicolumn{1}{c}{(0.161)} & \multicolumn{1}{c}{(0.129)} & \multicolumn{1}{c}{(0.08)} & \multicolumn{1}{c}{(0.085)} \\
\multicolumn{1}{l}{UJIVE} & \multicolumn{1}{c}{-0.481} & \multicolumn{1}{c}{-0.599} & \multicolumn{1}{c}{-0.378} & \multicolumn{1}{c}{-0.386} & \multicolumn{1}{c}{-0.289} & \multicolumn{1}{c}{-0.299} \\
      & \multicolumn{1}{c}{(0.183)} & \multicolumn{1}{c}{(0.189)} & \multicolumn{1}{c}{(0.158)} & \multicolumn{1}{c}{(0.128)} & \multicolumn{1}{c}{(0.08)} & \multicolumn{1}{c}{(0.084)} \\
\multicolumn{1}{l}{post-Lasso IV} & \multicolumn{1}{c}{-0.733} & \multicolumn{1}{c}{-0.677} & \multicolumn{1}{c}{-0.631} & \multicolumn{1}{c}{-0.454} & \multicolumn{1}{c}{-0.453} & \multicolumn{1}{c}{-0.456} \\
      & \multicolumn{1}{c}{(1.054)} & \multicolumn{1}{c}{(1.109)} & \multicolumn{1}{c}{(0.995)} & \multicolumn{1}{c}{(0.579)} & \multicolumn{1}{c}{(0.261)} & \multicolumn{1}{c}{(0.259)} \\
      &       &       &       &       &       &  \\
min \# cases & 30    & 30    & 30    & 30    & 30    & 200 \\
max \# cases & 400   & 500   & 600   & 700   & \multicolumn{1}{c}{$\infty$} & \multicolumn{1}{c}{$\infty$} \\
Observations & 13,259 & 23,505 & 31,093 & 46,612 & 94,253 & 91,421 \\
\midrule
\bottomrule
\end{tabular}%
  \begin{tablenotes}[para,flushleft]
  \scriptsize
  \item \textit{Notes.} ``CIV (K=2)'' and ``CIV (K=5)''  denote the proposed categorical IV estimators restricted to 2 and 5 support points in the first stage, ``TSLS'' denotes two-stage least squares, ``JIVE'', ``IJIVE'', and ``UJIVE'' denotes the jackknife-based IV estimators of \citet{angrist1999jackknife}, \citet{ackerberg2009improved}, and \citet{kolesar2013estimation}, respectively, ``post-Lasso IV'' denotes the post-lasso IV estimator of \citet{belloni2012sparse}, and ``OLS'' denotes ordinary least squares. Heteroskedasticity-robust standard errors are in parentheses.
  \end{tablenotes}
    \end{threeparttable}
\end{table}

In the settings of columns (1)-(4) with moderately many cases per judge, regularization of the first stage via a support point restriction has a substantial impact on the estimates. The first considered CIV estimator restricts the estimated judge leniency to $K=2$ support points. This restriction can correspond to the assumption that all judges are either ``lenient'' or ``strict.'' Alternatively, $K=2$ can be viewed as optimally approximating a higher-dimensional vector of judge types with a binary ``pseudo'' type. For insights into the estimated binary (pseudo) type, Figure \ref{fig:judge_types} presents the empirical distribution of estimated judge leniency measures as well as how the leniency measures map to the binary (pseudo) type estimate. Here, the dotted vertical lines correspond to the leniency of the ``lenient'' and ``strict'' judges, respectively, and the dashed vertical line indicates the cut-off for assigning judges to these types. 

The CIV estimator with $K=2$ effectively regularizes the first stage problem of estimating judge leniency. In columns (1)-(4), the CIV estimate differs from the (unregularized) TSLS estimate by about one standard error. Further, the CIV estimate is negative and statistically significant throughout, suggesting that pre-trial release has a negative causal effect on conviction of a defendant regardless of the particular subsample. This is in strong contrast to the estimates associated with JIVE, which does not result in statistically significant coefficients at conventional levels. The IJIVE and UJIVE estimates are numerically similar to the CIV (K=2) estimates but are associated with substantially larger standard errors throughout. For the smaller subsamples, IJIVE and UJIVE standard errors are between 30-60\% larger than CIV standard errors.

The CIV estimator with $K=5$ (pseudo) types is included here to illustrate that lower regularization (i.e., higher $K$) results in the CIV estimator to tend towards TSLS. Indeed, for all considered subsamples, CIV with $K=5$ is similar to (unregularized) TSLS.

\begin{figure}[!h] 
\caption{Estimated Judge Leniency and Binary (Pseudo) Types}\label{fig:judge_types}
    \centering
    \includegraphics[width=0.5\textwidth]{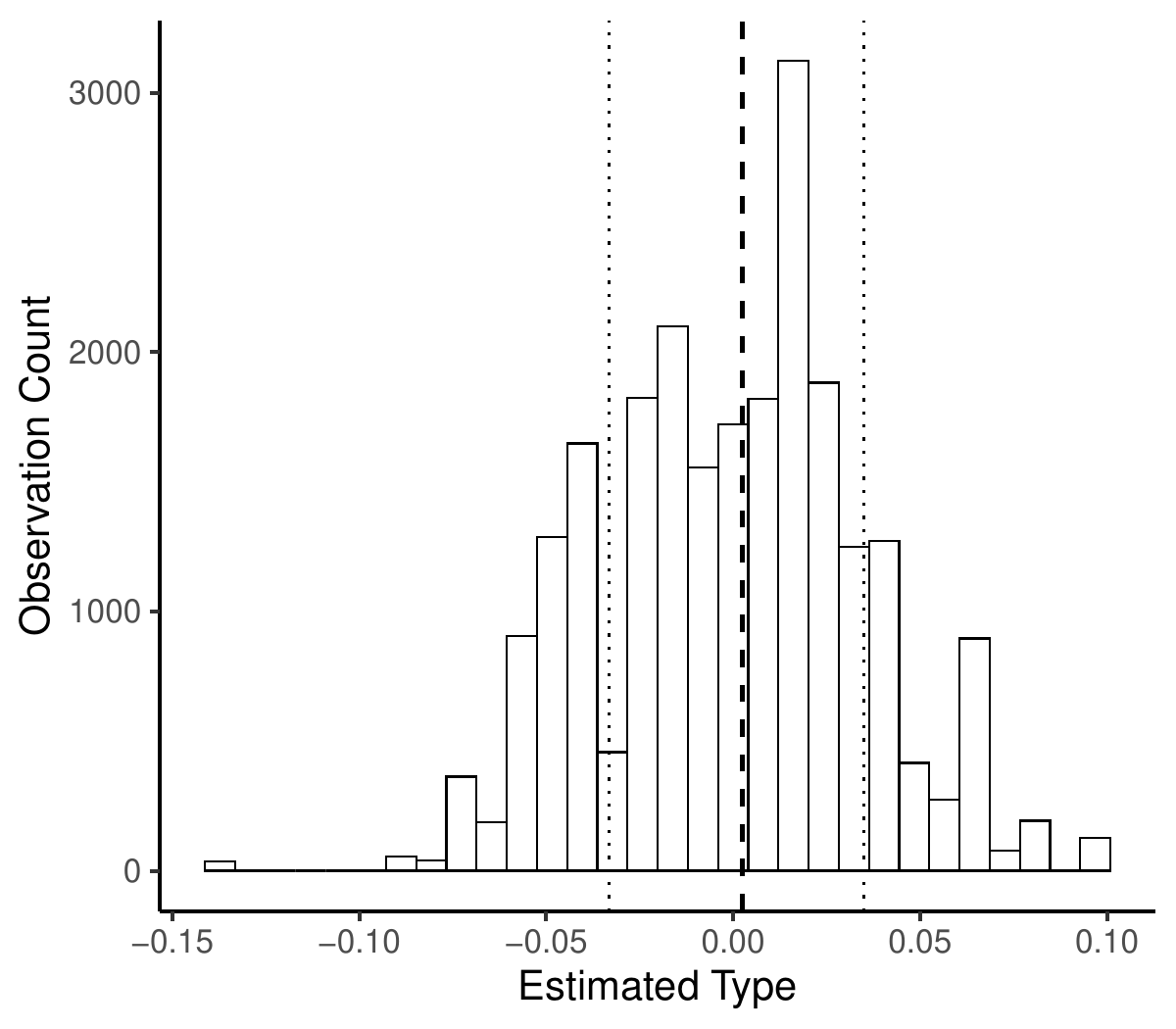}
     \vskip0.25em
 \begin{minipage}{0.9\textwidth} 
 {\scriptsize \textit{Notes.} The histogram plots the empirical distribution of estimated judge leniency for the sample of judges with at most 500 and at least 30 cases. The dotted lines represent the estimated binary (pseudo) types underlying the CIV ($K=2$) estimator. The dashed line represents the cut-off value. All observations to its left are assigned to the ``lenient'' (pseudo) type, all observations to its right are assigned to the ``strict'' (pseudo) type.}
 \end{minipage}
\end{figure}

In summary, the presented results using the Miami-Dade County data suggest that the jackknife-estimators are either highly subsample-dependent (as is the case for JIVE) or are similar to highly regularized CIV with a binary (pseudo) type but with larger standard errors. The proposed CIV estimator (e.g., with a binary (pseudo) type) may thus provide a useful alternative to currently popular jackknife-based estimators in examiner fixed effects settings with moderately many decisions per decision maker.

\section{Conclusion} \label{sec:conclusion}

This paper considers estimation with categorical instrumental variables when the number of observations per category is relatively small. The proposed categorical instrumental variable estimator is motivated by a first-stage regularization assumption that restricts the unknown optimal instrument to have fixed finite support. In a ``moderately many instruments'' regime that allows the number of observations to grow at arbitrarily slow polynomial rate with the sample, I show that when the number of support points of the optimal instrument is known, CIV achieves the same asymptotic variance as the infeasible oracle two-stage least squares estimator that presumes knowledge of the optimal instrument and is semiparametrically efficient under homoskedasticity. Further, under-specifying the number of support points maintains asymptotic normality but results in efficiency loss. A simulation exercise illustrates the finite sample performance of the proposed CIV estimator and highlights pitfalls associated with alternative optimal instrument estimators in the setting of categorical instruments. Similar to results in \citet{angrist2020machine}, lasso-based IV estimators do not improve upon the bias of TSLS and fail to control size. In contrast, CIV successfully leverages the low-dimensional structure of the optimal instrument to obtain near-oracle estimates. Finally, the analysis of pre-trial release on conviction using bail judge identities as instruments as in \citet{dobbie2018effects} and \citet{chyn2024examiner} illustrates potential practical advantages of CIV over conventional jackknife-based instrumental variable estimators. In particular, highly regularized CIV with a binary (pseudo) type achieves lower standard errors than IJIVE and UJIVE while resulting in numerically similar point estimates.


\interlinepenalty=10000
\addcontentsline{toc}{section}{References}
\bibliographystyle{apalike}
\bibliography{biblio}
\interlinepenalty=10

\clearpage
\appendix

\newpage
\pagenumbering{gobble}

\begin{titlepage}

	\newgeometry{left=1in, right=1in,top=2.25in, bottom=1in}
	
	{\begin{spacing}{1}  \LARGE \centering Optimal Categorical Instrumental Variables \end{spacing}}

    \vspace{1cm}
	{\Large \centering Supplemental Appendices\par}
	
	\vspace{1cm}
	\vspace{2.5cm}

\end{titlepage}
\restoregeometry

\newpage

\pagenumbering{arabic}
\section{Proofs}\label{app:proofs}

The proof of Theorem \ref{theorem:CIV_w_covariates_mk} and Corollary \ref{corollary:semiparametric_efficiency} proceeds in four steps: First, I begin the proof with a set of lemmas to characterize the asymptotic properties of the first stage estimator $\hat{m}_K^{(n)}$. The proof of lemmas \ref{lemma:consistency_m_covariates}-\ref{lemma:rate_g} heavily leverages the arguments of \citet{bonhomme2015grouped}. Novel arguments provided here include the characterization of the approximation $m^{(n)}_K$ of $m^{(n)}_{0}$ for $K\leq K_0$ in Lemma \ref{lemma:mK_characteristics} which leverages the dynamic programming characterization of KMeans in $\mathbbm{R}$ of \citet{wang2011ckmeans}, and accommodations to unbalanced categorical variables including Lemma \ref{lemma:binom_bound}. Second, I characterize the asymptotic distribution of $\hat{\theta}^K$. Third, I prove consistency of the covariance estimator $\hat{\Sigma}_K$ in Lemma \ref{lemma:SimgaK_consistency}. Finally, I prove semiparametric efficiency of $\hat{\theta}^{K_0}$ under homoskedasticity.

\textbf{Notation.} For notational brevity, the proof defines $Z \equiv Z^{(n)}$, $\mathcal{Z}\equiv \mathcal{Z}^{(n)}$, $m_K \equiv m_K^{(n)}$, and $\hat{m}_K^{(n)} \equiv \hat{m}_K$, omitting the explicit dependence on $n$.

\subsection{Convergence of $\hat{m}_K$}

I begin by highlighting two important properties of the optimal instrument $Z^{(0)}$. Note that by Assumption \ref{assumption:Z0_setup}, it follows from the definition of the support of a random variable that \begin{align}\label{eq:Z0_nonvanishing}
    \exists \underline{p}>0: \: \Pr( Z^{(0)} = d_z)\geq \underline{p}, \forall d_z \in \mathcal{Z}^{(0)}.
\end{align}
Further, since $\mathcal{Z}^{(0)}$ is a finite collection of points in $\mathbbm{R}$, \begin{align}\label{eq:Z0_wellseperated}
    \exists c>0:\: (d_z - \tilde{d}_z)^2 \geq c,\forall d_z \neq \tilde{d}_z \in \mathcal{Z}^{(0)}.
\end{align}

Lemma \ref{lemma:mK_characteristics} shows that the approximation of $Z^{(0)}$ with $K\in\{2, \ldots, K_0\}$ unique support points also satisfies \eqref{eq:Z0_nonvanishing} and \eqref{eq:Z0_wellseperated}.

\begin{lemma}\label{lemma:mK_characteristics}
    Let the assumptions of Theorem \ref{theorem:CIV_w_covariates_mk} hold. Then, $\forall K \in \{2, \ldots, K_0\}$, $m_K$ satisfies  \begin{enumerate}[start=1,label={(\alph*)}]\setlength\itemsep{0.25em}
\item$\Pr(m_K(Z) = \alpha)\geq \underline{p}, \forall \alpha \in m_K(\mathcal{Z})$, and 
\item $(\alpha - \Tilde{\alpha})^2 \geq c,\forall \alpha \neq \tilde{\alpha} \in m_K(\mathcal{Z}),$
\end{enumerate}
where $\underline{p}$ and $c$ are positive constants from conditions \eqref{eq:Z0_nonvanishing} and \eqref{eq:Z0_wellseperated}.
\end{lemma}

\begin{proof}

Note that by Assumption \ref{assumption:Z0_setup} (b) we have $Z^{(0)} = m_0(Z)$ so that the result follows directly for $K = K_0$. Further, note that $m_K$ is equivalently defined by 
\begin{align*}
    m_K \equiv \argmin_{\substack{m: \mathcal{Z}\to \mathcal{M} \\ \vert m(\mathcal{Z})\vert \leq K}} \: \E(m_0(Z) - m(Z) )^2.
\end{align*}

Let $\alpha^0_1 \leq \ldots \leq \alpha^0_{K_0}$ and $\alpha^K_1 \leq \ldots \leq \alpha^K_K$ denote the support points of $m_0(Z)$ and $m_K(Z)$, respectively, in non-decreasing order, and let $p^0_1, \ldots, p^0_K$ and $p^K_1, \ldots, p^K_K$ denote the corresponding point masses. 

    Let $G[k, M]$ denote the minimum objective function in the sub-problem of approximating the first $2\leq k\leq K_0$ support points of $m_0(Z)$ by $1<M\leq k$ clusters -- i.e., \begin{align*}
        G[k, M] \equiv \min_{\substack{m:\mathcal{Z}\to \mathcal{M}\\ \vert m(\mathcal{Z})\vert \leq M}}\E\left[\left(m_0(Z) - m(Z)\right)^2\big\vert m_0(Z) \leq \alpha^0_k\right],
    \end{align*}
    so that $G[K_0, K]$ is the minimum objective function of the original problem. Now following the arguments of \citet{wang2011ckmeans}, let $M\leq j\leq k$ denote the smallest support point in the $M$th cluster of the solution to $G[k,M]$. Then $G[j-1, M-1]$ must be the minimum objective function of clustering the first $j-1$ support points into $M-1$ cluster as otherwise $G[k,M]$ would not be optimal. As a consequence, the problem admits the following sub-structure:\begin{align*}
        G[k, M] = \min_{M\leq j \leq k} & \:\bigg\{ G[j-1, M-1]\Pr(m_0(Z) \leq \alpha^0_{j-1} \vert m_0(Z) \leq \alpha^0_k) \\
        & + \min_{a_M^K \in \mathbbm{M}}\bigg[\E\left[\left(m_0(Z) - a_M^K\right)^2\big\vert \alpha^0_{j} \leq m_0(Z) \leq \alpha^0_k\right] \\
        & \qquad \quad \quad \times \Pr(m_0(Z) \geq \alpha^0_{j} \vert m_0(Z) \leq \alpha^0_k)\bigg] \bigg\},
    \end{align*}
    for all $2 \leq M \leq k \leq K_0$, where the second term characterizes the optimal within-cluster points. The desired properties of $m_K$ now follow directly from the sub-structure. 
    
    First, for any $\ell \geq 2$, we have \begin{align*}
        \vert \alpha^K_{\ell -1} - \alpha^K_{\ell} \vert \geq \vert \alpha^0_{j_{\ell} - 1} - \alpha^0_{j_{\ell}} \vert \geq c, 
    \end{align*}
    where $j_\ell$ is the index of the smallest support point in the $\ell$th cluster and the second inequality follows from \eqref{eq:Z0_nonvanishing}. 

    Second, it is clear from the sub-structure that support points are combined so that \begin{align*}
        \min \{p_k^K\}_{k=1}^{K} \geq \min \{p_k^0\}_{k=1}^{K_0} \geq \underline{p},
    \end{align*}
    where the second inequality follows from \eqref{eq:Z0_wellseperated}.
\end{proof}

Note that, because $\vert \mathcal{Z} \vert$ is finite for every $n\in \mathbbm{N}$ by Assumption \ref{assumption:rate_condition} and $m_K$ takes $K$ unique values, each such function $m:\mathcal{Z} \to \mathcal{M}, \vert m(\mathcal{Z})\vert \leq K$ is fully characterized by a set of coefficients $(\alpha_k)_{k=1}^K$ and a map $\gamma: \mathcal{Z} \to \{1,\ldots, K\}$. It is thus possible to re-cast $m_K$ as \begin{align*}
   (\alpha^K, \gamma^K)\equiv\argmin_{(\alpha, \gamma) \in \mathcal{M}^K \times \Gamma^K} \: \E(Z^{(0)} - \alpha_{\gamma(Z)}) )^2.
\end{align*}
where $\Gamma^K = \{\gamma: \mathcal{Z} \to \{1,\ldots, K\}\}$ so that \begin{align*}
    m_K(z) = \alpha^K_{\gamma}(z), \: \forall z \in \mathcal{Z}.
\end{align*}
Similarly, for the corresponding estimator $\hat{m}_K$ we have \begin{align}\label{eq:definition_alphagamma_hat}
    (\hat{\alpha}^K, \hat{\gamma}^K)\equiv \argmin_{(\alpha, \gamma) \in \mathcal{M}^K \times \Gamma^K} \: \En(D - X^\top \hat{\pi} - \alpha_{\gamma(Z)})^2, 
\end{align}
so that \begin{align*}
    \hat{m}_K(z) = \hat{\alpha}^K_{\hat{\gamma}^K(z)}, \: \forall z \in \mathcal{Z}.
\end{align*}
This representation of $m_K$ and $\hat{m}_K$ follows the group fixed effects estimator of \citet{bonhomme2015grouped}. I adopt it here because it allows for separate analysis of estimation properties of the partition and the coefficients.

\begin{lemma}\label{lemma:consistency_m_covariates}
    Let the assumptions of Theorem \ref{theorem:CIV_w_covariates_mk} hold. Then, $\forall K \in \{2, \ldots, K_0 \}$, we have \begin{align*}
      \En \left(m_K(Z) - \hat{\alpha}^K_{\hat{\gamma}^K(Z)}\right)^2 = o_p(1).
    \end{align*}
\end{lemma}

\begin{proof}
Fix an arbitrary $K \in \{2, \ldots, K_0 \}$ and define \begin{align*}
    \hat{Q}_K(\alpha, \gamma, \pi) = \En \left(D - \alpha_{\gamma(Z)} - X^\top\pi\right)^2 = \En \left(m_K(Z) + X^\top \pi_0 + V_{K} - \alpha_{\gamma(Z)} -  X^\top\pi\right)^2
\end{align*}
and \begin{align*}
    \tilde{Q}_K(\alpha, \gamma, \pi) = \En \left(m_K(Z) - \alpha_{\gamma(Z)} + X^\top(\pi_0 - \pi)\right)^2 + \En V_{K}^2,
\end{align*}
where \begin{align*}
    V_K \equiv V + m_0(Z) - m_K(Z).
\end{align*}

Now consider
\begin{align}\label{eq:eq:proof_lemma2_0}
\begin{aligned}
            \hat{Q}_K(\alpha, \gamma, \hat{\pi}) - \tilde{Q}_K(\alpha, \gamma, \hat{\pi}) &= 2\En  V_K(m_K(Z) - \alpha_{\gamma(Z)}+ X^\top(\pi_0 - \hat{\pi})) \\
        &= 2\En V_{K}m_K(Z) + 2\En V_KX^\top(\pi_0 - \hat{\pi})  - 2\En V_{K}\alpha_{\gamma(Z)},
\end{aligned}
    \end{align}
where only the third term depends on $(\alpha, \gamma)$. In particular, we have  \begin{align*}
    \En V_{K}\alpha_{\gamma(Z)} & =\En\left[\sum_{z \in \mathcal{Z}}\mathbbm{1}_z(Z)V_{K}\alpha_{\gamma(z)}\right] =\frac{1}{K_Z}\sum_{z \in \mathcal{Z}}\alpha_{\gamma(z)}\left(\frac{K_Z}{n}\sum_{i=1}^n \mathbbm{1}_z(Z_i)V_{iK}\right),
\end{align*}
so by Cauchy-Schwarz \begin{align}\label{eq:proof_lemma2_1}
    \left(\frac{1}{K_Z}\sum_{z \in \mathcal{Z}}\alpha_{\gamma(z)}\left(\frac{K_Z}{n}\sum_{i=1}^n \mathbbm{1}_z(Z_i)V_{iK}\right)\right)^2 \leq \left(\frac{1}{K_Z}\sum_{z \in \mathcal{Z}}\alpha_{\gamma(z)}^2\right)\left(\frac{1}{K_Z}\sum_{z \in \mathcal{Z}}\left(\frac{K_Z}{n}\sum_{i=1}^{n} \mathbbm{1}_z(Z_i)V_{iK}\right)^2\right).
\end{align}
The first term in \eqref{eq:proof_lemma2_1} is uniformly bounded by Assumption \ref{assumption:asymptotics_setup} (b). For the second term, we have \begin{align*}
    \frac{1}{K_Z}\sum_{z \in \mathcal{Z}}\left(\frac{K_Z}{n}\sum_{i=1}^{n} \mathbbm{1}_z(Z_i)V_{iK}\right)^2 = \frac{K_Z}{n^2}\sum_{z \in \mathcal{Z}}\sum_{i=1}^{n} \sum_{j=1}^{n} V_{iK}V_{jK}\mathbbm{1}_z(Z_i)\mathbbm{1}_z(Z_j).
\end{align*}
Taking expectations results in \begin{align*}
   & \E\left[\frac{K_Z}{n^2}\sum_{z \in \mathcal{Z}}\sum_{i=1}^{n} \sum_{j=1}^{n} V_{iK}V_{jK}\mathbbm{1}_z(Z_i)\mathbbm{1}_z(Z_j)\right]\\
   \overset{[1]}{=}& \frac{K_Z}{n} \E V_{K}^2 \\
   \overset{[2]}{\leq}&\frac{K_Z}{n}(\E V^2+ \| m_0 - m_K\|_{\infty}^2)  \overset{[3]}{=} o(1),
\end{align*}
where {\footnotesize [1]} follows from Assumption \ref{assumption:asymptotics_setup} (d) and the fact that $\E V_K = 0$, {\footnotesize [2]} follows from Assumption \ref{assumption:iv_setup} (b), and {\footnotesize [3]} is a consequence of and Assumption \ref{assumption:Pn_setup} (a), Assumption \ref{assumption:asymptotics_setup} (b), Assumption \ref{assumption:rate_condition} and the fact that $\sum_{z \in \mathcal{Z}} \Pr(Z=z) = 1$ which implies $K_Z=o(n)$. 

The first and second term in \eqref{eq:eq:proof_lemma2_0} do not depend on $(\alpha, \gamma)$ and converge in probability to zero. In particular, the first term in \eqref{eq:eq:proof_lemma2_0} is $o_p(1)$ by arguments analogous to arguments presented above for the third term by replacing the constants $\alpha_{\gamma(z)}$ with the constants $m_K(z)$. For the second term in \eqref{eq:eq:proof_lemma2_0}, convergence to zero follows from $\En V_k X^\top = O_p(1)$ by Assumption \ref{assumption:asymptotics_setup} (d), Assumption \ref{assumption:Pn_setup}, and Assumption \ref{assumption:asymptotics_setup} (b), and $\pi_0 - \hat{\pi} = o_p(1)$ by Assumption \ref{assumption:asymptotics_setup} (c). 

We thus have that \begin{align}\label{eq:uniform_convergence}
    \sup_{\substack{(\alpha, \gamma) \in \mathcal{M}^{K}\times\Gamma^K}} \big\vert\hat{Q}_K(\alpha, \gamma, \hat{\pi}) - \tilde{Q}_K(\alpha, \gamma, \hat{\pi}) \big\vert = o_p(1).
\end{align}

Consider now
\begin{align}\label{eq:pain_convergence}
\begin{aligned}
   \tilde{Q}_K(\hat{\alpha}^K, \hat{\gamma}^K, \hat{\pi}) &\overset{[1]}{=} \hat{Q}_K(\hat{\alpha}^K, \hat{\gamma}^K,\hat{\pi}) + o_p(1) \overset{[2]}{\leq} \hat{Q}_K(\alpha^K, \gamma^K, \hat{\pi}) + o_p(1)  \overset{[3]}{=} \tilde{Q}_K(\alpha^K, \gamma^K, \hat{\pi}) + o_p(1),
\end{aligned}
\end{align}
where {\footnotesize [1]} and {\footnotesize [3]} follow from \eqref{eq:uniform_convergence}, and {\footnotesize [2]} follows from the definition of $(\hat{\alpha}^K, \hat{\gamma}^K)$.

The desired result then follows from \begin{align*}
   o_p(1) &= \tilde{Q}_K(\hat{\alpha}^K, \hat{\gamma}^K, \hat{\pi}) - \tilde{Q}_K(\alpha^K, \gamma^K, \hat{\pi}) \\
   &= \En(m_K(Z) -\hat{\alpha}^K_{\hat{\gamma}^K(Z)})^2  + 2\En (m_K(Z) -\hat{\alpha}^K_{\hat{\gamma}^K(Z)})(X^\top(\pi_0 - \hat{\pi}))\\
   & = \En(m_K(Z) -\hat{\alpha}^K_{\hat{\gamma}^K(Z)})^2  + o_p(1),
\end{align*}
where the last equality follows from Cauchy-Schwarz applied to the second term \begin{align*}
    \left(\En (m_K(Z) -\hat{\alpha}^K_{\hat{\gamma}^K(Z)})(X^\top(\pi_0 - \hat{\pi}))\right)^2 &\leq\En (m_K(Z) -\hat{\alpha}^K_{\hat{\gamma}^K(Z)})^2 \En \|X\|^2 \|\pi_0 - \hat{\pi}\|^2\\
    & \overset{[1]}{=} O_p(1)o_p(1) = o_p(1),
\end{align*}
where {\footnotesize [1]} follows from Assumption \ref{assumption:Pn_setup} (b) and Assumption \ref{assumption:asymptotics_setup} (b)-(c).

\end{proof}


Since the estimator $(\hat{\alpha}^K, \hat{\gamma}^K)$ of $m_K$ is invariant to relabeling of the coefficients $\hat{\alpha}^K$ and clustering $\hat{\gamma}^K$, it is useful to consider the Hausdorff distance $d_H$ in $\mathbbm{R}^{K}$ to characterize the asymptotic properties of the estimators $\hat{\alpha}^K$. In particular, define  \begin{align*}
    d_H(a, b)^2 = \max\bigg\{&\max_{k \in \{1, \ldots, K\}}\left(\min_{\tilde{k} \in \{1, \ldots, K\}}\left(\alpha_{\tilde{k}} - b_k\right)^2\right), \max_{\tilde{k} \in \{1, \ldots, K\}}\left(\min_{k \in \{1, \ldots, K\}}\left(\alpha_{\tilde{k}} - b_k\right)^2\right)\bigg\}.
\end{align*}

Define $\alpha^K = (\alpha_1^K, \ldots, \alpha_K^K)$ to be a vector of the $K$ support points of $m_K$ and let $\gamma^K$ be the corresponding clustering so that \begin{align*}
    m_K(Z) = \alpha^K_{\gamma^K(Z)}.
\end{align*}

\begin{lemma}\label{lemma:consistent_alpha}
    Let the assumptions of Theorem \ref{theorem:CIV_w_covariates_mk} hold. Then, $ \forall K \in \{2, \ldots, K_0 \}$,\begin{align*}
      d_H(\hat{\alpha}^K, \alpha^K) = o_p(1).
    \end{align*}
\end{lemma}

\begin{proof}

The proof proceeds in two steps. I first show that for all $k \in \{1,\ldots, K\}$, \begin{align}\label{eq:proof_consistent_alpha_1}
    \min_{\tilde{k} \in \{1, \ldots, K\}}\left(\hat{\alpha}^K_{\tilde{k}} - \alpha^K_k\right)^2 = o_p(1).
\end{align}

Let $k \in \{1, \ldots, K\}$. We have \begin{align*}
    \En \left(\min_{\tilde{k}\in\{1,\ldots, K\}}\mathbbm{1}\{\gamma^K(Z)=k\}(\hat{\alpha}^K_{\tilde{k}} - \alpha^K_k)^2\right) = \left(\En \mathbbm{1}\{\gamma^K(Z)=k\}\right)\left(\min_{\tilde{k}\in\{1,\ldots, K\}}(\hat{\alpha}^K_{\tilde{k}} - \alpha^K_k)^2\right).
\end{align*}
Since the first term of the right-hand side is non-vanishing by Lemma \ref{lemma:mK_characteristics} (a), it suffices to show that the left-hand side is $o_p(1)$ for all $k \in \{1, \ldots, K\}$. In particular, \begin{align*}
  \En \left(\min_{\tilde{k}\in\{1,\ldots, K\}}\mathbbm{1}\{\gamma^K(Z)=k\}(\hat{\alpha}^K_{\tilde{k}} - \alpha^K_k)^2\right) 
    \leq& \En \left(\mathbbm{1}\{\gamma^K(Z)=k\}(\hat{\alpha}^K_{\hat{\gamma}^K(Z)} - \alpha^K_{k})^2\right)\\
    \leq&\En (\hat{\alpha}^K_{\hat{\gamma}^K(Z)} - \alpha^K_{\gamma^K(Z)})^2 \\
    =& \En (\hat{\alpha}^K_{\hat{\gamma}^K(Z)} - m_K(Z))^2 \\
    =& o_p(1),
\end{align*}
where the final equality follows from Lemma \ref{lemma:consistency_m_covariates}.

Next, I show that for all $\tilde{k}\in \{1,\ldots, K\}$, \begin{align}\label{eq:proof_consistent_alpha_2}
    \min_{k \in \{1, \ldots, K\}}\left(\hat{\alpha}^K_{\tilde{k}} - \alpha^K_k\right)^2 = o_p(1).
\end{align}

Define \begin{align*}
    \sigma_K(k) \equiv \argmin_{\tilde{k}\in\{1, \ldots, K\}} \left(\hat{\alpha}^K_{\tilde{k}} - \alpha^K_k\right)^2.
\end{align*}
By the triangle inequality, it holds that \begin{align*}
    \left\vert \hat{\alpha}^K_{\sigma_K(k)} - \hat{\alpha}^K_{\sigma_K(\tilde{k})}\right\vert &\geq \left\vert\alpha^K_k - \alpha^K_{\tilde{k}}\right\vert - \left\vert\hat{\alpha}^K_{\sigma_K(k)} - \alpha^K_k\right\vert - \left\vert\hat{\alpha}^K_{\sigma_K(\tilde{k})} - \alpha^K_{\tilde{k}}\right\vert
\end{align*}
where $\left\vert\hat{\alpha}^K_{\sigma_K(k)} - \alpha^K_k\right\vert$ and $\left\vert\hat{\alpha}^K_{\sigma_K(\tilde{k})} - \alpha^K_{\tilde{k}}\right\vert$ are $o_p(1)$ by the first result \eqref{eq:proof_consistent_alpha_1}, and $\left\vert\alpha^K_k - \alpha^K_{\tilde{k}}\right\vert > 0$ by Lemma \ref{lemma:mK_characteristics} (b). Thus $\sigma_K(k)\neq \sigma_K(\tilde{k})$ with probability approaching one, implying that the inverse $\sigma_K^{-1}$ is well-defined.

Now, with probability approaching one, we have \begin{align*}
    \min_{k\in \{1,\ldots, K\}} \left(\hat{\alpha}^K_{\tilde{k}} - \alpha^K_k\right)^2 &\leq \left(\hat{\alpha}^K_{\tilde{k}} - \alpha^K_{\sigma_K^{-1}(\tilde{k})}\right)^2\\
    &\overset{[1]}{=}\min_{h\in\{1,\ldots, K\}}\left(\hat{\alpha}^K_{h} - \alpha^K_{\sigma_K^{-1}(\tilde{k})}\right)^2\\
    &\overset{[2]}{=}o_p(1),
\end{align*}
where {\footnotesize [1]} follows from $\tilde{k}=\sigma_K(\sigma_K^{-1}(\tilde{k}))$, and  {\footnotesize [2]} follows from the first result \eqref{eq:proof_consistent_alpha_1}.

Combining \eqref{eq:proof_consistent_alpha_1} and \eqref{eq:proof_consistent_alpha_2} completes the proof.

\end{proof}


The proof of Lemma \ref{lemma:consistent_alpha} shows that for every $K\in \{2, \ldots, K_0\}$ there exists a permutation $\sigma_K:\{1,\ldots,K\}\to\{1,\ldots,K\}$ such that $ \left(\hat{\alpha}^K_{\sigma_K(k)} - \alpha^K_k\right)^2 = o_p(1).$ It is thus possible take $\sigma_K(k) = k$ by simply relabeling the elements of $\hat{\alpha}^K$. The remainder of the proof adopts this convention.

Further, let an $\eta$-neighborhood around a vector $\alpha^K$ be defined as $\mathcal{N}_{\alpha^K}(\eta) \equiv \{\alpha \in \mathcal{M}^{K}: \|\alpha - \alpha^K\| < \eta\}.$

\begin{lemma}\label{lemma:rate_g}
For $\eta>0$ small enough, we have for all $K\in \{2, \ldots, K_0\}$ and $\delta >0$\begin{align*}
        \sup_{(\alpha, \pi) \in \mathcal{N}_{(\alpha^K, \pi^0)}(\eta)} \En \mathbbm{1}\{\hat{\gamma}^K(Z; \alpha, \pi) \not= \gamma^K(Z)\} = o_p(n^{-\delta}).
    \end{align*}
\end{lemma}

\begin{proof}

Fix an arbitrary $K \in \{2, \ldots, K_0\}$.

From the definition of $\hat{\gamma}$ in \eqref{eq:definition_alphagamma_hat}, we have for all $k \in \{1,\ldots, K\}, z \in \mathcal{Z}$, \begin{align*}
   \mathbbm{1}\{\hat{\gamma}^K(z; \alpha, \pi) = k\} 
   \leq\mathbbm{1}\left\{\En \mathbbm{1}_z(Z)(Y - \alpha_k - X^\top \pi)^2 \leq \En \mathbbm{1}_z(Z)(Y - \alpha_{\gamma^K(z)}- X^\top \pi)^2\right\}.
\end{align*}
As a consequence, \begin{align}\label{eq:proof_g_rate_0}
\begin{aligned}
    &\En \mathbbm{1}\{\hat{\gamma}^K(Z; \alpha, \pi) \not= \gamma^K(Z)\}\\
    =& \sum_{k=1}^{K}\En \mathbbm{1}\{\gamma^K(Z) \not= k\} \mathbbm{1}\{\hat{\gamma}^K(Z; \alpha, \pi) = k\} \\
    \leq& \En \sum_{k=1}^{K}M^K_{Zk}(\alpha, \pi),
\end{aligned}
\end{align}
where for all $z\in \mathcal{Z}$ and $k \in \{1, \ldots, K\}$\begin{align*}
    M^K_{zk}(\alpha, \pi) &\equiv \mathbbm{1}\{\gamma^K(Z) \not= k\} \mathbbm{1}\left\{\En \mathbbm{1}_z(Z)(D - \alpha_k- X^\top \pi)^2 \leq \En \mathbbm{1}_z(Z)(D - \alpha_{\gamma^K(z)}- X^\top \pi)^2\right\}\\
    &=\mathbbm{1}\{\gamma^K(Z) \not= k\} \\ & \qquad \times \mathbbm{1}\left\{\En \mathbbm{1}_z(Z)\left(2D(\alpha_{\gamma^K(z)} - \alpha_{k}) - (\alpha_{\gamma^K(z)} - \alpha_k)(\alpha_{\gamma^K(z)} + \alpha_k + 2X^\top\pi)\right)\leq 0\right\}\\
    &=\mathbbm{1}\{\gamma^K(Z) \not= k\}  \\ & \qquad \times \mathbbm{1}\left\{\En\mathbbm{1}_z(Z)(\alpha_{\gamma^K(z)} - \alpha_k)\left(V_K + \alpha_{\gamma^K(z)}^K - \frac{(\alpha_{\gamma^K(z)} + \alpha_k)}{2} + X^\top(\pi^0 - \pi))\right)\leq 0\right\},
\end{align*}
where I have substituted for $D = \alpha^K_{\gamma^K(Z)} + X^\top\pi^0+ V_K$ and rearranged terms for simplification.

Now, whenever $(\alpha, \pi) \in \mathcal{N}_{(\alpha^K, \pi^0)}(\eta),$ it is possible to bound $M^K_{zk}(\alpha, \pi)$ by a quantity that does not depend on $(\alpha, \pi)$. In particular, note that \begin{align*}
    M^K_{zk}(\alpha, \pi) &\leq \max_{\tilde{k} \not = k} \mathbbm{1}\left\{\En \mathbbm{1}_z(Z) (\alpha_{\tilde{k}} - \alpha_k)\left(V_K + \alpha_{\tilde{k}}^K - \frac{ \alpha_{\tilde{k}} + \alpha_k}{2}+ X^\top(\pi^0 - \pi)\right)\leq 0\right\}.
    \end{align*}

Further, using that $V_K = V + m_0(Z) - m_K(Z)=V+\alpha^0_{\gamma^0(Z)} - \alpha^K_{\gamma^K(Z)}$, we have by simple application of the triangle inequality \begin{align}\label{eq:proof_g_rate_1}
\begin{aligned}
     &\bigg\vert\En \mathbbm{1}_z(Z) (\alpha_{\Tilde{k}} - \alpha_k)\left(V  + \alpha_{\tilde{k}}^K - \frac{\alpha_{\tilde{k}} + \alpha_k}{2}+ X^\top(\pi^0 - \pi)+ m_0(z) - m_K(z)\right) \\
     &\quad - \En \mathbbm{1}_z(Z)(\alpha^K_{\Tilde{k}} - \alpha^K_k)\left(V+ \alpha_{\tilde{k}}^K - \frac{ \alpha^K_{\tilde{k}} + \alpha^K_k}{2}\right) \bigg\vert\\
    \leq & \left\vert \En\mathbbm{1}_z(Z) \left[(\alpha_{\Tilde{k}} - \alpha_k) - (\alpha^K_{\Tilde{k}} - \alpha^K_k)\right]V \right\vert \\
    &\quad + \left\vert \En\mathbbm{1}_z(Z)\left[(\alpha_{\Tilde{k}} - \alpha_k)\left(\alpha_{\tilde{k}}^K - \frac{ \alpha_{\tilde{k}} + \alpha_k}{2}\right) - (\alpha^K_{\Tilde{k}} - \alpha^K_k)\left(\alpha_{\tilde{k}}^K - \frac{ \alpha^K_{\tilde{k}} + \alpha^K_k}{2}\right)\right] \right\vert\\
    &\quad + \left\vert \En\mathbbm{1}_z(Z) (\alpha_{\Tilde{k}} - \alpha_k)X^\top(\pi^0-\pi) \right\vert\\
    &\quad + \left\vert \En\mathbbm{1}_z(Z) (\alpha_{\Tilde{k}} - \alpha_k)\top(m_0(z) - m_K(z)) \right\vert.
\end{aligned}
\end{align}
For the first term in \eqref{eq:proof_g_rate_1}, it holds that \begin{align*}
     &\left\vert \En\mathbbm{1}_z(Z) \left[(\alpha_{\tilde{k}} - \alpha_k) - (\alpha^K_{\tilde{k}} - \alpha^K_k)\right]V \right\vert \\
     \overset{[1]}{\leq}& \left(\En\mathbbm{1}_z(Z)\right)\left\vert(\alpha_{\tilde{k}} - \alpha_k) - (\alpha^K_{\tilde{k}} - \alpha^K_k)\right\vert\left(\frac{1}{\En\mathbbm{1}_z(Z)}\En\mathbbm{1}_z(Z)V^2\right)^{\frac{1}{2}}\\
     \overset{[2]}{\leq}&2\left(\En\mathbbm{1}_z(Z)\right)\sqrt{\eta}\left(\frac{1}{\En\mathbbm{1}_z(Z)}\En\mathbbm{1}_z(Z)V^2\right)^{\frac{1}{2}},
\end{align*}
where {\footnotesize [1]} follows from Jensen's inequality, and {\footnotesize [2]} follows from $\alpha \in \mathcal{N}_{\alpha^K}(\eta)$. Similarly, for the second term in \eqref{eq:proof_g_rate_1}, it holds that
\begin{align*}
&\left\vert \En\mathbbm{1}_z(Z)\left[(\alpha_{\tilde{k}} - \alpha_k)\left(\alpha_{\tilde{k}}^K - \frac{ \alpha_{\tilde{k}} + \alpha_k}{2}\right) - (\alpha^K_{\tilde{k}} - \alpha^K_k)\left(\alpha_{\tilde{k}}^K - \frac{ \alpha^K_{\tilde{k}} + \alpha^K_k}{2}\right)\right] \right\vert\\
   =&\frac{1}{2}\left( \En\mathbbm{1}_z(Z)\right) \left\vert \left[\left(\alpha_{\tilde{k}} - \alpha_k\right)-\left(\alpha^K_{\tilde{k}} - \alpha^K_k\right)\right]\left(\alpha_{\tilde{k}}^K -   \alpha^K_k \right) + \left(\alpha_{\tilde{k}} - \alpha_k\right)\left[(\alpha^K_{\tilde{k}} + \alpha^K_k) - (\alpha_{\tilde{k}} + \alpha_k) \right] \right\vert\\
   \leq&\frac{1}{2}\left( \En\mathbbm{1}_z(Z)\right) \bigg(\left\vert \left[\left(\alpha_{\tilde{k}} - \alpha_k\right)-\left(\alpha^K_{\tilde{k}} - \alpha^K_k\right)\right]\left(\alpha_{\tilde{k}}^K -   \alpha^K_k \right) + \left(\alpha^K_{\tilde{k}} - \alpha^K_k\right)\left[(\alpha^K_{\tilde{k}} + \alpha^K_k) - (\alpha_{\tilde{k}} + \alpha_k) \right] \right\vert \\
   &\quad + \left\vert \left(\alpha_{\tilde{k}} - \alpha_k\right)\left[(\alpha^K_{\tilde{k}} + \alpha^K_k) - (\alpha_{\tilde{k}} + \alpha_k) \right] - \left(\alpha^K_{\tilde{k}} - \alpha^K_k\right)\left[(\alpha^K_{\tilde{k}} + \alpha^K_k) - (\alpha_{\tilde{k}} + \alpha_k) \right] \right\vert\bigg)\\
   \overset{[1]}{\leq}& \left( \En \mathbbm{1}_z(Z)\right) \left(\vert \alpha^K_{\tilde{k}} - \alpha^K_k\vert \sqrt{\eta} + 2\eta\right)\\
   \overset{[2]}{\leq}& \left( \En\mathbbm{1}_z(Z)\right) C_1\sqrt{\eta},
\end{align*}
with $C_1\equiv \vert \alpha^K_{\tilde{k}} - \alpha^K_k\vert + 2$ a constant independent of $\eta, n$ and $(\alpha, \pi)$, and where {\footnotesize [1]} follows from $\alpha \in \mathcal{N}_{\alpha^K}(\eta)$, and {\footnotesize [2]} follows from $\sqrt{\eta} \geq \eta$ for $\eta \in [0,1]$. Finally, using analogous arguments as above, the third term in \eqref{eq:proof_g_rate_1} can be bounded by \begin{align*}
    \left\vert \En\mathbbm{1}_z(Z) (\alpha_{\Tilde{k}} - \alpha_k)X^\top(\pi^0-\pi) \right\vert \leq \left( \En\mathbbm{1}_z(Z)\right) C_2 \sqrt{\eta} \left(\frac{1}{\En\mathbbm{1}_z(Z)}\En\mathbbm{1}_z(Z)\|X\|^2\right),
\end{align*}
where $C_2$ is a constant independent of $\eta, n$ and $(\alpha, \pi)$. And finally, the fourth term in  \eqref{eq:proof_g_rate_1} can be bounded by  \begin{align*}
   \left\vert \En\mathbbm{1}_z(Z) (\alpha_{\Tilde{k}} - \alpha_k)\top(m_0(z) - m_K(z)) \right\vert \leq \left( \En\mathbbm{1}_z(Z)\right) C_3 \sqrt{\eta} \|m_0 - m_K\|_{\infty},
\end{align*}
where $C_3$ is a constant independent of $\eta, n$ and $(\alpha, \pi)$.

Therefore,\begin{align}\label{eq:proof_g_rate_2}
\begin{aligned}
    M^K_{zk}(\alpha, \pi) & \leq \max_{\tilde{k} \not= k}\mathbbm{1}\bigg\{\En\mathbbm{1}_z(Z)(\alpha^K_{\tilde{k}} - \alpha^K_k)\left(V+ \alpha_{\tilde{k}}^K - \frac{ \alpha^K_{\tilde{k}} + \alpha^K_k}{2}\right) \\ 
    & \qquad \quad  \leq2\left(\En\mathbbm{1}_z(Z)\right)\sqrt{\eta}\left(\frac{1}{\En\mathbbm{1}_z(Z)}\En\mathbbm{1}_z(Z)V^2\right)^{\frac{1}{2}} +\left( \En\mathbbm{1}_z(Z)\right) C_1\sqrt{\eta}\\
    &\qquad \qquad + \left( \En\mathbbm{1}_z(Z)\right) C_2 \sqrt{\eta} \left(\frac{1}{\En\mathbbm{1}_z(Z)}\En\mathbbm{1}_z(Z)\|X\|^2\right) \\
    & \qquad \qquad + \left( \En\mathbbm{1}_z(Z)\right) C_3 \sqrt{\eta} \|m_0 - m_K\|_{\infty} \bigg\}\\
    & = \max_{\tilde{k} \not= k}\mathbbm{1}\bigg\{\frac{1}{\En\mathbbm{1}_z(Z)}\En\mathbbm{1}_z(Z)V(\alpha^K_{\tilde{k}} - \alpha^K_k)\\ 
    & \qquad \quad \leq 2\sqrt{\eta}\left(\frac{1}{\En\mathbbm{1}_z(Z)}\En\mathbbm{1}_z(Z)V^2\right)^{\frac{1}{2}} +C_1\sqrt{\eta}  \\
    &\qquad \qquad +  C_2 \sqrt{\eta} \left(\frac{1}{\En\mathbbm{1}_z(Z)}\En\mathbbm{1}_z(Z)\|X\|^2\right) \\
    & \qquad \qquad + C_3\sqrt{\eta} \|m_0 - m_K\|_{\infty} - \frac{1}{2}\left(\alpha^K_{\tilde{k}} - \alpha^K_k\right)^2\bigg\},
\end{aligned}
\end{align}
where the right-hand side does not depend on $(\alpha, \pi)$. 

Hence, for $\eta < 1$ we have \begin{align*}
    \sup_{(\alpha, \pi) \in \mathcal{N}_{(\alpha^K, \pi^0)}(\eta)} M^K_{zk}(\alpha, \pi) \leq \tilde{M}^K_{zk},
\end{align*}
where $\tilde{M}^K_{zk}$ denotes the final term in \eqref{eq:proof_g_rate_2}. Combining with \eqref{eq:proof_g_rate_0} then implies\begin{align*}
    &\sup_{(\alpha, \pi) \in \mathcal{N}_{(\alpha^K, \pi^0)}(\eta)}\En \mathbbm{1}\{\hat{\gamma}^K(Z; \alpha, \pi) \not= \gamma^K(Z)\} \leq \sum_{k=1}^{K} \En \Tilde{M}^K_{Zk} ,
\end{align*}
and therefore for any $\epsilon>0$ and $\delta>0$,\begin{align*}
    &\Pr\left(\sup_{(\alpha, \pi) \in \mathcal{N}_{(\alpha^K, \pi^0)}(\eta)}\En \mathbbm{1}\{\hat{\gamma}^K(Z; \alpha, \pi) \not= \gamma^K(Z)\}>\epsilon n^{-\delta}\right)\\
    \leq&\Pr\left(\sum_{k=1}^{K} \En \Tilde{M}^K_{Zk}>\epsilon n^{-\delta} \right) \\
    \overset{[1]}{\leq}&\frac{\E\left[ \sum_{k=1}^{K} \En \Tilde{M}^K_{Zk}\right]}{\epsilon n^{-\delta}}\\
    \leq&\frac{\sum_{k=1}^{K} \En \Pr(\Tilde{M}^K_{Zk} = 1)}{\epsilon n^{-\delta}}
\end{align*}
where {\footnotesize [1]} follows from Markov's inequality. It thus suffices to show that $\forall z \in \mathcal{Z}$, $k\in \{1, \ldots, K\}$, and $\delta>0$,\begin{align}\label{eq:final_eq_proof_g}
   \Pr\left(\Tilde{M}^K_{zk} = 1\right) = o(n^{-\delta}).
\end{align}

Note that by the exponential tail Assumption \ref{assumption:Pn_setup}, there exists a bounded constant $ \Tilde{L} < \infty$ such that for all $n\in \mathbbm{N}$, $P_n$ satisfies $\E[V^2\vert Z^{(n)}] \overset{a.s.}{<} \tilde{L}$ and $\E[\|X\|^2\vert Z^{(n)}] \overset{a.s.}{<} \tilde{L}$. Then, let  $L^K=\max\{\tilde{L}, \sqrt{\tilde{L}}, \|m_0 - m_K\|_\infty\}$ and consider 

\begin{align}\label{eq:proof_opdelta_Mtilde}
\begin{aligned}
    \Pr\left(\Tilde{M}^K_{zk}\right) &\overset{[1]}{\leq} \sum_{\Tilde{k} \neq k} \Pr\bigg(\frac{1}{\En\mathbbm{1}_z(Z)}\En\mathbbm{1}_z(Z)V(\alpha^K_{\tilde{k}} - \alpha^K_k)\\ 
    & \qquad \quad \leq 2\sqrt{\eta}\left(\frac{1}{\En \mathbbm{1}_z(Z)}\En\mathbbm{1}_z(Z)V^2\right)^{\frac{1}{2}} +C_1\sqrt{\eta}  \\
    & \qquad \qquad + C_2 \sqrt{\eta} \left(\frac{1}{\En\mathbbm{1}_z(Z)}\En\mathbbm{1}_z(Z)\|X\|^2\right) \\
    & \qquad \qquad + C_3\sqrt{\eta} \|m_0 - m_K\|_{\infty} - \frac{1}{2}\left(\alpha^K_{\tilde{k}} - \alpha^K_k\right)^2\bigg)\\
    &\overset{[2]}{\leq} \sum_{\Tilde{k} \neq k} \Pr\bigg(\frac{1}{\En\mathbbm{1}_z(Z)}\En\mathbbm{1}_z(Z)V(\alpha^K_{\tilde{k}} - \alpha^K_k) \leq \sqrt{\eta}(C_1 + L^K(2 + JC_2 + C_3))  - \frac{1}{2}\left(\alpha^K_{\tilde{k}} - \alpha^K_k\right)^2\bigg)\\
    & \qquad \quad  + K\Pr\left(\frac{1}{\En \mathbbm{1}_z(Z)}\En\mathbbm{1}_z(Z)V^2 > L^K\right) \\
    & \qquad \quad  + K\Pr\left(\frac{1}{\En \mathbbm{1}_z(Z)}\En\mathbbm{1}_z(Z)\|X\|^2 > JL^K\right)\\
    & \qquad \quad  + K\Pr\left( \|m_0 - m_K\|_\infty > L^K\right),
    \end{aligned}
\end{align}
where {\footnotesize [1]} follows from the union bound, and {\footnotesize [2]} follows from the triangle inequality. I now consider each term separately.

Focusing on the first term, fix $\eta\geq 0$ such that $\eta \leq \min\{1, \Tilde{\eta}^K\}$ where \begin{align*}
    \Tilde{\eta}^K <\left(\frac{c}{C_1 + L^K(2 + JC_2 + C_3)}\right)^2
\end{align*}
with $c$ defined by \eqref{eq:Z0_nonvanishing}. Denote \begin{align*}
    \Tilde{c}_{\Tilde{k}, k} \equiv \sqrt{\eta}(C_1 + L^K(2 + JC_2 + C_3))  - \frac{1}{2}\left(\alpha^K_{\tilde{k}} - \alpha^K_k\right)^2.
\end{align*}
Note that the choice of $\eta$ above implies that $\Tilde{c}_{\Tilde{k}, k}<0$ for all combinations $\Tilde{k}\neq k$. Further, fix $\Tilde{\lambda}>0$ such that $\tilde{\lambda}< \lambda_z$ as defined by Assumption \ref{assumption:rate_condition} and consider \begin{align}\label{eq:proof_opdelta_0}
\begin{aligned}
    &\Pr\bigg(\frac{1}{\sum_{i=1}^n\mathbbm{1}_z(Z_i)}\sum_{i=1}^n\mathbbm{1}_z(Z_i)V_i(\alpha^K_{\tilde{k}} - \alpha^K_g) \leq \Tilde{c}_{\Tilde{k}, k}\bigg)\\
    =&\Pr\bigg(\frac{1}{\sum_{i=1}^n\mathbbm{1}_z(Z_i)}\sum_{i=1}^n\mathbbm{1}_z(Z_i)V_i(\alpha^K_{\tilde{k}} - \alpha^K_k) \leq \Tilde{c}_{\Tilde{k}, k}\bigg\vert \sum_{i=1}^n\mathbbm{1}_z(Z_i) > n^{\tilde{\lambda}} \bigg)\Pr\left(\sum_{i=1}^n\mathbbm{1}_z(Z_i) > n^{\tilde{\lambda}}\right)\\
    &+\Pr\bigg(\frac{1}{\sum_{i=1}^n\mathbbm{1}_z(Z_i)}\sum_{i=1}^n\mathbbm{1}_z(Z_i)V_i(\alpha^K_{\tilde{k}} - \alpha^K_k) \leq \Tilde{c}_{\Tilde{k}, k}\bigg\vert \sum_{i=1}^n\mathbbm{1}_z(Z_i) \leq n^{\tilde{\lambda}} \bigg)\Pr\left(\sum_{i=1}^n\mathbbm{1}_z(Z_i) \leq n^{\tilde{\lambda}}\right)\\
    \leq &\Pr\bigg(\frac{1}{\sum_{i=1}^n\mathbbm{1}_z(Z_i)}\sum_{i=1}^n\mathbbm{1}_z(Z_i)V_i(\alpha^K_{\tilde{k}} - \alpha^K_k) \leq \Tilde{c}_{\Tilde{k}, k}\bigg\vert \sum_{i=1}^n\mathbbm{1}_z(Z_i) > n^{\tilde{\lambda}} \bigg)\\
    &+\Pr\left(\sum_{i=1}^n\mathbbm{1}_z(Z_i) \leq n^{\tilde{\lambda}}\right),
    \end{aligned}
\end{align}
where the inequality follows from probabilities being bounded by 1. For the first term, it holds for any $\delta >0$ that \begin{align}\label{eq:proof_opdelta_1}
\begin{aligned}
    &\Pr\bigg(\frac{1}{\sum_{i=1}^n\mathbbm{1}_z(Z_i)}\sum_{i=1}^n\mathbbm{1}_z(Z_i)V_i(\alpha^K_{\tilde{k}} - \alpha^K_k) \leq \Tilde{c}_{\Tilde{k}, k}\: \bigg\vert\sum_{i=1}^n\mathbbm{1}_z(Z_i) > n^{\tilde{\lambda}} \bigg)\\
    \overset{[1]}{=}&\Pr\bigg(\frac{1}{N_z}\sum_{i=1}^{N_z}V_{iz}(\alpha^K_{\tilde{k}} - \alpha^K_k) \leq \Tilde{c}_{\Tilde{k}, k}\: \bigg\vert N_z > n^{\tilde{\lambda}} \bigg)\\
     \overset{[2]}{\leq} &\Pr\bigg(\left\vert\frac{1}{N_z}\sum_{i=1}^{N_z}V_{iz}\right\vert \geq \frac{\vert \Tilde{c}_{\Tilde{k}, k}\vert }{\vert \alpha^K_{\tilde{k}} - \alpha^K_k\vert} \:\bigg\vert N_z > n^{\tilde{\lambda}} \bigg)\\
    \overset{[3]}{\leq} & \Pr\bigg(\left\vert\frac{1}{\lfloor n^{\tilde{\lambda}}\rfloor}\sum_{i=1}^{\lfloor n^{\tilde{\lambda}}\rfloor}V_{iz} \right\vert \geq \frac{\vert \Tilde{c}_{\Tilde{k}, k}\vert }{\vert \alpha^K_{\tilde{k}} - \alpha^K_k\vert}  \: \bigg)\\
    \overset{[4]}{=}&o(n^{-\delta}),
\end{aligned}
\end{align}
where {\footnotesize [1]} takes $V_{iz} \equiv (V_i\vert Z_i=z)$ and $N_z \equiv \sum_{i=1}^n\mathbbm{1}_z(Z_i)$, {\footnotesize [2]} follows from $\Tilde{c}_{\Tilde{k}, k}<0$ and the fact that $\vert \alpha^K_{\tilde{k}} - \alpha^K_k\vert>0, \forall \Tilde{k}\neq k$ by Lemma \ref{lemma:mK_characteristics} (b), and {\footnotesize [3]} follows from $\Pr(\vert\frac{1}{n}\sum_{i=1}^n V_{iz}\vert \geq b ) \leq \Pr(\vert\frac{1}{\Tilde{n}}\sum_{i=1}^{\Tilde{n}} V_{iz}\vert \geq b), \forall \Tilde{n} \leq n, b>0$, and $N_z\independent (\sum_{i=1}^{\lfloor n^{\tilde{\lambda}}\rfloor}V_{iz})$. Finally, {\footnotesize [4]} follows by application of Lemma B.5 in \citet{bonhomme2015grouped} where I take $T\equiv \lfloor n^{\tilde{\lambda}}\rfloor$, $z_t \equiv V_{iz}$, and $z \equiv \frac{\vert \Tilde{c}_{\Tilde{k}, k}\vert }{\vert \alpha^K_{\tilde{k}} - \alpha^K_k\vert}.$\footnote{Note that Lemma B.5 implies the rate $o(\lfloor n^{(\tilde{\lambda})}\rfloor^{-\delta})$, but since  this holds for any $\delta >0$ and $\tilde{\lambda} >0$ is a fixed constant, the stated result follows.} The lemma applies by Assumption \ref{assumption:Pn_setup} (a) and Assumption \ref{assumption:asymptotics_setup} (d).

To bound the second term in \eqref{eq:proof_opdelta_0}, I prove a simple concentration inequality in Lemma \ref{lemma:binom_bound}, whose application implies for any $\delta >0$, \begin{align}\label{eq:proof_opdelta_2}
    \Pr\left(\sum_{i=1}^n\mathbbm{1}_z(Z_i) \leq n^{a_z\tilde{\lambda}}\right) = o(n^{-\delta}), \quad \forall z \in \mathcal{Z}.
\end{align}


\begin{lemma}\label{lemma:binom_bound}
    Let $X_n$ be a Binomial random variable with $n$ trials and success probability $p_n = a n^{\lambda - 1}$ for fixed $a> 0$ and $\lambda \in (0,1)$. Then, $\forall \delta >0$ and $\Tilde{\lambda}>0: \lambda > \Tilde{\lambda},$\begin{align*}
        \Pr\left(X_n \leq an^{\tilde{\lambda}}\right) = o(n^{-\delta}).
    \end{align*}
\end{lemma}

\begin{proof}
    By Chernoff's inequality,\begin{align*}
        &\Pr\left(X_n \leq an^{\tilde{\lambda}}\right)\\
        \leq &\exp\left\{-n\left[an^{\tilde{\lambda} -1}\log\left(\frac{an^{\tilde{\lambda} - 1}}{an^{\lambda - 1}}\right) + \left(1-an^{\tilde{\lambda} -1}\right)\log\left(\frac{1 - an^{\tilde{\lambda} - 1}}{1 - an^{\lambda - 1}}\right)  \right]\right\}\\
        =&\exp\left\{-n\left[ -an^{\Tilde{\lambda} - 1}(\lambda-\tilde{\lambda})\log(n) + \left(1-an^{\tilde{\lambda} -1}\right)\left(\log\left(1-an^{\tilde{\lambda} -1}\right) - \log\left(1 - an^{\lambda - 1}\right)\right) \right]\right\}.
    \end{align*}
It is possible to bound the terms involving the logarithm via the following simple inequalities:\begin{align*}
    \log(n) \leq \gamma(n^{1/\gamma} - 1), \quad &\forall n, \gamma >0,\\
    \frac{-x}{1-x} \leq \log(1-x) \leq -x, \quad &\forall x \in [0, 1).
\end{align*}
Therefore, fixing $\gamma>0: \lambda > \tilde{\lambda} + \frac{1}{\gamma}$, \begin{align*}
     &\Pr\left(X_n \leq an^{\tilde{\lambda}}\right)\\
        \leq & \exp\left\{-n\left[ -an^{\Tilde{\lambda} - 1}(\lambda-\tilde{\lambda})\gamma\left(n^{1/\gamma} - 1\right) + \left(1-an^{\tilde{\lambda} -1}\right)\left( an^{\lambda - 1} - \frac{an^{\Tilde{\lambda} - 1}}{1-an^{\Tilde{\lambda} - 1}} \right) \right]\right\}\\
        =&\exp\left\{ -\left[an^{\lambda} + an^{\tilde{\lambda}}(\lambda - \Tilde{\lambda})\gamma - an^{\Tilde{\lambda} + 1/\gamma}(\lambda - \Tilde{\lambda})\gamma - a^2n^{\lambda + \Tilde{\lambda} -1} - an^{\tilde{\lambda}}\right] \right\}\\
        =&\exp\left\{ -n^{\lambda}\left[a + an^{\tilde{\lambda} - \lambda}(\lambda - \Tilde{\lambda})\gamma - an^{\Tilde{\lambda} + 1/\gamma - \lambda}(\lambda - \Tilde{\lambda})\gamma - a^2n^{\Tilde{\lambda} -1} - an^{\tilde{\lambda} - \lambda}\right] \right\},
\end{align*}
where the term in brackets tends to $a>0$ as $n\to\infty$. Finally, note that for any $\delta, \lambda >0$,\begin{align*}
    \exp\left\{-n^{\lambda}\right\} = o(n^{-\delta}),
\end{align*}
which completes the proof.

\end{proof}

Combining \eqref{eq:proof_opdelta_0}-\eqref{eq:proof_opdelta_2} then implies for any $\delta > 0$, \begin{align}\label{eq:proof_opdelta_3}
    \Pr\bigg(\frac{1}{\sum_{i=1}^n\mathbbm{1}_z(Z_i)}\sum_{i=1}^n\mathbbm{1}_z(Z_i)V_i(\alpha^K_{\tilde{k}} - \alpha^K_k) \leq \Tilde{c}_{\Tilde{k}, k}\bigg) = o(n^{-\delta}).
\end{align}

Focusing now on the second term in \eqref{eq:proof_opdelta_Mtilde} and following similar arguments as before, we have for any $\delta >0$\begin{align*}
    &\Pr\left(\frac{1}{\sum_{i=1}^n\mathbbm{1}_z(Z_i)}\sum_{i=1}^n\mathbbm{1}_z(Z_i)V_i^2 > L^K\right)\\
    \leq & \Pr\left(\frac{1}{\lfloor n^{\lambda - \tilde{\lambda}}\rfloor}\sum_{i=1}^{\lfloor n^{\lambda - \tilde{\lambda}}\rfloor}V_{iz}^2 > L^K\right) + \Pr\left(\sum_{i=1}^n\mathbbm{1}_z(Z_i) \leq n^{\lambda - \tilde{\lambda}} \right)\\
    \leq&\Pr\left(\left\vert\frac{1}{\lfloor n^{\lambda - \tilde{\lambda}}\rfloor}\sum_{i=1}^{\lfloor n^{\lambda - \tilde{\lambda}}\rfloor}V_{iz}^2 - \E V_{z}^2\right\vert > L^K- \E V_{z}^2\right) + o(n^{-\delta}).
\end{align*}
We can again apply Lemma B.5 in \citet{bonhomme2015grouped} where I take $T\equiv \lfloor n^{\tilde{\lambda}}\rfloor$, $z_t \equiv V^2_{iz} - \E V^2_{z}$, and $z \equiv L^K - \E V^2_{z}$. Note that $L^K - \E V^2_{z}>0$ is implied by Assumption \ref{assumption:Pn_setup} (a). As a consequence, for any $\delta >0$\begin{align}\label{eq:proof_opdelta_4}
    \Pr\left(\frac{1}{\sum_{i=1}^n\mathbbm{1}_z(Z_i)}\sum_{i=1}^n\mathbbm{1}_z(Z_i)V_i^2 > L^K\right) = o(n^{-\delta}).
\end{align}

Similarly for the third term in \eqref{eq:proof_opdelta_Mtilde} we have for any $\delta > 0$\begin{align*}
    &\Pr\left(\frac{1}{\sum_{i=1}^n\mathbbm{1}_z(Z_i)}\sum_{i=1}^n\mathbbm{1}_z(Z_i)\|X_i\|^2 > JL^K\right)\\
    \leq & \Pr\left(\frac{1}{\lfloor n^{\lambda - \tilde{\lambda}}\rfloor}\sum_{i=1}^{\lfloor n^{\lambda - \tilde{\lambda}}\rfloor}\|X_{iz}\|^2 > JL^K\right) + o(n^{-\delta})\\
    \leq & \Pr\left(\exists j \in \{1,\ldots, J\}:\: \frac{1}{\lfloor n^{\lambda - \tilde{\lambda}}\rfloor}\sum_{i=1}^{\lfloor n^{\lambda - \tilde{\lambda}}\rfloor}X_{ijz}^2 >  L^K\right) + o(n^{-\delta})\\
    \leq & \sum_{j=1}^J \Pr\left(\frac{1}{\lfloor n^{\lambda - \tilde{\lambda}}\rfloor}\sum_{i=1}^{\lfloor n^{\lambda - \tilde{\lambda}}\rfloor}X_{ijz}^2 > L^K \right) + o(n^{-\delta})\\
    \leq& \sum_{j=1}^J \Pr\left(\left\vert \frac{1}{\lfloor n^{\lambda - \tilde{\lambda}}\rfloor}\sum_{i=1}^{\lfloor n^{\lambda - \tilde{\lambda}}\rfloor}X_{ijz}^2 - \E X_{jz}^2\right\vert > L^K - \E X_{jz}^2\right)
    + o(n^{-\delta}).
\end{align*}
We can then again apply Lemma B.5 in \citet{bonhomme2015grouped} where I take $T\equiv \lfloor n^{\tilde{\lambda}}\rfloor$, $z_t \equiv X_{ijz}^2 - \E X_{jz}^2$, and $z \equiv  L^K - \E X_{jz}^2$. Note that  $L^K - \E X_{jz}^2>0$ is implied by Assumption \ref{assumption:Pn_setup} (b). The lemma applies by Assumption \ref{assumption:Pn_setup} (b) and Assumption \ref{assumption:asymptotics_setup} (d). As a consequence, for any $\delta >0$\begin{align}\label{eq:proof_opdelta_5}
    \Pr\left(\frac{1}{\sum_{i=1}^n\mathbbm{1}_z(Z_i)}\sum_{i=1}^n\mathbbm{1}_z(Z_i)\|X_i\|^2 > JL^K\right) = o(n^{-\delta}).
\end{align}

Finally, not that for the fourth term in \eqref{eq:proof_opdelta_Mtilde} we simply have $\Pr\left( \|m_0 - m_K\|_\infty > L^K\right) = 0$ by definition of $L^K$. Combining with \eqref{eq:proof_opdelta_Mtilde}, \eqref{eq:proof_opdelta_3}, \eqref{eq:proof_opdelta_4}, and  \eqref{eq:proof_opdelta_5} shows \eqref{eq:final_eq_proof_g} and thus completes the proof.

\end{proof}

Next, I show that $\hat{m}_K$ converges at exponential rate to the infeasible least squares estimator of $D - X^\top \hat{\pi}$ on the set of indicators $\{\mathbbm{1}_k(\gamma^K(Z)\}_{k=1}^K$. In particular, for $K\in \mathbbm{N}$, define 
\begin{align}\label{eq:definition_alpha_tilde}
    \tilde{\alpha}^K \equiv \argmin_{\alpha \in \mathcal{M}^K } \: \En(D - X^\top \hat{\pi} - \alpha_{\gamma^K(Z)} )^2, 
\end{align}
with $\hat{\pi}$ being a consistent first-step estimator for $\theta_0$.

\begin{lemma}\label{lemma:rate_hat_tilde}
    Let the assumptions of Theorem \ref{theorem:CIV_w_covariates_mk} hold. Then, for all $K\in\{2, \ldots, K_0\}$ and $\delta > 0$, \begin{align*}
        \En (\hat{\alpha}^K_{\hat{\gamma}^K(Z)} - \tilde{\alpha}^K_{\gamma^K(Z)})^2 = o_p(n^{-\delta}).
    \end{align*} 
\end{lemma}

\begin{proof}

Fix an arbitrary $K\in \{2, \ldots, K_0\}$. Define \begin{align*}
    \Bar{Q}(\alpha, \pi) \equiv \En (D - X^\top \pi - \alpha_{\gamma^K(Z)})^2, \qquad \text{and} \qquad \hat{Q}(\alpha, \pi) \equiv \En (D - X^\top \pi - \alpha_{\hat{\gamma}^K(Z; \alpha, \pi)})^2.
\end{align*}

For $\eta$ satisfying the condition of Lemma \ref{lemma:rate_g}, it holds for any $\delta >0$ that \begin{align}\label{eq:uniformconvergence_new}
\begin{aligned}
    &\sup_{(\alpha, \pi) \in \mathcal{N}_{(\alpha^K, \pi^0)}(\eta)} \: \left\vert  \Bar{Q}(\alpha, \pi) - \hat{Q}(\alpha, \pi)  \right\vert \\
    =&\sup_{(\alpha, \pi) \in \mathcal{N}_{(\alpha^K, \pi^0)}(\eta)} \: \bigg\vert \En \mathbbm{1}\{\hat{\gamma}^K(Z; \alpha, \pi) \not= \gamma^K(Z)\} (D - X^\top \pi - \alpha_{\hat{\gamma}^K(Z; \alpha, \pi)})^2 \\
    &\qquad + \En \mathbbm{1}\{\hat{\gamma}^K(Z; \alpha, \pi) = \gamma^K(Z)\} (D - X^\top \pi - \alpha_{\gamma^K(Z)})^2 - \En (D - X^\top \pi - \alpha_{\gamma^K(Z)})^2\bigg\vert
    \\
    =&\sup_{(\alpha, \pi) \in \mathcal{N}_{(\alpha^K, \pi^0)}(\eta)} \: \bigg\vert \En \mathbbm{1}\{\hat{\gamma}^K(Z; \alpha, \pi) \neq \gamma^K(Z)\} \\
    &\qquad \times \left[(D - X^\top \pi - \alpha_{\hat{\gamma}^K(Z; \alpha, \pi)})^2 - (D - X^\top \pi - \alpha_{\gamma^K(Z)})^2\right] \bigg\vert\\
    \leq&\sup_{(\alpha, \pi) \in \mathcal{N}_{(\alpha^K, \pi^0)}(\eta)} \:  \En \mathbbm{1}\{\hat{\gamma}^K(Z; \alpha, \pi) \neq \gamma^K(Z)\} \\
    &\qquad \times \left\vert\left[(D - X^\top \pi - \alpha_{\hat{\gamma}^K(Z; \alpha, \pi)})^2 - (D - X^\top \pi - \alpha_{\gamma^K(Z)})^2\right] \right\vert\\
    \overset{[1]}{\leq}&\sup_{(\alpha, \pi) \in \mathcal{N}_{(\alpha^K, \pi^0)}(\eta)} \:  \left(\En \mathbbm{1}\{\hat{\gamma}^K(Z; \alpha, \pi) \neq \gamma^K(Z)\}\right)^{1/2} \\
    &\qquad \times \left(\En \left\vert\left[(D - X^\top \pi - \alpha_{\hat{\gamma}^K(Z; \alpha, \pi)})^2 - (D - X^\top \pi - \alpha_{\gamma^K(Z)})^2\right] \right\vert^2\right)^{1/2}\\
    \leq&  \left(\sup_{(\alpha, \pi) \in \mathcal{N}_{(\alpha^K, \pi^0)}(\eta)} \: \En \mathbbm{1}\{\hat{\gamma}^K(Z; \alpha, \pi) \neq \gamma^K(Z)\}\right)^{1/2} \\
    &\qquad \times \left(\sup_{(\alpha, \pi) \in \mathcal{N}_{(\alpha^K, \pi^0)}(\eta)} \: \En \left\vert 2(D - X^\top \pi)(\alpha_{\gamma^K(Z)} - \alpha_{\hat{\gamma}^K(Z; \alpha, \pi)}) + \alpha_{\hat{\gamma}^K(Z; \alpha, \pi)}^2 - \alpha_{\gamma^K(Z)}^2 \right\vert^2\right)^{1/2}\\
    \overset{[2]}{=}& o_p(n^{-\delta}),
\end{aligned}
\end{align}
where {\footnotesize [1]} follows from Cauchy-Schwarz, and {\footnotesize [2]} follows from Lemma \ref{lemma:rate_g} and the fact that the second term is $O_p(1)$ under Assumption \ref{assumption:Pn_setup} and Assumption \ref{assumption:asymptotics_setup} (b).

Now, for both $\hat{\alpha}^K$ and the infeasible least squares coefficients $\tilde{\alpha}^K$, it holds for any $\delta>0$ that \begin{align}\label{eq:rate_Q}
    \Bar{Q}(\hat{\alpha}^K, \hat{\pi}) - \hat{Q}(\hat{\alpha}^K, \hat{\pi}) = o_p(n^{-\delta}),\qquad
    \Bar{Q}(\tilde{\alpha}^K, \hat{\pi}) - \hat{Q}(\tilde{\alpha}^K, \hat{\pi}) = o_p(n^{-\delta}).
\end{align}
To see this, fix $\varepsilon >0$ and consider \begin{align*}
    &\Pr\left(\left\vert\Bar{Q}(\hat{\alpha}^K, \hat{\pi}) - \hat{Q}(\hat{\alpha}^K, \hat{\pi})\right\vert>\varepsilon n^{-\delta}\right)\\
\overset{[1]}{=}&\Pr\left(\left\vert\Bar{Q}(\hat{\alpha}^K, \hat{\pi}) - \hat{Q}(\hat{\alpha}^K, \hat{\pi})\right\vert>\varepsilon n^{-\delta}\big\vert (\hat{\alpha}^K,\hat{\pi}) \in \mathcal{N}_{(\alpha^K, \pi^0)}(\eta)\right)\Pr\left((\hat{\alpha}^K,\hat{\pi}) \in \mathcal{N}_{(\alpha^K, \pi^0)}(\eta)\right) \\
&\qquad + \Pr\left(\left\vert\Bar{Q}(\hat{\alpha}^K,\hat{\pi}) - \hat{Q}(\hat{\alpha}^K,\hat{\pi})\right\vert>\varepsilon n^{-\delta}\big\vert (\hat{\alpha}^K,\hat{\pi}) \not\in \mathcal{N}_{(\alpha^K, \pi^0)}(\eta)\right)\Pr\left((\hat{\alpha}^K,\hat{\pi}) \not\in \mathcal{N}_{(\alpha^K, \pi^0)}(\eta)\right)\\
\overset{[2]}{\leq}&\Pr\left(\left\vert\Bar{Q}(\hat{\alpha}^K,\hat{\pi}) - \hat{Q}(\hat{\alpha}^K,\hat{\pi})\right\vert>\varepsilon n^{-\delta}\big\vert (\hat{\alpha}^K,\hat{\pi}) \in \mathcal{N}_{(\alpha^K, \pi^0)}(\eta)\right) + \Pr\left((\hat{\alpha}^K,\hat{\pi}) \not\in \mathcal{N}_{(\alpha^K, \pi^0)}(\eta)\right)\\
\overset{[3]}{\leq} & \Pr\left(\sup_{(\alpha, \pi) \in \mathcal{N}_{(\alpha^K, \pi^0)}(\eta)}\left\vert\Bar{Q}(\alpha, \pi) - \hat{Q}(\alpha, \pi)\right\vert>\varepsilon n^{-\delta} \right) + o(1)\\
\overset{[4]}{=}&o(1),
\end{align*}
where {\footnotesize [1]} follows from the law of total probability, {\footnotesize [2]} follows from probabilities being bounded by one, {\footnotesize [3]} follows from consistency of $\hat{\alpha}^K$ by Lemma \ref{lemma:consistent_alpha} and consistency of $\hat{\pi}$ by Assumption \ref{assumption:asymptotics_setup} (c), and {\footnotesize [4]} follows from \eqref{eq:uniformconvergence_new}. The arguments for the infeasible least squares coefficients are analogous. 

As a consequence, \begin{align}\label{eq:diff_in_Q}
    0 \leq& \Bar{Q}(\hat{\alpha}^K,\hat{\pi}) - \Bar{Q}(\tilde{\alpha}^K,\hat{\pi}) = \hat{Q}(\hat{\alpha}^K,\hat{\pi}) - \hat{Q}(\tilde{\alpha}^K,\hat{\pi}) + o_p(n^{-\delta}) 
    \leq o_p(n^{-\delta}),
\end{align}
where the inequalities follow from the definition of $\hat{\alpha}^K$ and $\tilde{\alpha}^K$ (minimizing $\hat{Q}$ and $\Bar{Q}$, respectively), and the equality follows from \eqref{eq:rate_Q}.

Note further that\begin{align}\label{eq:diff_in_alpha}
\begin{aligned}
   &\Bar{Q}(\hat{\alpha}^K, \hat{\pi}) - \Bar{Q}(\tilde{\alpha}^K, \hat{\pi}) \\
   =& \En\left(\hat{\alpha}^K_{\gamma^K(Z)} - \tilde{\alpha}^K_{\gamma^K(Z)}\right)^2 + 2\En \left(D - X^\top \hat{\pi} - \tilde{\alpha}^K_{\gamma^K(Z)}\right)\left(\hat{\alpha}^K_{\gamma^K(Z)} - \tilde{\alpha}^K_{\gamma^K(Z)}\right)\\
   =&\En\left(\hat{\alpha}^K_{\gamma^K(Z)} - \tilde{\alpha}^K_{\gamma^K(Z)}\right)^2,
\end{aligned}
\end{align}
where the equality follows from \begin{align*}
    &\En \left(D - X^\top \hat{\pi} - \tilde{\alpha}^K_{\gamma^K(Z)}\right)\left(\hat{\alpha}^K_{\gamma^K(Z)} - \tilde{\alpha}^K_{\gamma^K(Z)}\right)\\
    =&\En\sum_{k=1}^{K}\mathbbm{1}\{\gamma^K(Z) = k\}\left(D - X^\top \hat{\pi} - \tilde{\alpha}^K_{k}\right)\left(\hat{\alpha}^K_{k} - \tilde{\alpha}^K_{k}\right)\\
	=&\sum_{k=1}^{K}\left(\hat{\alpha}^K_{k} - \tilde{\alpha}^K_{k}\right)\En\mathbbm{1}\{\gamma^K(Z) = k\}\left(D - X^\top \hat{\pi} - \tilde{\alpha}^K_{k}\right)\\
	=&\sum_{k=1}^{K}\left(\hat{\alpha}^K_{k} - \tilde{\alpha}^K_{k}\right)\times 0
\end{align*}
because the infeasible least squares coefficients $(\tilde{\alpha}^K_k)_{k=1}^{K}$ correspond to the sample average of $D - X^\top\hat{\pi}$ with $\gamma^K(Z) = k$. Combining \eqref{eq:diff_in_Q} and \eqref{eq:diff_in_alpha} then implies \begin{align}\label{eq:diff_alpha_rate}
    \En\left(\hat{\alpha}^K_{\gamma^K(Z)} - \tilde{\alpha}^K_{\gamma^K(Z)}\right)^2 = o_p(n^{-\delta}).
\end{align}

Now, for any $\delta >0$,\begin{align}\label{eq:diff_m_rate}
\begin{aligned}
&\En\left(\hat{\alpha}^K_{\hat{\gamma}^K(Z)} - \tilde{\alpha}^K_{\gamma^K(Z)}\right)^2 \\
     =& \En\mathbbm{1}\{\hat{\gamma}^K(Z; \hat{\alpha}^K, \hat{\pi}) \not= \gamma^K(Z)\}\left(\hat{\alpha}^K_{\hat{\gamma}^K(Z; \alpha, \pi)} - \tilde{\alpha}^K_{\gamma^K(Z)}\right)^2 \\
     &\quad + \En\mathbbm{1}\{\hat{\gamma}^K(Z; \hat{\alpha}^K, \hat{\pi}) = \gamma^K(Z)\}\left(\hat{\alpha}^K_{\gamma^K(Z)} - \tilde{\alpha}^K_{\gamma^K(Z)}\right)^2\\
     \overset{[1]}{\leq}& \left(\En \mathbbm{1}\{\hat{\gamma}^K(Z; \hat{\alpha}^K, \hat{\pi}) \not= \gamma^K(Z)\}\right)^{1/2}\left(\En \left(\hat{\alpha}^K_{\hat{\gamma}^K(Z; \hat{\alpha}^K, \hat{\pi})} - \tilde{\alpha}^K_{\gamma^K(Z)}\right)^4\right)^{1/2} + o_p(n^{-\delta})
\end{aligned}
\end{align}
where {\footnotesize [1]} follows from Cauchy-Schwarz and \eqref{eq:diff_alpha_rate}. By Assumption \ref{assumption:asymptotics_setup} (b), $\En\left(\hat{\alpha}^K_{\hat{\gamma}^K(Z; \hat{\alpha}^K, \hat{\pi})} - \tilde{\alpha}^K_{\gamma^K(Z)}\right)^4 = O_p(1)$. Finally, taking $\eta$ satisfying the condition of Lemma \ref{lemma:rate_g} and fixing $\varepsilon>0$, it holds that\begin{align*}
    &\Pr\left(\En \mathbbm{1}\{\hat{\gamma}^K(Z; \hat{\alpha}^K, \hat{\pi}) \not= \gamma^K(Z)\}>\varepsilon n^{-\delta}\right)\\
    \overset{[1]}{\leq}&\Pr\left(\En \mathbbm{1}\{\hat{\gamma}^K(Z; \hat{\alpha}^K, \hat{\pi}) \not= \gamma^K(Z)\}>\varepsilon n^{-\delta}\vert (\hat{\alpha}^K,\hat{\pi}) \in \mathcal{N}_{(\alpha^K, \pi^0)}(\eta)\right) + \Pr\left((\hat{\alpha}^K,\hat{\pi}) \not\in \mathcal{N}_{(\alpha^K, \pi^0)}(\eta)\right)\\
    \overset{[2]}{\leq} & \Pr\left(\sup_{(\alpha, \pi) \in \mathcal{N}_{(\alpha^K, \pi^0)}(\eta)} \En \mathbbm{1}\{\hat{\gamma}^K(Z; \alpha, \pi) \not= \gamma^K(Z)\}\right) + o(1) \\
    \overset{[3]}{\leq}& o(1),
\end{align*} 
where {\footnotesize [1]} follows from the law of total probability and bounding probabilities by one, {\footnotesize [2]} follows from Lemma \ref{lemma:consistent_alpha} and Assumption \ref{assumption:asymptotics_setup} (c), and {\footnotesize [3]} follows from Lemma \ref{lemma:rate_g}. Combining with \eqref{eq:diff_m_rate} then completes the proof.

\end{proof}


\subsection{Asymptotic Distribution of $\hat{\theta}_K$}
It is now possible to proof Theorem \ref{theorem:CIV_w_covariates_mk}. As before, fix an arbitrary $K \in \{2, \ldots, K_0 \}$.

Note that \begin{align}\label{eq:diff_F}
    \hat{F}_K - F_K = \left(\Delta_n(Z, X),\: \mathbf{0}_J^\top\right)^\top,
\end{align}
where \begin{align*}
    \Delta_n(Z, X) \equiv \hat{m}_K(Z) - m_K(Z) + X^\top (\hat{\pi} - \pi_0)
\end{align*}
so that it suffices to focus on the first component of the difference $\hat{F}_K - F_K$.

By Assumption \ref{assumption:iv_setup} (a), we have \begin{align}\label{eq:Gn_term}
    \sqrt{n}\left(\left(\En\hat{F}_K W^\top\right)^{-1} \En \hat{F}_K Y - \theta_0\right) = \left(\En\hat{F}_K W^\top\right)^{-1} \sqrt{n}\En \hat{F}_K U.
\end{align}

For the first component of the first term, note that \begin{align}\label{eq:Deltan_W_op1}
\begin{aligned}
\left\|\En\Delta_n(Z, X)W^\top\right\| & \overset{[1]}{\leq} \left\|\En(\hat{m}_K(Z) - m_K(Z))W^\top\right\| + \left\|(\hat{\pi} - \pi_0)^\top \En XW^\top\right\|\\
    & \overset{[2]}{\leq} (\En (\hat{m}_K(Z) - m_K(Z))^2)^{1/2}(\En \|W\|^2)^{1/2} +\sum_{j=1}^J\left\vert\hat{\pi}_j - \pi_{j0}\right\vert \left\|\En X_jW\right\|\\
    & \overset{[3]}{=} O_p(1)\left(\En (\hat{m}_K(Z) - m_K(Z))^2)^{1/2} + \sum_{j=1}^J\left\vert\hat{\pi}_j - \pi_{j0}\right\vert\right)\\
    & \overset{[4]}{=} o_p(1),
\end{aligned}
\end{align}
where {\footnotesize [1]} follows from the triangle inequality, {\footnotesize [2]} follows from Cauchy-Schwarz, {\footnotesize [3]} follows from Assumption \ref{assumption:Pn_setup} and Assumption \ref{assumption:asymptotics_setup} (b), and {\footnotesize [4]} follows from Lemma \ref{lemma:consistency_m_covariates} and Assumption \ref{assumption:asymptotics_setup} (c). Hence, \begin{align*}
    \En\hat{F}_K W^\top = \En F_KW^\top + o_p(1) = \E F_KW^\top + o_p(1) 
\end{align*}
and consequently \begin{align}\label{eq:FKW_inv_consistency}
    \left(\En\hat{F}_K W^\top\right)^{-1} = \left( \E F_KW^\top\right)^{-1}+ o_p(1)
\end{align}
by the assumption of non-singularity of $\E F_KW^\top$ (as assumed in the statement of the theorem).

Let $\Tilde{m}_K(Z) \equiv \Tilde{\alpha}^K_{\gamma^K(Z)}$ be the infeasible least squares estimator defined in \eqref{eq:definition_alpha_tilde}. For the first component of the second term in \eqref{eq:Gn_term}, note that \begin{align*}
    \left\vert\sqrt{n}\En\Delta_n(Z, X)U\right\vert & \overset{[1]}{\leq} \left\vert \sqrt{n}\En (\hat{m}_K(Z) - \Tilde{m}_K(Z))U \right\vert +\left\vert \sum_{k=1}^{K} (\Tilde{\alpha}^K_k - \alpha^K_k)\Gn \mathbbm{1}_k(\gamma^K(Z))U\right\vert \\
    & \overset{[2]}{=}  \left\vert\sqrt{n}\En(\hat{m}_K(Z) - \Tilde{m}_K(Z))U \right\vert + \left(\max_{k \in \{1, \ldots, K\}} \vert \Tilde{\alpha}^K_k - \alpha^K_k\vert \right) \left\vert \sum_{k=1}^K \Gn \mathbbm{1}_k(\gamma^K(Z))U\right\vert\\
    & \overset{[3]}{=}  \left\vert\sqrt{n}\En(\hat{m}_K(Z) - \Tilde{m}_K(Z))U \right\vert + \left(\max_{k \in \{1, \ldots, K\}} \vert \Tilde{\alpha}^K_k - \alpha^K_k\vert \right)  O_p(1)\\
    & \overset{[4]}{=}  \left\vert\sqrt{n}\En (\hat{m}_K(Z) - \Tilde{m}_K(Z))U \right\vert + o_p(1)\\
    & \overset{[5]}{\leq}  (n\En (\hat{m}_K(Z) - \Tilde{m}_K(Z))^2)^{1/2} (\En U^2)^{1/2} + o_p(1) \\
    & \overset{[6]}{=} o_p(1),
\end{align*}
where {\footnotesize [1]} follows from the triangle inequality, {\footnotesize [2]} follows from Hölder's inequality, {\footnotesize [3]} follows from $\sum_{k=1}^K \Gn \mathbbm{1}_k(\gamma^K(Z))U = \Gn U$ and application of the central limit theorem, {\footnotesize [4]} follows from consistency of the infeasible least squares estimator, {\footnotesize [5]} follows from Cauchy-Schwarz, and {\footnotesize [6]} follows from Lemma \ref{lemma:rate_hat_tilde} and Assumption \ref{assumption:asymptotics_setup} (a). Hence, \begin{align*}
    \sqrt{n}\En \hat{F}_K U = \Gn F_K U + o_p(1).
\end{align*}

Combining, we have by Slutsky and the central limit theorem \begin{align}\label{eq:dist_convergence}
    \left(\En\hat{F}_K W^\top\right)^{-1} \sqrt{n}\En \hat{F}_K U = \left(\En F_K W^\top\right)^{-1}\Gn F_K U + o_p(1) \overset{d}{\to} N(0, \Sigma_K),
\end{align}
where $\Sigma_K$ is given in Theorem \ref{theorem:CIV_w_covariates_mk}. Pre-multiplication with $\Sigma_K^{-1/2}$ then gives the first desired result.

The proof of Theorem \ref{theorem:CIV_w_covariates_mk} concludes with Lemma \ref{lemma:SimgaK_consistency} which states that $\hat{\Sigma}_K$ is a consistent estimator for the covariance matrix $\Sigma_K$.

\subsection{Consistency of $\hat{\Sigma}_K$}

\begin{lemma}\label{lemma:SimgaK_consistency}
    Let the assumptions of Theorem \ref{theorem:CIV_w_covariates_mk} hold. Then, $\forall K \in \{2, \ldots, K_0 \}$, \begin{align*}
        \hat{\Sigma}_K = \Sigma_K + o_p(1).
    \end{align*}
\end{lemma}

\begin{proof}
    Fix an arbitrary $K \in \{2, \ldots, K_0 \}$. Given \eqref{eq:FKW_inv_consistency} and non-singularity of $\E U^2 F_K F_K^\top$ by Assumption \ref{assumption:iv_setup} (a) and relevance (as assumed in the statement of Theorem \ref{theorem:CIV_w_covariates_mk}), it suffices to show
\begin{align*}
    \En \hat{U}^2 \hat{F}_K\hat{F}_K^\top = \En U^2 F_K F_K^\top + o_p(1) = \E U^2 F_K F_K^\top + o_p(1) 
\end{align*}
to establish consistency of the covariance estimator $\hat{\Sigma}_K$. Note that by the triangle inequality, \begin{align}\label{eq:UFF_twoterms}
\begin{aligned}
    &\left\| \En \hat{U}^2 \hat{F}_K\hat{F}_K^\top - \En U^2 F_K F_K^\top\right\| \\
    \leq& \left\| \En \hat{U}^2 \hat{F}_K\hat{F}_K^\top - \En U^2 \hat{F}_K \hat{F}_K^\top\right\|  + \left\| \En U^2 \hat{F}_K\hat{F}_K^\top - \En U^2 F_K F_K^\top\right\|.
\end{aligned}
\end{align}
I begin by showing that both of the terms in \eqref{eq:UFF_twoterms} are $o_p(1)$.

For the first term in \eqref{eq:UFF_twoterms}, consider \begin{align}\label{eq:UFF_1}
    \left\| \En \hat{U}^2\hat{F}_K\hat{F}_K^\top - \En U^2\hat{F}_K\hat{F}_K^\top\right\| \leq \left\| \En\left((\hat{\theta}^K - \theta_0)^\top W\right)^2\hat{F}_K\hat{F}_K^\top \right\| + \left\|\En\left((\hat{\theta}^K - \theta_0)^\top W\right)U\hat{F}_K\hat{F}_K^\top \right\|,
\end{align}
by the triangle inequality. For the first term in \eqref{eq:UFF_1}, \begin{align*}
    \left\| \En\left((\hat{\theta}^K - \theta_0)^\top W\right)^2\hat{F}_K\hat{F}_K^\top \right\| &\leq \sum_{j=1}^{J + 1}   (\hat{\theta}^K_j - \theta_{0j})^2 \left\|\En W_j^2 \hat{F}_K\hat{F}_K^\top \right\|\\
    &\leq  \sum_{j=1}^{J + 1} \left(\sqrt{n}(\hat{\theta}^K_j - \theta_{0j})\right)^2  \left(\frac{1}{n}\max_{i\leq n} W_{ij}^2\right) \left\|\En  \hat{F}_K\hat{F}_K^\top \right\|\\
    &\overset{[1]}{=} O_p(1) \sum_{j=1}^{J + 1} \left(\frac{1}{n}\max_{i\leq n} W_{ij}^2\right)  \left\|\En \hat{F}_K\hat{F}_K^\top \right\|,
\end{align*}
where {\footnotesize [1]} follows from \eqref{eq:dist_convergence}. Similarly for the second term in \eqref{eq:UFF_1}, \begin{align*}
    \left\| \En\left((\hat{\theta}^K - \theta_0)^\top W\right)U\hat{F}_K\hat{F}_K^\top \right\| & \leq \sum_{j=1}^{J + 1} \left\vert\sqrt{n}(\hat{\theta}^K_j - \theta_{0j})\right\vert  \left(\frac{1}{\sqrt{n}}\max_{i\leq n} \vert W_{ij}U_i\vert\right) \left\| \En  \hat{F}_K\hat{F}_K^\top \right\|\\
    &=O_p(1)\sum_{j=1}^{J + 1} \left(\frac{1}{\sqrt{n}}\max_{i\leq n} \vert W_{ij}U_i\vert\right) \left\| \En \hat{F}_K\hat{F}_K^\top \right\|.
\end{align*}
To show \eqref{eq:UFF_1} is $o_p(1)$, it thus suffices to show that for all $j\in\{1, \ldots, J+1\}$\begin{align}\label{eq:Uhat_ops}
    \frac{1}{n}\max_{i\leq n} W_{ij}^2 = o_p(1), \quad \frac{1}{\sqrt{n}}\max_{i\leq n} \vert W_{ij}U_i\vert = o_p(1),
\end{align}
and \begin{align}\label{eq:hatFF_is_Op}
    \quad \left\| \En \hat{F}_K\hat{F}_K^\top \right\|=O_p(1).
\end{align}

To show \eqref{eq:Uhat_ops}, I leverage a simple inequality: For random variables $S_1, \ldots, S_n$ that for $r > 1$ satisfy $\En \vert S \vert^r = O_p(1)$, we have \begin{align}\label{eq:max_Wiinequality}
    \max_{i \leq n} \vert S_i\vert \leq n\En \vert S\vert \leq n^{1/r}\left(\En\vert S\vert^r\right)^{1/r} = O_p(n^{1/r}),
\end{align}
where the second inequality follows from Jensen's inequality. By Assumption \ref{assumption:Pn_setup} and Assumption \ref{assumption:asymptotics_setup} (a), application of \eqref{eq:max_Wiinequality} implies \eqref{eq:Uhat_ops}. 

To show \eqref{eq:hatFF_is_Op}, I prove \begin{align*}
    \En \hat{F}_K\hat{F}_K^\top = \En F_K F_K^\top + o_p(1) = \E F_K F_K^\top + o_p(1)
\end{align*}
so that \eqref{eq:hatFF_is_Op} follows from boundedness of $\| \E F_K F_K^\top \|$ by Assumption \ref{assumption:Pn_setup} and Assumption \ref{assumption:asymptotics_setup} (b). Note that by triangle inequality\begin{align*}
   \left\| \En \hat{F}_K\hat{F}_K^\top - \En \hat{F}_KF_K^\top + \En \hat{F}_KF_K^\top - \En F_KF_K^\top\right\| \leq \left\| \En \hat{F}_K(\hat{F}_K-F_K)^\top\right\| + \left\| \En (\hat{F}_K-F_K)F_K^\top\right\|.
\end{align*}
By \eqref{eq:diff_F}, it suffices to consider $\left\|\En \hat{F}_K\Delta_n(Z,X)\right\|$ and $\left\|\En F_K\Delta_n(Z,X)\right\|$. Note that by \eqref{eq:Deltan_W_op1}, $\left\|\En X\Delta_n(Z,X)\right\| = o_p(1)$, and by analogous arguments, $\left\|\En F_K\Delta_n(Z,X)\right\|=o_p(1)$. It thus suffices to consider only the first component of $\left\|\En \hat{F}_K\Delta_n(Z,X)\right\|$. In particular, \begin{align*}
    &\left\vert\En (\hat{m}_K(Z) + X^\top\hat{\pi}) \Delta_n(Z,X)\right\vert \\
    \overset{[1]}{\leq} & \left(\En \hat{m}_K(Z)^2\right)^{1/2}\left(\left(\En (\hat{m}_K(Z) - m_K(Z))^2\right)^{1/2} + \sum_{j=1}^J\vert\hat{\pi}_j - \pi_{0j}\vert\left(\En X_j^2\right)^{1/2}\right)\\
    &\qquad +\left(\En (\hat{m}_K(Z) - m_K(Z))^2\right)^{1/2}\left(\sum_{j=1}^J \vert\hat{\pi}_j\vert\left(\En X_j^2\right)^{1/2}\right) + \sum_{j,\ell=1}^J \vert\hat{\pi}_j - \hat{\pi}_{0j}\vert  \vert\hat{\pi}_\ell\vert \En \vert X_jX_\ell\vert\\
    \overset{[2]}{=} & O_p(1)\left(\left(\En (\hat{m}_K(Z) - m_K(Z))^2\right)^{1/2} + \sum_{j=1}^J\vert\hat{\pi}_j - \pi_{0j}\vert\right)\\
    \overset{[3]}{=}& o_p(1),
\end{align*}
where {\footnotesize [1]} applies the triangle inequality and Cauchy-Schwarz, {\footnotesize [2]} follows from Assumption \ref{assumption:Pn_setup} (b) and Assumption \ref{assumption:asymptotics_setup} (b), and {\footnotesize [3]} follows from Lemma \ref{lemma:consistency_m_covariates} and Assumption \ref{assumption:asymptotics_setup} (c). This shows \eqref{eq:UFF_1} is $o_p(1)$.

For the second term in \eqref{eq:UFF_twoterms}, consider \begin{align}\label{eq:UFF_2}
\begin{aligned}
        \left\| \En U^2 (\hat{F}_K\hat{F}_K^\top - F_KF_K^\top) \right\| &\leq 2\left\| \En U^2 F_K(\hat{F}_K - F_K)^\top \right\| + \left\| \En U^2 (\hat{F}_K - F_K)(\hat{F}_K - F_K)^\top \right\|\\
    &\leq 2\left\| \En U^2 F_K \Delta_n(Z,X) \right\| + \left\| \En U^2 \Delta_n(Z,X)^2 \right\|
\end{aligned}
\end{align}
which follows from the triangle inequality and \eqref{eq:diff_F}. For the first term in \eqref{eq:UFF_2}, we have \begin{align}
\begin{aligned}\label{eq:derivation_U2FDelta}
    &\left\| \En U^2 F_K \Delta_n(Z,X) \right\|\\
    \overset{[1]}{\leq} & \left\| \En U^2 F_K (\hat{m}_K(Z) - \tilde{m}_K(Z)) \right\| + \left\| \En U^2 F_K (\tilde{m}_K(Z) - m_K(Z)) \right\| +  \sum_{j=1}^{J}\vert\hat{\pi}_j - \pi_{j0}\vert\left\| \En U^2 F_K X_j \right\|\\
     \overset{[2]}{\leq} & \left(\frac{1}{\sqrt{n}}\max_{i\leq n} U_i^2\right)\left(n\En \left(\hat{m}_K(Z) - \tilde{m}_K(Z)\right)^2\right)^{1/2}\left\| \En F_KF_K^\top \right\|  + \sum_{k=1}^K\vert \tilde{\alpha}^K_k - \alpha^K_k \vert \left\| \En U^2 F_K \right\| \\
    &\qquad + \sum_{j=1}^{J}\vert(\sqrt{n}(\hat{\pi}_j - \pi_{j0})\vert\left(\frac{1}{\sqrt{n}}\max_{i\leq n} \vert X_{ij}\vert \right)\left\| \En U^2 F_K \right\|\\
     \overset{[3]}{=} & O_p(1) \left(\left(n\En \left(\hat{m}_K(Z) - \tilde{m}_K(Z)\right)^2\right)^{1/2} + \left(\frac{1}{\sqrt{n}}\max_{i\leq n} \vert X_{ij}\vert \right) + \sum_{k=1}^K\vert \tilde{\alpha}^K_k - \alpha^K_k \vert \right)\\
    \overset{[4]}{=}& o_p(1),
\end{aligned}
\end{align}
where {\footnotesize [1]} applies the triangle inequality, {\footnotesize [2]} applies Cauchy-Schwarz, {\footnotesize [3]} follows under Assumption \ref{assumption:Pn_setup} (b) and Assumption \ref{assumption:asymptotics_setup} (a)-(b) and application of \eqref{eq:max_Wiinequality}, and {\footnotesize [4]} follows from Lemma \ref{lemma:rate_hat_tilde} application of \eqref{eq:max_Wiinequality}, and consistency of the infeasible least squares estimator. For the second term in \eqref{eq:UFF_2}, we have by the triangle inequality\begin{align*}
    &\left\| \En U^2 \Delta_n(Z,X)^2 \right\|\\
    \leq & \left\| \En U^2 (\hat{m}_K(Z) - m_K(Z))^2 \right\| + 2 \left\| \En U^2 (\hat{m}_K(Z) - m_K(Z))X^\top (\hat{\pi}- \pi_0)\right\| \\
    & \qquad + \left\|(\hat{\pi}- \pi_0)^\top\left(\En U^2XX^\top \right)(\hat{\pi}- \pi_0)\right\|.
\end{align*}
Note that by arguments analogous to those in \eqref{eq:derivation_U2FDelta}, we have \begin{align*}
    &\left\| \En U^2 (\hat{m}_K(Z) - m_K(Z))^2 \right\|\\
    \leq& \left\| \En U^2 (\hat{m}_K(Z) - \tilde{m}_K(Z))^2 \right\| + 2\left\| \En U^2(\hat{m}_K(Z) - \tilde{m}_K(Z))(\tilde{m}_K(Z) - m_K(Z)) \right\| \\
    & \qquad + \left\| \En U^2 (\tilde{m}_K(Z) - m_K(Z))^2 \right\|\\
    \leq & \left(\frac{1}{\sqrt{n}}\max_{i\leq n} U_i^2\right)\left(n\En (\hat{m}_K(Z) - \Tilde{m}_K(Z))^2\right)^{1/2}\\
    &\qquad + 2\sum_{k=1}^K\vert \tilde{\alpha}_k^K - \alpha_k^K \vert \left(\En U^4\right)^{1/2}\left(\En (\hat{m}_K(Z) - \tilde{m}_K(Z))^2\right)^{1/2} \\
    & \qquad + \sum_{k=1}^K\vert \tilde{\alpha}_k^K - \alpha_k^K \vert  \En U^2\\
    =& O_p(1)\left(\left(n\En (\hat{m}_K(Z) - \Tilde{m}_K(Z))^2\right)^{1/2} + \sum_{k=1}^K\vert \tilde{\alpha}_k^K - \alpha_k^K \vert\right)\\
   =&o_p(1),
\end{align*}
and similarly \begin{align*}
    &\left\| \En U^2 (\hat{m}_K(Z) - m_K(Z))X^\top (\hat{\pi}- \pi_0)\right\|\\
    \leq& \sum_{j=1}^J \vert \sqrt{n}(\hat{\pi}_j - \pi_{0j})\vert\left(\frac{1}{\sqrt{n}} \max_{i\leq n} \vert X_{ij}\vert\right)\left(\En U^4\right)^{1/2}\left(\En (\hat{m}_K(Z) - m_K(Z))^2\right)^{1/2}\\
    =& O_p(1)\sum_{j=1}^J\left(\frac{1}{\sqrt{n}} \max_{i\leq n} \vert X_{ij}\vert\right)\\
    =& o_p(1).
\end{align*}
Finally, applying the same arguments again, we have\begin{align*}
    &\left\|(\hat{\pi}- \pi_0)^\top\left(\En U^2XX^\top \right)(\hat{\pi}- \pi_0)\right\|\\
    \leq &\sum_{j,\ell = 1}^J \vert \hat{\pi}_j- \pi_{0j}\vert\vert \hat{\pi}_\ell- \pi_{0\ell}\vert\left\vert\En U^2X_jX_\ell \right\vert\\
    \leq & J\sum_{j=1}^J\left\vert \sqrt{n}(\hat{\pi}_j- \pi_{0j}) \right\vert^2 \left(\frac{1}{\sqrt{n}} \max_{i\leq n} \vert X_{ij}\vert\right)^2 \left\|\En U^2\right\|\\
    = & O_p(1)\left(\frac{1}{\sqrt{n}} \max_{i\leq n} \vert X_{ij}\vert\right)^2\\
    =& o_p(1).
\end{align*}
This shows that \eqref{eq:UFF_2} is $o_p(1)$ and hence that \eqref{eq:UFF_twoterms} is $o_p(1)$.

The proof is concluded by noting that under Assumption \ref{assumption:asymptotics_setup} (a) and (b) we have \begin{align*}
    \En U^2F_KF_K = \E U^2 F_KF_K + o_p(1).
\end{align*}

\end{proof}

\subsection{Semiparametric Efficiency of $\hat{\theta}_{K_0}$}\label{app:semiparametric_proof}

Finally, I prove Corollary \ref{corollary:semiparametric_efficiency}:
 \begin{proof}
Consider the expression of the asymptotic covariance $\Sigma_{K}$ of Theorem \ref{theorem:CIV_w_covariates_mk} for $K = K_0$. Then,\begin{align}\label{eq:efficiency_proof}
\begin{aligned}
         \Sigma_{K_0} & = \E[F_{K_0}W^\top]^{-1}\E[U^2F_{K_0}F_{K_0}^\top]\E[WF_{K_0}^\top]^{-1}\\
         & \overset{[1]}{=} \E[h_0(Z^{(0)},X)h_0(Z^{(0)},X)^\top]^{-1}\E[U^2h_0(Z^{(0)},X)h_0(Z^{(0)},X)^\top]\E[h_0(Z^{(0)},X)h_0(Z^{(0)},X)^\top]^{-1}\\
         & \overset{[2]}{=} \sigma^2\E[h_0(Z^{(0)},X)h_0(Z^{(0)},X)^\top]^{-1},
\end{aligned}
     \end{align}
     where {\footnotesize [1]} follows from  the first component of $F_{K_0}$ being equal to $$m_0^{(n)}(Z^{(n)}) + X^\top \pi_0 = Z^{(0)} + X^\top \pi_0  = E[D\vert X, Z^{(0)}]$$ by Assumption \ref{assumption:Z0_setup} (b) and Assumption \ref{assumption:iv_setup} (b), and {\footnotesize [2]} follows from homoskedasticity. The proof concludes with the note that the final term in \eqref{eq:efficiency_proof} is the semiparametric efficiency bound for $\theta_0$ for a fixed law $P$ of $(Y, D, X^\top, Z^{(0)}, U)$. See, e.g.,  \citet{chamberlain1987asymptotic}.
 \end{proof}

\clearpage
\newpage
\section{Additional Simulation Results} \label{app:additional_sims}

This appendix presents additional results for the simulation described in Section \ref{sec:simulation}. Table \ref{tab:app_sim_resK2_constant} and Table \ref{tab:app_sim_resK4_constant} provide simulation results analogous to those presented in Table \ref{tab:app_sim_resK2} and \ref{tab:app_sim_resK4} but where there are no heterogeneous effects in the second stage -- i.e., $\pi(X) = 0$. Further,  Figure \ref{fig:simMK_2} provides results for the additional machine learning-based IV estimators for the same DGP as Figure \ref{fig:simMK_1} in the main text. Note that IJIVE and UJIVE are numerically nearly identical in the considered setting so that their lines largely overlap.

\begin{table}[!htbp]\small
        \begin{threeparttable}
  \caption{Simulation Results with Constant Effects ($K_0=2$)}\label{tab:app_sim_resK2_constant}
\begin{tabular}{lccccccccc}
\toprule
\midrule
$K_0=2$ & \multicolumn{4}{c}{$\E N_z=20$ } &       & \multicolumn{4}{c}{$\E N_z= 25$ } \\
\cmidrule{2-5}\cmidrule{7-10}      & Bias  & MAE   & rp(0.05) & iqr(10,90) &       & Bias  & MAE   & rp(0.05) & iqr(10,90) \\ \midrule
Oracle & -0.006 & 0.057 & 0.047 & 0.216 &       & -0.003 & 0.052 & 0.063 & 0.205 \\
CIV (K=2) & 0.032 & 0.064 & 0.081 & 0.213 &       & 0.016 & 0.052 & 0.070 & 0.193 \\
CIV (K=4) & 0.101 & 0.108 & 0.339 & 0.178 &       & 0.081 & 0.086 & 0.262 & 0.169 \\
Lasso-IV (cv) & 0.163 & 0.196 & 0.753 & 0.187 &       & 0.159 & 0.160 & 0.688 & 0.165 \\
Lasso-IV (plug-in) & 0.282 & 0.290 & 0.547 & 0.437 &       & 0.262 & 0.269 & 0.571 & 0.348 \\
Ridge-IV (cv) & 0.125 & 0.126 & 0.454 & 0.174 &       & 0.101 & 0.105 & 0.383 & 0.162 \\
xboost-IV & 0.121 & 0.126 & 0.455 & 0.173 &       & 0.101 & 0.106 & 0.384 & 0.162 \\
ranger-IV & 0.130 & 0.136 & 0.502 & 0.174 &       & 0.112 & 0.115 & 0.442 & 0.160 \\
TSLS  & 0.120 & 0.126 & 0.452 & 0.173 &       & 0.101 & 0.105 & 0.382 & 0.162 \\
JIVE  & -0.022 & 0.068 & 0.056 & 0.256 &       & -0.014 & 0.057 & 0.054 & 0.229 \\
IJIVE & -0.012 & 0.068 & 0.056 & 0.250 &       & -0.007 & 0.057 & 0.064 & 0.225 \\
UJIVE & -0.013 & 0.068 & 0.055 & 0.250 &       & -0.007 & 0.057 & 0.064 & 0.225 \\
LIML  & -0.006 & 0.061 & 0.048 & 0.237 &       & -0.004 & 0.054 & 0.063 & 0.210 \\
      &       &       &       &       &       &       &       &       &  \\
$K_0=2$ & \multicolumn{4}{c}{$\E N_z=100$ } &       & \multicolumn{4}{c}{$\E N_z= 150$ } \\
\cmidrule{2-5}\cmidrule{7-10}      & Bias  & MAE   & rp(0.05) & iqr(10,90) &       & Bias  & MAE   & rp(0.05) & iqr(10,90) \\
\midrule
Oracle & 0.000 & 0.025 & 0.055 & 0.094 &       & 0.000 & 0.021 & 0.058 & 0.077 \\
CIV (K=2) & 0.000 & 0.025 & 0.055 & 0.094 &       & 0.000 & 0.021 & 0.058 & 0.077 \\
CIV (K=4) & 0.020 & 0.029 & 0.097 & 0.092 &       & 0.014 & 0.023 & 0.080 & 0.076 \\
Lasso-IV (cv) & 0.041 & 0.042 & 0.237 & 0.089 &       & 0.027 & 0.030 & 0.167 & 0.074 \\
Lasso-IV (plug-in) & 0.068 & 0.068 & 0.446 & 0.091 &       & 0.025 & 0.029 & 0.154 & 0.076 \\
Ridge-IV (cv) & 0.029 & 0.033 & 0.143 & 0.090 &       & 0.019 & 0.026 & 0.112 & 0.076 \\
xboost-IV & 0.029 & 0.033 & 0.144 & 0.090 &       & 0.019 & 0.026 & 0.112 & 0.076 \\
ranger-IV & 0.037 & 0.038 & 0.197 & 0.089 &       & 0.025 & 0.029 & 0.159 & 0.075 \\
TSLS  & 0.029 & 0.033 & 0.144 & 0.090 &       & 0.019 & 0.026 & 0.112 & 0.076 \\
JIVE  & -0.002 & 0.026 & 0.053 & 0.097 &       & -0.002 & 0.021 & 0.054 & 0.080 \\
IJIVE & 0.000 & 0.026 & 0.055 & 0.096 &       & -0.001 & 0.021 & 0.054 & 0.080 \\
UJIVE & -0.001 & 0.026 & 0.055 & 0.096 &       & -0.001 & 0.021 & 0.054 & 0.080 \\
LIML  & 0.000 & 0.025 & 0.061 & 0.096 &       & 0.000 & 0.021 & 0.061 & 0.078 \\
\midrule\bottomrule
\end{tabular}
  \begin{tablenotes}[para,flushleft]
  \scriptsize
  \item \textit{Notes.} Simulation results are based on 1000 replications using the DGP described in Section \ref{sec:simulation} with $K_0=2$ with constant second-stage treatment effects. See the notes of Table \ref{tab:app_sim_resK2} for a description of the estimators.
  \end{tablenotes}
    \end{threeparttable}
\end{table}

\begin{table}[h]\small
        \begin{threeparttable}
  \caption{Simulation Results with Constant Effects ($K_0=4$)}\label{tab:app_sim_resK4_constant}
\begin{tabular}{lccccccccc}
\toprule
\midrule
$K_0=4$ & \multicolumn{4}{c}{$\E N_z=20$ } &       & \multicolumn{4}{c}{$\E N_z= 25$ } \\
\cmidrule{2-5}\cmidrule{7-10}      & Bias  & MAE   & rp(0.05) & iqr(10,90) &       & Bias  & MAE   & rp(0.05) & iqr(10,90) \\ \midrule
Oracle & 0.008 & 0.057 & 0.055 & 0.211 &       & 0.001 & 0.048 & 0.045 & 0.184 \\
CIV (K=2) & 0.095 & 0.099 & 0.251 & 0.207 &       & 0.071 & 0.079 & 0.172 & 0.194 \\
CIV (K=4) & 0.114 & 0.116 & 0.374 & 0.179 &       & 0.090 & 0.094 & 0.297 & 0.170 \\
Lasso-IV (cv) & 0.158 & 0.155 & 0.558 & 0.192 &       & 0.129 & 0.131 & 0.504 & 0.170 \\
Lasso-IV (plug-in) & 0.204 & 0.211 & 0.457 & 0.324 &       & 0.181 & 0.180 & 0.458 & 0.242 \\
Ridge-IV (cv) & 0.102 & 0.125 & 0.444 & 0.175 &       & 0.100 & 0.102 & 0.363 & 0.161 \\
xboost-IV & 0.124 & 0.125 & 0.445 & 0.174 &       & 0.101 & 0.103 & 0.363 & 0.161 \\
ranger-IV & 0.126 & 0.126 & 0.462 & 0.173 &       & 0.104 & 0.106 & 0.380 & 0.161 \\
TSLS  & 0.124 & 0.125 & 0.445 & 0.174 &       & 0.101 & 0.103 & 0.363 & 0.161 \\
JIVE  & -0.016 & 0.068 & 0.054 & 0.259 &       & -0.015 & 0.058 & 0.047 & 0.220 \\
IJIVE & -0.007 & 0.066 & 0.056 & 0.253 &       & -0.008 & 0.057 & 0.052 & 0.216 \\
UJIVE & -0.007 & 0.066 & 0.056 & 0.254 &       & -0.009 & 0.057 & 0.052 & 0.217 \\
LIML  & 0.000 & 0.061 & 0.058 & 0.235 &       & -0.005 & 0.051 & 0.055 & 0.200 \\
      &       &       &       &       &       &       &       &       &  \\
$K_0=4$ & \multicolumn{4}{c}{$\E N_z=100$ } &       & \multicolumn{4}{c}{$\E N_z= 150$ } \\
\cmidrule{2-5}\cmidrule{7-10}      & Bias  & MAE   & rp(0.05) & iqr(10,90) &       & Bias  & MAE   & rp(0.05) & iqr(10,90) \\
\midrule
Oracle & -0.001 & 0.026 & 0.049 & 0.096 &       & 0.002 & 0.020 & 0.061 & 0.081 \\
CIV (K=2) & 0.008 & 0.028 & 0.069 & 0.106 &       & 0.003 & 0.023 & 0.065 & 0.087 \\
CIV (K=4) & 0.010 & 0.026 & 0.064 & 0.095 &       & 0.005 & 0.020 & 0.067 & 0.079 \\
Lasso-IV (cv) & 0.032 & 0.035 & 0.162 & 0.089 &       & 0.025 & 0.029 & 0.149 & 0.077 \\
Lasso-IV (plug-in) & 0.021 & 0.031 & 0.096 & 0.094 &       & 0.018 & 0.026 & 0.110 & 0.081 \\
Ridge-IV (cv) & 0.026 & 0.031 & 0.138 & 0.091 &       & 0.021 & 0.025 & 0.127 & 0.076 \\
xboost-IV & 0.026 & 0.031 & 0.138 & 0.091 &       & 0.021 & 0.025 & 0.126 & 0.076 \\
ranger-IV & 0.028 & 0.032 & 0.149 & 0.090 &       & 0.022 & 0.026 & 0.136 & 0.076 \\
TSLS  & 0.026 & 0.031 & 0.138 & 0.091 &       & 0.021 & 0.025 & 0.126 & 0.076 \\
JIVE  & -0.005 & 0.027 & 0.047 & 0.099 &       & -0.001 & 0.021 & 0.060 & 0.080 \\
IJIVE & -0.003 & 0.027 & 0.042 & 0.098 &       & 0.000 & 0.021 & 0.065 & 0.080 \\
UJIVE & -0.003 & 0.027 & 0.042 & 0.098 &       & 0.000 & 0.021 & 0.065 & 0.080 \\
LIML  & -0.002 & 0.027 & 0.050 & 0.097 &       & 0.001 & 0.020 & 0.062 & 0.079 \\ \midrule \bottomrule
\end{tabular}%
  \begin{tablenotes}[para,flushleft]
  \scriptsize
  \item \textit{Notes.} Simulation results are based on 1000 replications using the DGP described in Section \ref{sec:simulation} with $K_0=4$ with constant second-stage treatment effects. See the notes of Table \ref{tab:app_sim_resK2} for a description of the estimators.
  \end{tablenotes}
    \end{threeparttable}
\end{table}
\clearpage

\begin{figure}[!h] 
\caption{Power Curves with and without Treatment Effect Heterogeneity (Contd.)}\label{fig:simMK_2}
    \centering
    \begin{subfigure}[b]{0.375\textwidth}
    \centering
    \includegraphics[width=1\textwidth]{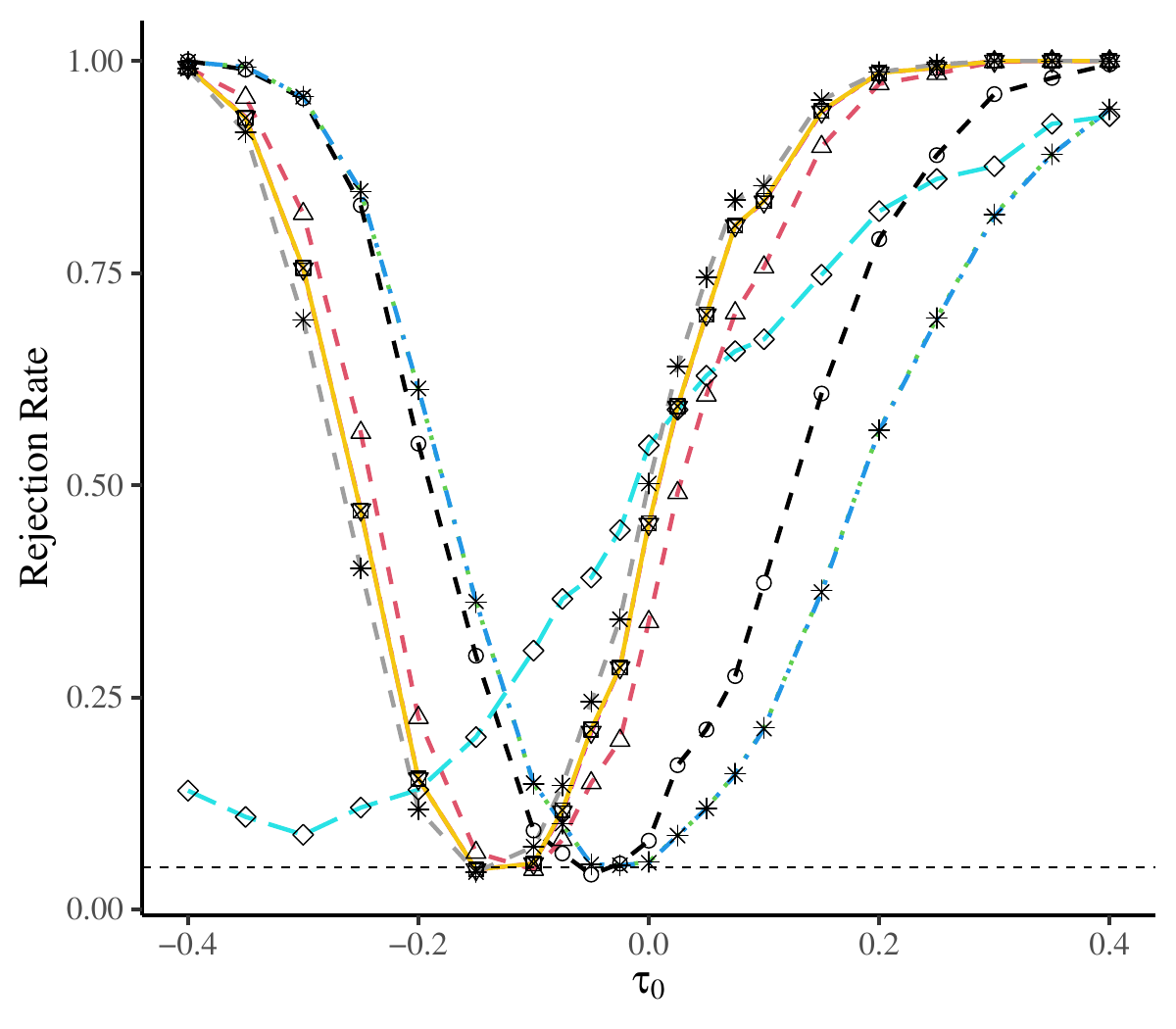}
    \subcaption{constant $\pi_0$, $\E N_z= 20$}
    \end{subfigure}
    \begin{subfigure}[b]{0.375\textwidth}
    \centering
    \includegraphics[width=1\textwidth]{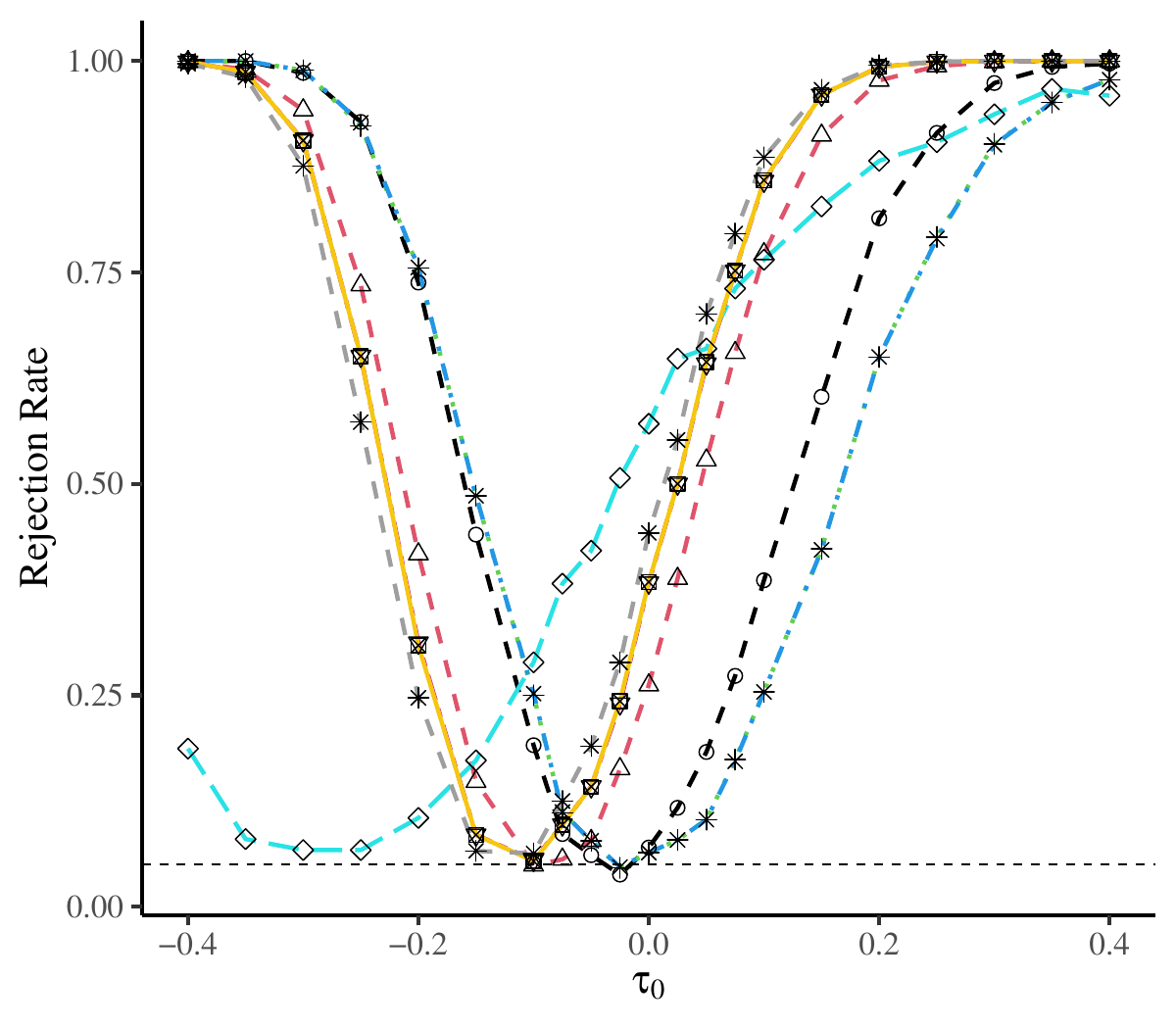}
    \subcaption{constant $\pi_0$, $\E N_z= 25$}
    \end{subfigure}
    \begin{subfigure}[b]{0.16\textwidth}
    \centering
    \includegraphics[width=1\textwidth]{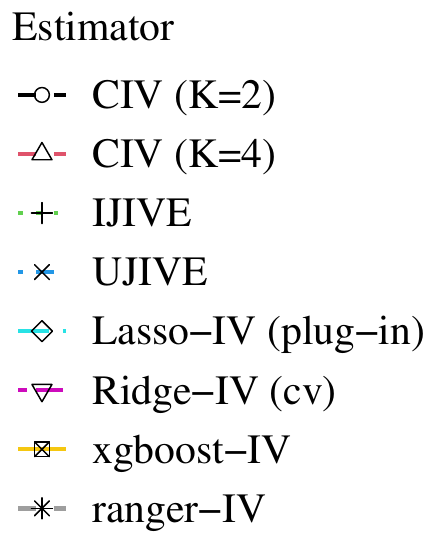}\vspace{6.5em}
    \end{subfigure}
    \\
    \begin{subfigure}[b]{0.375\textwidth}
    \centering
    \includegraphics[width=1\textwidth]{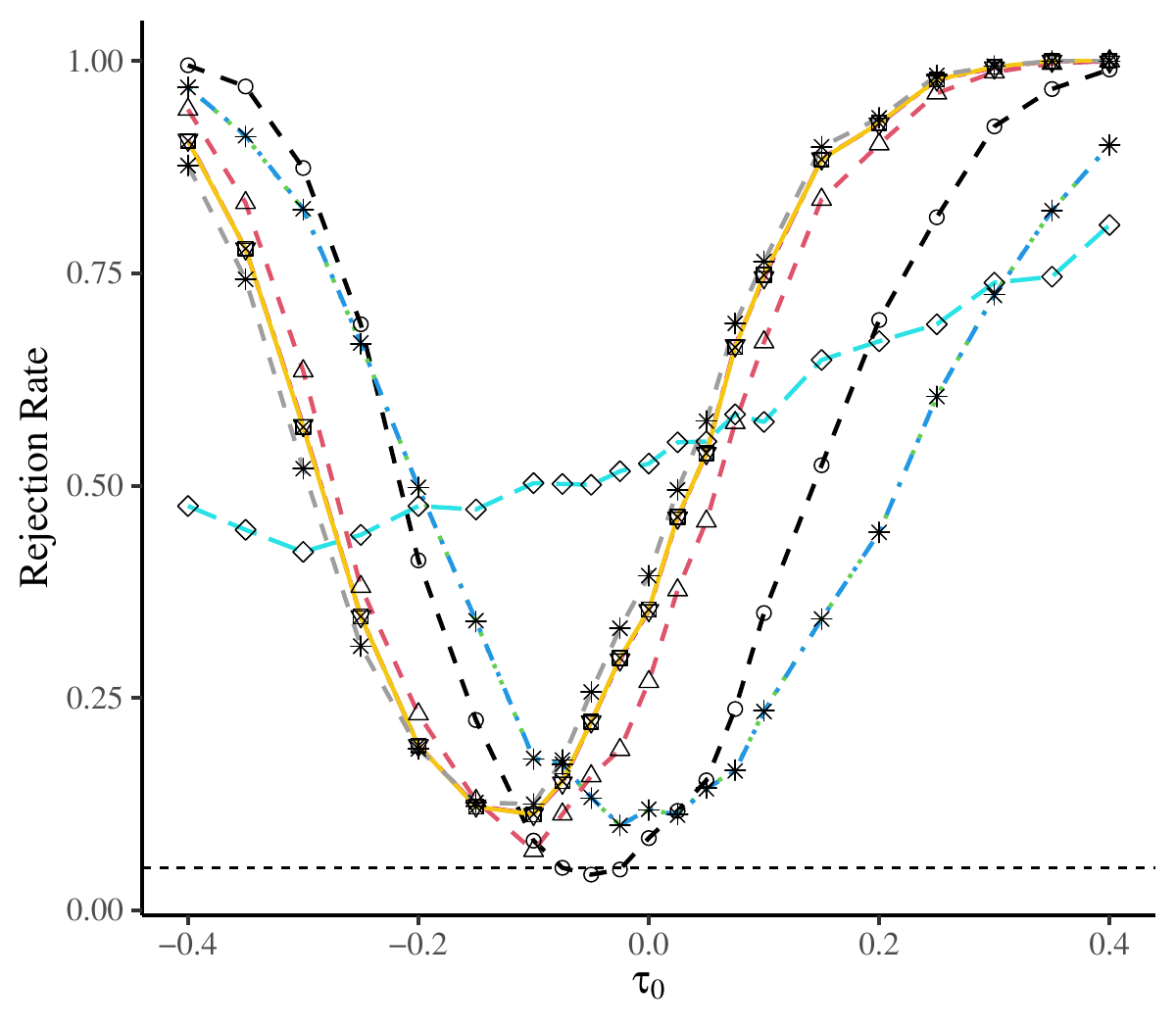}
    \subcaption{heterogeneous $\pi_0$, $\E N_z= 20$}
    \end{subfigure}
    \begin{subfigure}[b]{0.375\textwidth}
    \centering
    \includegraphics[width=1\textwidth]{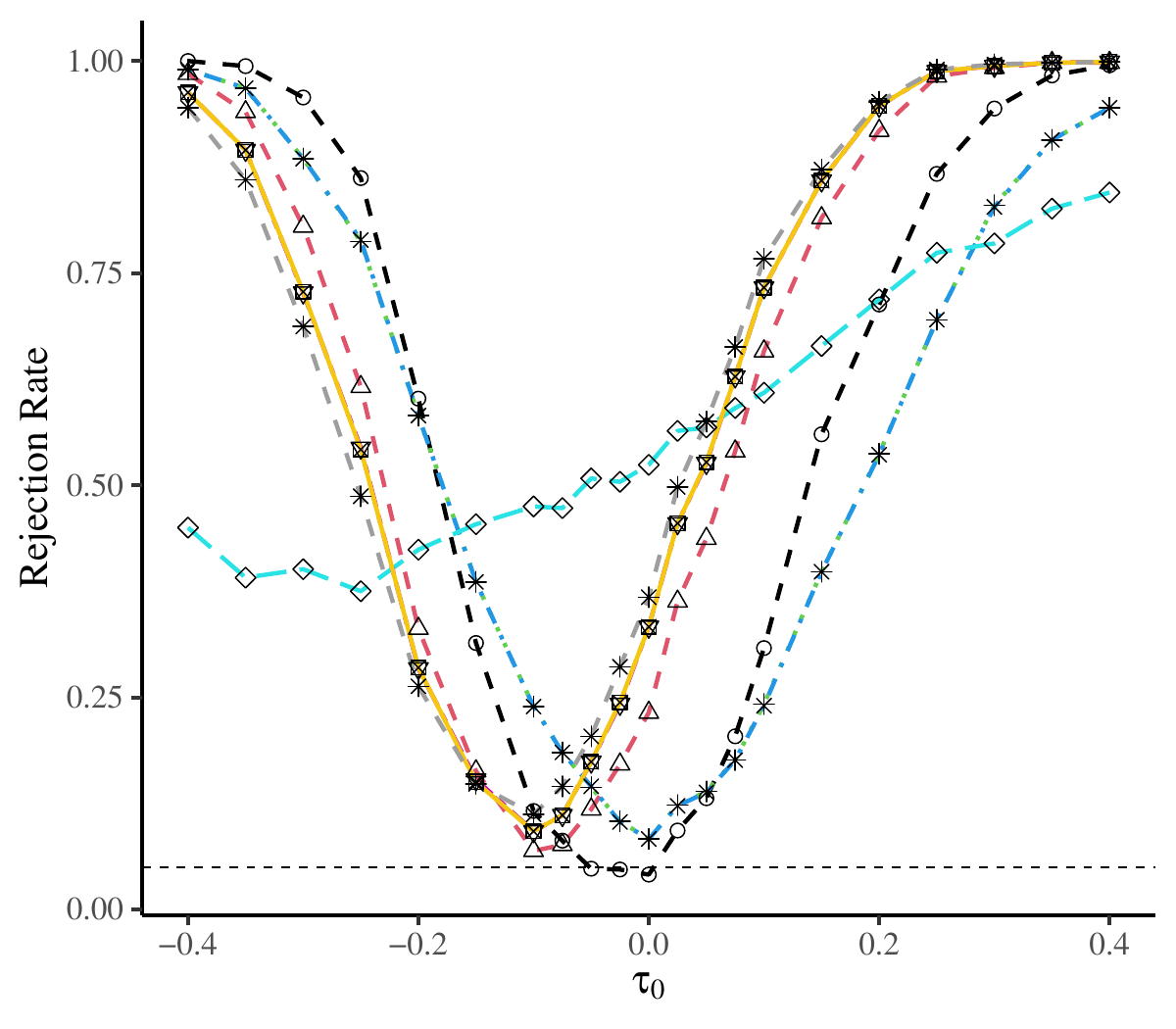}
    \subcaption{heterogeneous $\pi_0$, $\E N_z= 25$}
    \end{subfigure}\hspace{6em}
     \vskip0.25em
 \begin{minipage}{0.9\textwidth} 
 {\scriptsize \textit{Notes.} Simulation results are based on 1000 replications using the DGP described in Section \ref{sec:simulation} with $K_0=2$. Panels (a) and (b) are with constant effects so that $\pi_0(X_i) = \tau_0$. Panels (c) and (d) allow for covariate-dependent effects with $\pi_0(X_i) = 1 - 2X_i +  \tau_0$. The power curves plot the rejection rate of testing $H_0: \tau_0 = 0$ at significance level $\alpha=0.05$. ``CIV ($K=2$)'' and ``CIV ($K=4$)'' correspond to the proposed categorical IV estimators restricted to 2 and 4 support points in the first stage, ``IJIVE'' and ``UJIVE'' correspond to the jackknife-based estimators of \citet{ackerberg2009improved} and \citet{kolesar2013estimation}, ``Lasso-IV (plug-in)'' denotes IV estimator that uses lasso to estimate the optimal instrument using penalty parameters chosen via the plug-in rule of \citet{belloni2012sparse} , ``Ridge-IV (cv)'' denotes an IV estimator that uses ridge regression to estimate the optimal instrument using a penalty parameter chosen via 10-fold cross validation, and ``\texttt{xgboost}-IV'' and ``\texttt{ranger}-IV'' denote IV estimators that use gradient tree boosting as implemented by the \texttt{xgboost} package and random forests as implemented by the \texttt{ranger} package to estimate the optimal instrument, respectively. }
 \end{minipage}
\end{figure}

\clearpage
\section{Additional Plots for the Miami-Dade County Data}\label{app:additional_plots_judges}

Figure \ref{fig:hist_ncases_FE} plots the distribution of cases per court-by-time fixed effect in the Miami-Dade data considered in Section \ref{sec:empirical_example_judges}.

\begin{figure}[!h] 
\caption{Distribution of Observations per Court-by-Time fixed effect}\label{fig:hist_ncases_FE}
    \centering
    \includegraphics[width=0.5\textwidth]{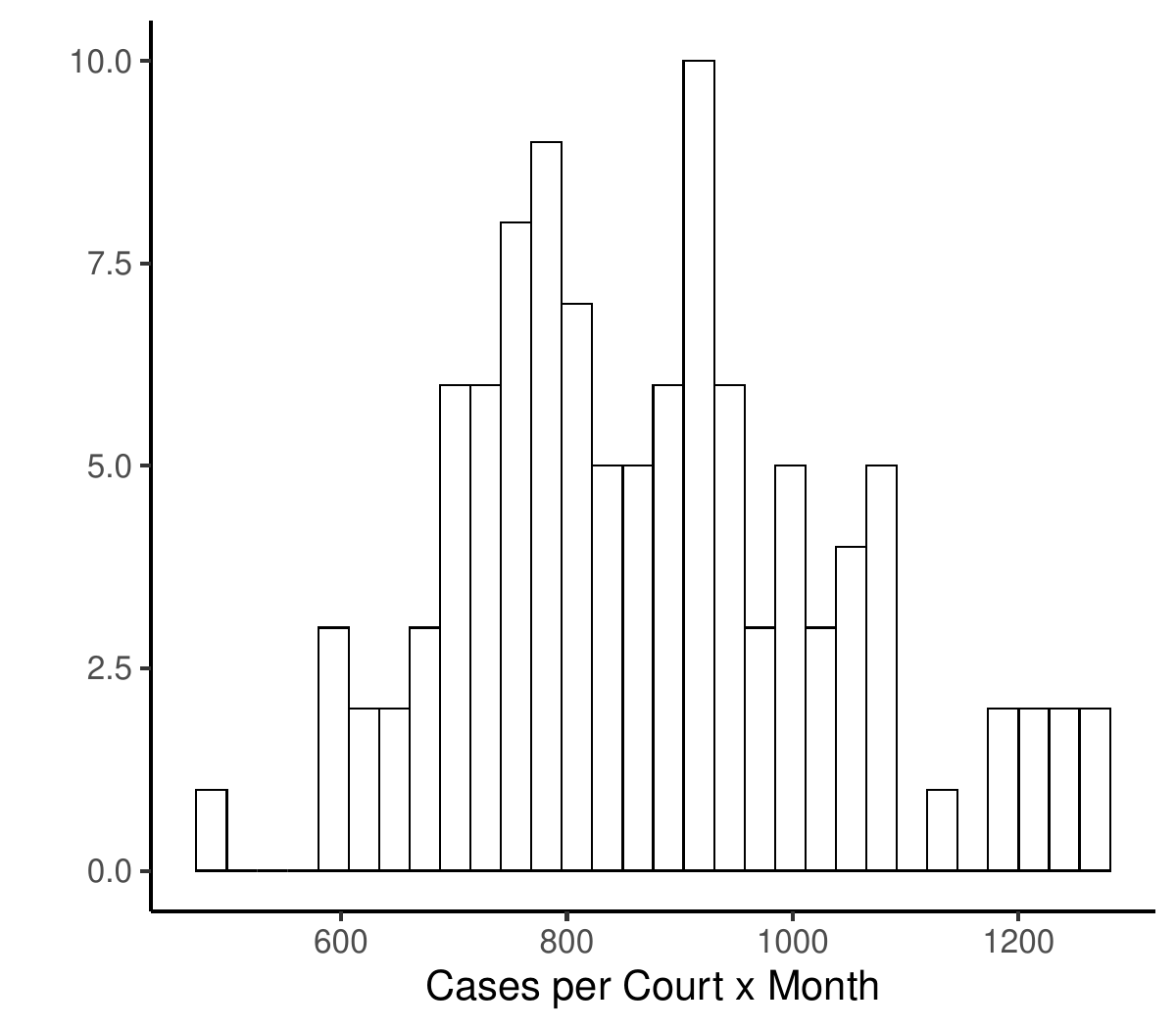}
     \vskip0.25em
 \begin{minipage}{0.9\textwidth} 
 {\scriptsize \textit{Notes.} The histogram plots the distribution of cases per court-by-time fixed effects in the Miami-Dade data. The vertical axis corresponds to the number of court-by-time cells (not the number of observations).}
 \end{minipage}
\end{figure} 

\clearpage
\section{Empirical Example: Returns to Schooling}\label{app:ak_app}

This appendix revisits the \citet[AK91 hereafter]{angrist1991does} application as an empirical illustration. The AK91 example has been the primary empirical example in the many and weak instruments literature \citep[among many others, see, e.g.,][]{bound1995problems,angrist1995split,angrist1999jackknife, donald2001choosing,hansen2008estimation,angrist2020machine, mikusheva2020inference}. Despite the many applications of IV estimators to the AK91 setting, the purely categorical nature of the considered instruments is not commonly taken advantage of. I revisit the returns to schooling analysis of AK91 in a sample of 329,509 American men born between 1930 and 1939. The authors use quarter of birth (QOB) indicators as instruments for the highest grade completed. This approach is motivated by two arguments. First, quarter of birth is plausibly exogenous with other determinants of wages. Second, children born in later quarters attain the minimum dropout age after having completed more schooling. While the QOB instrument thus appears a valid approach to instrument for years of schooling, it averages over potential heterogeneity in educational policy. A larger number of categorical instruments arises when interactions between QOB and indicators for year of birth and place of birth (POB) are formed. These interactions capture the fact that mandatory schooling laws differ across cohorts and states. In particular, the interactions account for the possibility that two students born in the same QOB may differ on whether they can dropout depending on the particular policy in place in their state. I focus on interactions of QOB with the 51 values of POB to compare estimators in random subsamples of the original data.\footnote{With the exception of \citet{angrist2020machine} and \citet{mikusheva2020inference}, most empirical analyzes of the \citet{angrist1991does} data consider disjoint interactions of QOB with YOB and POB, respectively, leading to 180 excluded instruments. As highlighted in \citet{blandhol2022tsls}, causal interpretations of such specifications likely violate the monotonicity correctness of the first stage. Monotonicity correctness is guaranteed mechanically by the saturated interaction specification of QOB and POB considered in this paper.} 

The setting of AK91 is an interesting application of the CIV estimator proposed here as the economic motivation of the QOB-instrument directly suggests existence of a latent \textit{binary} optimal instrument: A student's dropout decision either is or is not constrained by the mandatory attendance law in place in their state. While it would thus be ideal to know of the specific educational policies across all states, the CIV estimator provides an approach to estimate this map from interactions of QOB and POB to the latent binary instrument that captures whether or not a student was constrained in their schooling decision. In particular, consider a sample $\{(Y_i, D_i, \text{QOB}_i, \text{POB}_i)\}_{i=1}^n$ of $n$ students with log-weekly wage $Y_i$, $D_i$ years of education, and born in $\text{QOB}_i$ and $\text{POB}_i$. The CIV estimator with $K=2$ directly applies if \begin{align}\label{eq:equation_AK}
    \E[D_i \vert \text{QOB}_i, \text{POB}_i] = m_0(\text{QOB}_i, \text{POB}_i) + \sum_{j=1}^{51}\mathbbm{1}_j(\text{POB}_i)\pi_{0j},
\end{align}
where $m_0$ is the map to the optimal instrument in Assumption \ref{assumption:Z0_setup} (b), the optimal instrument captures the difference in schooling for a constrained or unconstrained student, and $(\pi_{0j})_{j=1}^{51}$ in the second term capture level-differences in years of schooling across states. Note that the restriction of $m_0$ to two support points implies that the expected increment in years of schooling implied by the constraint of a mandatory attendance law is constant across states. If there is additional heterogeneity in the first stage effect of being constrained, CIV with $K=2$ serves as an approximation to the optimal instrument estimator but remains root-$n$ normal at the cost of statistical efficiency (see Theorem \ref{theorem:CIV_w_covariates_mk}).

Table \ref{tab:AK_res} presents estimation results of the coefficient of log-weekly wage on years of schooling in the original AK91 data using ten estimators: TSLS, three CIV estimators with $K=2$, $K=3$, and $K=4$, LIML, JIVE, IJIVE, UJIVE, and two specifications of the post-Lasso TSLS estimator proposed by \citet{belloni2012sparse}. The two lasso-based IV estimators differ in how the indicators for the categorical instruments are constructed: While Lasso-IV (1) includes three QOB indicators (dropping the indicator for the first QOB) and 150 indicators for interactions between QOB and POB, Lasso-IV (2) directly includes 153 indicators for the interactions between the QOB and POB values.\footnote{In particular, the indicator specifications for Lasso-IV (1) and (2) are constructed in \textsf{R} using the commands \texttt{model.matrix($\sim$ QOB*POB)} and \texttt{model.matrix($\sim$ QOB:POB)}, respectively.} Lasso-IV (1) and Lasso-IV (2) thus differ in how they define the constant, or, ``base-case''. Of course, such differences in indicator specification results in numerically identical estimates for any of the other considered estimators. 

Column (5) of Table \ref{tab:AK_res} provides estimation results using the full sample of 329,509 observations. In this sample size, TSLS, LIML, all jackknife estimators, and all three CIV estimators provide qualitatively similar estimation results: All coefficient estimates are near 0.1 and statistically significant at a 5\% nominal level. Similar performances of TSLS and the CIV estimators at this very large sample size are reassuring since one may reasonably expect the many instruments problem to be avoided given only the 150 excluded instruments, and closeness of the LIML estimator may suggest only little heterogeneity in the second stage effects. Among this set of estimators, the qualitative conclusions thus seem robust. In contrast, the results of the lasso-based IV estimators depend strongly on the specification of the indicators in the first stage. While Lasso-IV (1) is statistically insignificant, Lasso-IV (2) returns the largest point estimate paired with the smallest standard error of all considered IV estimators. Since the choice of indicator specification is often arbitrary, lasso-based IV estimation in this setting leaves the researcher with two qualitatively very different results. Much akin the approach of \citet{donald2001choosing} who consider an ordered list of instruments in increasing importance from which to choose from, applications of lasso-based estimators for categorical variables require careful consideration of the constructed sets of indicators.\footnote{Note that the post-lasso estimator allows for multiple sets of indicators. However, to achieve the same first stage fitted values as CIV, researchers would be required to include the power set. In this AK91 application the power set corresponds to $2^{150}$ indicators.}

The true coefficient in the AK91 application is unknown. To nevertheless provide some insights into the finite-sample trade-offs between the different estimators, I re-compute the estimators on random subsamples of the original data. Column (1)-(4) provide the corresponding mean coefficient and mean standard error estimates across 250 replications along with the median absolute difference to the corresponding full-sample coefficient estimate.\footnote{This subsampling exercise is similar to those considered in \citet{wuthrich2020omitted}.} For comparison, the OLS estimator is 0.067 for all sample sizes with a standard error between 0.001 and 0.0004. 

Compared to the TSLS estimator, the CIV estimator with $K=2$ is slightly less biased towards OLS (0.067) for smaller sample sizes while maintaining roughly equal deviations from its corresponding full-sample estimate, suggesting that $K=2$ provides some well-suited regularization to the first stage. For CIV with $K=3$ and $K=4$, estimates are slightly closer to the TSLS estimate further suggesting that the first stage model in \eqref{eq:equation_AK} along with $K=2$ may indeed be a good approximation of the optimal instrument structure. In contrast to TSLS and the CIV estimators, the jackknife estimators and the LIML estimator are highly variable for small and moderate sample sizes. Indeed, the median absolute deviations of these estimators are substantially larger compared to CIV. Further, the jackknife-based estimators do not compute in 17\% of the replications when $n=32,950$ because of instrument categories with just a single observation. Finally, the two lasso-based estimators suffer from the differing qualitative conclusions for all sample sizes: While Lasso-IV (1) is highly variable in particular for small sample sizes, Lasso-IV (2) is very stable, statistically significant, and large. In addition, at smaller sample sizes, Lasso-IV (1) occasionally does not select any instruments in the first stage, resulting in non-defined estimates. In particular, Lasso-IV (1) does not select any instruments in 58\% of replications when $n=32,950$, in 20\% of replications when $n=98,852$, and in 11\% of replications when $n=131,803$.

\clearpage
\begin{table}[!h]\scriptsize
\centering
        \begin{threeparttable}
  \caption{\citet{angrist1991does} Estimation Results}\label{tab:AK_res}
\begin{tabular}{rlccccc}
\toprule
\midrule
      &       & $n=32,950$ & $n=98,852$ & $n=131,803$ & $n=197,705$ & Full Sample \\
            &       & (1) & (2) & (3) & (4) & (5) \\
\midrule
\multicolumn{1}{l}{TSLS} & Mean Coef. & 0.074 & 0.083 & 0.088 & 0.093 & 0.099 \\
      & Mean S.E. & 0.015 & 0.014 & 0.013 & 0.012 & 0.01 \\
      & SD Coef. & 0.016 & 0.012 & 0.011 & 0.009 & - \\
      & MAE   & 0.026 & 0.016 & 0.012 & 0.008 & - \\
      & & & & & & \\
\multicolumn{1}{l}{CIV ($K=2$)} & Mean Coef. & 0.078 & 0.088 & 0.092 & 0.096 & 0.102 \\
      & Mean S.E. & 0.023 & 0.019 & 0.017 & 0.015 & 0.013 \\
      & SD Coef. & 0.025 & 0.017 & 0.016 & 0.012 & - \\
      & MAE   & 0.024 & 0.016 & 0.013 & 0.009 & - \\
            & & & & & & \\
\multicolumn{1}{l}{CIV ($K=3$)} & Mean Coef. & 0.076 & 0.085 & 0.090 & 0.095 & 0.110 \\
      & Mean S.E. & 0.019 & 0.016 & 0.015 & 0.014 & 0.012 \\
      & SD Coef. & 0.020 & 0.015 & 0.014 & 0.011 & - \\
      & MAE   & 0.036 & 0.024 & 0.020 & 0.016 & - \\
            & & & & & & \\
\multicolumn{1}{l}{CIV ($K=4$)} & Mean Coef. & 0.076 & 0.085 & 0.089 & 0.093 & 0.095 \\
      & Mean S.E. & 0.017 & 0.015 & 0.014 & 0.013 & 0.011 \\
      & SD Coef. & 0.018 & 0.013 & 0.013 & 0.010 & - \\
      & MAE   & 0.021 & 0.013 & 0.010 & 0.007 & - \\
            & & & & & & \\
\multicolumn{1}{l}{Lasso-IV (1)} & Mean Coef. & 0.087 & 0.092 & 0.088 & 0.092 & 0.091 \\
      & Mean S.E. & 0.423 & 0.308 & 0.211 & 0.121 & 0.061 \\
      & SD Coef. & 0.092 & 0.057 & 0.044 & 0.026 & - \\
      & MAE   & 0.054 & 0.026 & 0.021 & 0.013 & - \\
            & & & & & & \\
\multicolumn{1}{l}{Lasso-IV (2)} & Mean Coef. & 0.131 & 0.137 & 0.137 & 0.137 & 0.134 \\
      & Mean S.E. & 0.008 & 0.004 & 0.003 & 0.002 & 0.002 \\
      & SD Coef. & 0.008 & 0.004 & 0.002 & 0.002 & - \\
      & MAE   & 0.007 & 0.002 & 0.002 & 0.001 & - \\
            & & & & & & \\
\multicolumn{1}{l}{JIVE} & Mean Coef. & -0.008 & 0.325 & 0.056 & 0.156 & 0.134 \\
      & Mean S.E. & -0.431 & 29.912 & -8.673 & 0.044 & 0.022 \\
      & SD Coef. & 0.544 & 3.504 & 1.471 & 0.036 & - \\
      & MAE   & 0.140 & 0.163 & 0.068 & 0.023 & - \\
            & & & & & & \\
\multicolumn{1}{l}{IJIVE} & Mean Coef. & -0.185 & 0.103 & 0.122 & 0.120 & 0.119 \\
      & Mean S.E. & -3.392 & 0.043 & 0.036 & 0.024 & 0.017 \\
      & SD Coef. & 1.655 & 0.218 & 0.037 & 0.018 & - \\
      & MAE   & 0.140 & 0.029 & 0.021 & 0.012 & - \\
            & & & & & & \\
\multicolumn{1}{l}{UJIVE} & Mean Coef. & 0.288 & 0.143 & 0.125 & 0.120 & 0.119 \\
      & Mean S.E. & 0.485 & 0.042 & 0.036 & 0.024 & 0.017 \\
      & SD Coef. & 0.553 & 0.102 & 0.044 & 0.018 & - \\
      & MAE   & 0.191 & 0.028 & 0.020 & 0.013 & - \\
            & & & & & & \\
\multicolumn{1}{l}{LIML} & Mean Coef. & 0.042 & 0.119 & 0.119 & 0.116 & 0.115 \\
      & Mean S.E. & 0.657 & 0.025 & 0.021 & 0.016 & 0.012 \\
      & SD Coef. & 0.846 & 0.044 & 0.028 & 0.017 & - \\
      & MAE   & 0.069 & 0.024 & 0.020 & 0.011 & - \\
\midrule
\bottomrule
\end{tabular}
  \begin{tablenotes}[para,flushleft]
  \scriptsize
  \item \textit{Notes.} Subsampling estimation results are based on 250 replications. ``Mean Coef.'' and ``Mean S.E.'' denote the mean coefficient and standard error estimate across subsampling replications. ``MAE'' denotes the median absolute difference to the full sample coefficient estimate. ``TSLS'' and ``LIML'' denote two-stage least squares, and limited information maximum likelihood, respectively. ``CIV ($K=2$)'', ``CIV ($K=3$)'' , ``CIV ($K=4$)''  denotes the proposed categorical IV estimators restricted to 2, 3, and 4 support points in the first stage. ``JIVE'', ``IJIVE'', and ``UJIVE'' denote the jackknife IV estimators of \citet{angrist1999jackknife}, \citet{ackerberg2009improved}, and \citet{kolesar2013estimation}, respectively. ``Lasso-IV (1)'' and ``Lasso-IV (2)'' denote post-lasso IV estimators proposed by \citet{belloni2012sparse} using two indicator constructions: ``Lasso-IV (1)'' uses the set of indicators for QOB and the interaction terms QOB$\times$POB, ``Lasso-IV (2)'' uses the set of indicators for the fully interacted set QOB$\times$POB. Lasso-IV (1) does not select any instruments in 58\% of replications when $n=32,950$, in 20\% of replications when $n=98,852$, and in 11\% of replications when $n=131,803$. Similarly, because of instruments with a single observation, the jackknife estimators do not compute in   17\% of replication when $n=32,950$, and in 1\% replications when $n=98,852$. The corresponding point estimates are omitted from computation of the summary statistics of Lasso-IV (1), JIVE, IJIVE, and UJIVE. For comparision: The OLS coefficient is 0.067 for all sample sizes, wiht a standard error between 0.001 and 0.0004.
  \end{tablenotes}
    \end{threeparttable}
\end{table}

\pagenumbering{gobble}

\end{document}